\documentclass[preprint,12pt]{elsarticle}

\usepackage{amssymb}
\usepackage{amsthm}

\makeatletter
\def\ps@pprintTitle{%
  \let\@oddhead\@empty
  \let\@evenhead\@empty
  \def\@oddfoot{\reset@font\hfil\thepage\hfil}
  \let\@evenfoot\@oddfoot
}
\makeatother

\usepackage{graphicx}
\usepackage{graphbox}
\usepackage[utf8]{inputenc}
\usepackage{amsmath}
\usepackage[labelformat=simple]{subcaption}
\usepackage[top=2.5cm,bottom=2.5cm,left=2cm,right=2cm]{geometry}
\usepackage[english]{babel}

\usepackage[toc,page]{appendix} 
\usepackage{tabularx} 
\usepackage[table,xcdraw]{xcolor}

\usepackage{pgfplots}
\pgfplotsset{compat=newest}
\usepgfplotslibrary{groupplots}

\definecolor{myblue}{rgb}{0.050980,0.34118,0.69804}
\definecolor{myred}{rgb}{1,0.21961,0.28627}
\definecolor{mygrey}{rgb}{0.43137,0.43137,0.43137}

\usepackage{enumerate}
\usepackage[mathscr]{euscript}
\usepackage[T1]{fontenc}
\usepackage{caption}
\usepackage{multirow}

\usepackage{color}
\usepackage{amsfonts}
\usepackage{amstext}
\usepackage{alltt}
\usepackage{verbatim}
\usepackage{array}
\usepackage{booktabs}
\usepackage[colorlinks=true,linkcolor=blue,citecolor=blue,urlcolor=blue]{hyperref}
\usepackage[linesnumbered,ruled,resetcount]{algorithm2e}
\usepackage{nicefrac}
\usepackage{mathtools}
\usepackage{pdfpages}
\usepackage{bbm}
\usepackage{chngcntr}  
\usepackage{etoolbox}
\usepackage{empheq}
\usepackage[inline]{enumitem}
\usepackage{float}
\usepackage{cleveref}
\usepackage{lscape}
\usepackage{thmtools}
\usepackage{thm-restate}
\usepackage{bbm}

\makeatletter
\let\old@algocf@pre@ruled\@algocf@pre@ruled
\renewcommand{\@algocf@pre@ruled}{%
  \Hy@raisedlink{\hyper@anchorstart{algocf.\thealgocf}\hyper@anchorend}%
  \old@algocf@pre@ruled}
\makeatother

\newcommand{\ST}{\mbox{such that}}
\newcommand{\WITH}{\mbox{with}}
\newcommand{\WHEN}{\mbox{when}}
\newcommand{\WHERE}{\mbox{where}}
\newcommand{\IdEst}{\mbox{i.e.,}}
\newcommand{\IF}{\mbox{if}}

\newcommand{\OSapproach}{\mbox{OS}}
\newcommand{\OSCapproach}{\mbox{OSC}}
\newcommand{\OSCmin}{\mbox{GOS}}

\newcommand{\Dimensional}[1]{\tilde{#1}}
\newcommand{\UNITS}[1]{[#1]}
\newcommand{\NUM}[1]{N_{\mathsf{#1}}}
\newcommand{\VEC}[1]{\vec{#1}}
\newcommand{\MATR}[1]{\textbf{#1}}
\newcommand{\MATRmn}[3]{\MATR{#1}_{#2\times#3}}
\newcommand{\indicator}[2]{\chi_{#1}(#2)}
\newcommand{\NOUGHT}[1]{#1_0} 
\newcommand{\unitsfont}[1]{\mathrm{#1}} 

\newcommand{\eps}{\varepsilon}
\newcommand{\q}{q}
\newcommand{\anum}{a}
\newcommand{\bnum}{b}

\newcommand{\xPP}{y}
\newcommand{\tPP}{t}
\newcommand{\xPPdl}{\varphi}
\newcommand{\tPPdl}{\tau}
\newcommand{\gPP}{g}
\newcommand{\rPP}{R}
\newcommand{\hPP}{h}
\newcommand{\vPP}{v}

\newcommand{\meter}{\unitsfont{m}}
\newcommand{\second}{\unitsfont{s}}
\newcommand{\Gmatr}{\MATR{G}}
\newcommand{\Pmatr}{\MATR{P}}

\newcommand{\xPPref}{\xPPdl_{\mathrm{ref}}}
\newcommand{\odeFF}{\mathrm{\it{ode45}}}
\newcommand{\Ndigit}{n}
\newcommand{\xPPNdig}{\xPPdl_{\Ndigit}}
\newcommand{\TmaxPP}{T}
\newcommand{\errRoundPP}{e_{\Ndigit}}
\newcommand{\errReducePP}{\epsilon_{\mathtt{s}}(\TmaxPP)}
\newcommand{\TmaxRedPP}{\TmaxPP_\mathtt{m}}
\newcommand{\TsingRedPP}{\TmaxPP_0}
\newcommand{\GammaLam}{\gamma}
\newcommand{\DeltaLam}{\delta}

\newcommand{\F}{F}
\newcommand{\var}{x}
\newcommand{\p}{p}
\newcommand{\thcc}{\theta}

\newcommand{\f}{f}
\newcommand{\varDL}{\xi}
\newcommand{\lam}{\lambda}
\newcommand{\NUMlam}{\NUM{l}}

\newcommand{\Klam}{\kappa}
\newcommand{\DD}{d}
\newcommand{\CC}{c}
\newcommand{\TT}{t}
\newcommand{\Rthcc}{\rho}

\newcommand{\be}{\beta}
\newcommand{\unit}{\unitsfont{u}}
\newcommand{\NUMunit}{\NUM{u}}
\newcommand{\bp}{m}
\newcommand{\MM}{\MATR{M}}
\newcommand{\PI}{\pi}
\newcommand{\NUMPI}{\NUM{0}}
\newcommand{\zz}{z}

\newcommand{\Ker}{\mathrm{ker}}
\newcommand{\Dim}{\mathrm{dim}}
\newcommand{\Rank}{\mathrm{rank}}

\newcommand{\Colsp}{\mathrm{colsp}}
\newcommand{\Transp}{T}
\newcommand{\Ortho}{\bot}
\newcommand{\Inv}{-1}
\newcommand{\PseudoInv}{+}
\newcommand{\argmin}{\mathop{\mathrm{argmin}}}
\newcommand{\Max}{\mathop{\mathrm{max}}}
\newcommand{\Min}{\mathop{\mathrm{min}}}
\newcommand{\grad}{\nabla}
\newcommand{\Partial}{\partial}

\newcommand{\costOS}{C}
\newcommand{\Thcc}{\Theta}
\newcommand{\thOS}{\VEC{\thcc}_{\mathrm{opt}}}
\newcommand{\thOScomp}[1]{\VEC{\thcc}_{\mathrm{opt},#1}}
\newcommand{\thOSC}{\VEC{\thcc}^{\star}_{\mathrm{opt}}}
\newcommand{\thOSCcomp}[1]{\VEC{\thcc}^{\star}_{\mathrm{opt},#1}}
 
\newcommand{\UNITscale}{U}
\newcommand{\SMALLscaleIDX}{\mathtt{s}}

\newcommand{\matrixYY}{\MATR{Y}}
\newcommand{\vectorRHSyy}{\VEC{\gamma}}

\newcommand{\BB}{b}
\newcommand{\bbexp}{a}
\newcommand{\BBmatrix}{\MATR{B}}
\newcommand{\VV}{v}
\newcommand{\VVmatrix}{\MATR{V}}
\newcommand{\WW}{w}
\newcommand{\WWmatrix}{\MATR{W}}

\newcommand{\LL}{l}

\newcommand{\aamatrix}{\MATR{A}}
\newcommand{\aavec}{\VEC{\omega}}
\newcommand{\matrixTT}{\MATR{T}}
\newcommand{\vectorRHStt}{\VEC{\tau}}

\newcommand{\AAcoeff}{\alpha}
\newcommand{\AAcoeffvecOS}{\VEC{\AAcoeff}_{\mathrm{opt}}}

\newcommand{\matrixSysAA}{\MATR{S}}
\newcommand{\lhsSysAA}{\VEC{s}}
\newcommand{\NUMrowsAA}{\NUM{r}}
\newcommand{\NUMcolsAA}{\NUM{c}}

\newcommand{\helpMatSysAA}{\boldsymbol{\Sigma}}
\newcommand{\helpRhsSysAA}{\VEC{\sigma}}
\newcommand{\hVSysAAcomp}{\nu}
\newcommand{\helpVecSysAA}{\VEC{\hVSysAAcomp}}
\newcommand{\helpNrSysAA}{\hat{\NUMrowsAA}}

\newif\ifproofread

\newcommand{\changegreen}[1]{%
\ifproofread
{\color{green} #1}%
\else
#1%
\fi
}

\newif\ifrevision

\newcommand{\changered}[1]{%
\ifrevision
{\color{red} #1}
\else
#1%
\fi
}

\crefname{definition}{Definition}{Definitions}
\newtheorem{assumption}{Assumption}
\crefname{assumption}{Assumption}{Assumptions}

\crefname{theorem}{Theorem}{Theorems}

\crefname{proposition}{Proposition}{Propositions}
\newcommand{\figuresname}{Figures~}

\begin{document}

\proofreadfalse 

\revisionfalse

\renewcommand{\sectionautorefname}{Section}
\renewcommand{\subsectionautorefname}{Section}
\renewcommand{\subsubsectionautorefname}{Section}
\renewcommand{\algorithmautorefname}{Algorithm}
\renewcommand{\itemautorefname}{Item}
\setcounter{MaxMatrixCols}{50}	

\begin{frontmatter}



\title{Complexity reduction of physical models:\\an equation-free approach by means of scaling}

\author[cunef]{Simone Rusconi\corref{cor1}\fnref{sr}}
\author[imdea]{Christina Schenk\fnref{cs}} 
\author[asachi]{Razvan Ceuca\fnref{rc}}
\author[bcam,iker,stoilow]{Arghir Zarnescu\fnref{az}}
\author[bcam,iker]{Elena Akhmatskaya\corref{cor2}\fnref{ea}}

\address[cunef]{CUNEF Universidad, C/ de los Pirineos 55, 28040 Madrid, Spain}
\address[imdea]{IMDEA Materials Institute, C/ Eric Kandel 2, Tecnogetafe, 28906 Getafe (Madrid), Spain}
\address[asachi]{``Gh. Asachi'' Technical University, Bd. Carol I, nr. 11, 700506 Iasi, Romania}
\address[bcam]{BCAM - Basque Center for Applied Mathematics, Alameda de Mazarredo 14, 48009 Bilbao, Spain}
\address[iker]{IKERBASQUE, Basque Foundation for Science, Plaza Euskadi 5, 48009 Bilbao, Spain}
\address[stoilow]{``Simion Stoilow'' Institute of the Romanian Academy, 21 Calea Grivi\c{t}ei, 010702 Bucharest, Romania}

\cortext[cor1]{Contact address: \texttt{rusconis89@gmail.com}}
\cortext[cor2]{Contact address: \texttt{akhmatskaya@bcamath.org}}

\fntext[sr]{\href{https://orcid.org/0000-0003-3324-7395}{orcid.org/0000-0003-3324-7395}}
\fntext[cs]{\href{https://orcid.org/0000-0002-7817-6757}{orcid.org/0000-0002-7817-6757}} 
\fntext[rc]{\href{https://orcid.org/0000-0003-0899-2633}{orcid.org/0000-0003-0899-2633}} 
\fntext[az]{\href{https://orcid.org/0000-0002-3620-6196}{orcid.org/0000-0002-3620-6196}} 
\fntext[ea]{\href{https://orcid.org/0000-0002-5136-7991}{orcid.org/0000-0002-5136-7991}} 

\begin{abstract}
 
	The description of complex physical phenomena often involves sophisticated models that rely on a large number of parameters, with many dimensions and scales. One practical \changegreen{way to simplify those kinds of models is to} discard some of the parameters, or  terms of underlying equations, thus giving rise to reduced models. Here, we propose a general approach to obtaining such reduced models. The method is independent of the model in use, i.e., equation-free, depends only on the interplay between the scales and dimensions involved in the description of the phenomena, and controls over-parametrization. It also quantifies conditions for asymptotic models by providing explicitly computable thresholds on values of parameters that allow for reducing complexity of a model, while preserving essential predictive properties. Although our focus is on complexity reduction, this approach may also help with calibration by mitigating the risks of over-parameterization and instability in parameter estimation. The benefits of this approach are discussed in the context of the classical projectile model. 		
	
\end{abstract}




\end{frontmatter}






\section{Introduction}
\label{sec:intro}

	When dealing with complex models, it is often useful to consider an asymptotic analysis with the aim of simplifying the model under consideration. Such an approach identifies a parameter $\eps>0$ and investigates the limit for $\eps \to 0$ of the solutions to the given equations. When $\eps$ is small, the solutions to the model considered in first place can be accurately reproduced by the solutions of the same equations with $\eps$ formally set to $0$, i.e., with some terms neglected. However, in general, it is hard to understand how small such an $\eps$ must be to consider it ``small enough'' and, then, to allow the use of solutions to the simplified model.		
		
	\changered{In particular, for models involving very different scales or over-parametrization, reducing complexity and applying appropriate scaling becomes crucial for preventing instability in calibration \cite{https://doi.org/10.1111/1467-9868.00294,https://doi.org/10.1038/s41467-021-26107-z,ZHANG2024109854,Schenk2026_Calibration}.} \changegreen{This turns out to be} even more critical when feeding synthetic data from physical models into machine learning models \cite{technologies9030052}, as imbalances or miscalibrations in the data can significantly affect model performance \cite{5128907}. Proper scaling ensures that synthetic data reflects realistic parameter distributions, helping to mitigate overfitting and instability in parameter estimation \cite{https://doi.org/10.1111/1467-9868.00294}. Ultimately, scaling improves calibration, stabilizes predictions and may lead to more reliable machine learning model training, even when using synthetic data.
	
	\changered{Several approaches exist for the non-dimensionalization and scaling of physical models. Classical dimensional analysis uses the Buckingham $\pi$ Theorem to derive independent dimensionless groups \cite{PhysRev.4.345}, while characteristic-scale normalization introduces reference scales for the governing variables to obtain dimensionless equations and identify the dominant physical parameters \cite{Scaling_Langtangen}. More recently, automated symbolic approaches have integrated dimensional consistency checking with equation normalization, enabling automatic scaling, identification of dimensionless groups, and improved numerical conditioning of Partial Differential Equation (PDE) formulations \cite{habera2026}.}
				
	Our approach is based on already introduced methodologies, that we called Optimal Scaling ($\OSapproach$) \cite{RUSCONI2019106944} and Optimal Scaling with Constraints ($\OSCapproach$) \cite{RUSCONI2023127756}. Similarly to traditional nondimensionalization approaches, those allow recasting dimensional equations into unitless form. However, in contrast to traditional techniques, they do not require a problem-dependent insight into the phenomenon under consideration. In addition, they ensure a strong reduction of ranges of magnitudes assumed by the equations' coefficients, provide accurate estimation of characteristic features of considered systems and allow reducing round-off and numerical errors.

	Such approaches were applied in \cite{RUSCONI2019106944,RUSCONI2023127756} to the prospective Population Balance Model \cite{Ramkrishna2000} for Latex Particles Morphology Formation (LPMF PBM) \cite{DDPM_2016,PhDThesis_Rusconi_PMCQS}. First, in \cite{RUSCONI2019106944}, we obtained dimensionless scaled equations with minimized variation of the relevant quantities. Such a feature was beneficial in diminishing numerical errors, and avoiding unphysical behavior of the computed solution. Second, in \cite{RUSCONI2023127756}, a quantitative criterion for locating regions of slow and fast aggregation was introduced to derive a family of dimensionless LPMF PBM of reduced complexity. Moreover, in \cite{RUSCONI2019106944}, the Optimal Scaling was also applied to liquid crystal and quantum models, helping to decrease the range of magnitudes assumed by the involved quantities, without compromising \changegreen{the} accuracy of characteristic features of considered systems \cite{Gartland2018,PhysRevA.53.2135}.

	The purpose of this study is to generalize the already introduced methodologies and demonstrate in a transparent way the use and advantages of such an approach \changegreen{on a classical benchmark in scaling theory}. In particular, we present an advanced variant of the Optimal Scaling approaches which inherits benefits of its predecessors and, in addition, \changegreen{controls over-parameterization}. We refer to this method as Generalized Optimal Scaling ($\OSCmin$). The proposed methodology can provide a quantitative criterion for identifying those values of physical parameters that allow discarding a few terms \changegreen{from} the considered equations without significantly modifying their solutions. This technique does not depend on the model, and as such its application to any specific model will produce results which exhibit a certain dependency on the model, that needs to be investigated numerically. In order to concretely illustrate our derivations and findings, we consider the classical projectile model.
	
	The paper is structured as follows. \autoref{sec:prerequisites} clarifies the employed notations and summarizes the projectile model we use in this work as a case study. \autoref{sec:OS_overview} outlines assumptions and statements formulated and proven in this study to support the presented scaling approaches. Moreover, \autoref{sec:OS_overview} sums up the main achievements of the proposed methodologies \cite{RUSCONI2019106944,RUSCONI2023127756} and puts the $\OSCapproach$ method, previously formulated in the context of the LPMF PBM, into a general framework. \autoref{sec:achieve_OS_MinParam} \changegreen{further improves} the generalized $\OSCapproach$ by extending its functionalities with \changegreen{the} ability to reduce to its minimum the number of parameters employed in a model. A detailed formulation of the resulting $\OSCmin$ method is presented in \autoref{tab:summary_OS} in comparison with the original $\OSCapproach$ method. In \autoref{sec:intro_ProjPb}, we test the Optimal Scaling methodologies on the projectile model. \autoref{sec:concl_discuss} provides conclusions and discusses the capabilities of presented techniques.
		
\section{Prerequisites}
\label{sec:prerequisites}		
				
\subsection{Employed Notation}
\label{sec:notation}

	Through all this work, the following notation is employed.

\begin{itemize}

\item We denote any scalar quantity $\Dimensional{\q}$ as
		
\begin{equation}
\Dimensional{\q} \coloneqq \q \, \UNITS{\Dimensional{\q}},
\label{eqn:notation_dimensional_quant}
\end{equation}
	
\noindent being $\UNITS{\Dimensional{\q}}$ the dimensions of $\Dimensional{\q}$, i.e., its units of measure, and $\q$ the numerical value associated with $\Dimensional{\q}$ in the chosen units $\UNITS{\Dimensional{\q}}$.

\item  \changered{We call coefficient $\lam_i$, with $i \in \mathbb{N}$, a multiplicative constant factor in the terms of dimensionless equations.}

\item The symbol $\ll$ is used when a strictly positive number $\anum$ is at least an order of magnitude smaller than another strictly positive number $\bnum$, i.e., $\anum \ll \bnum \Leftrightarrow \anum / \bnum \le 10^{-1} $.

\item The symbol $\approx$ is used when the two strictly positive numbers $\anum$ and $\bnum$ have similar magnitudes, i.e., $\anum \approx \bnum \Leftrightarrow 10^{-1} < \anum / \bnum < 10$.



\item The superscript $\Transp$ \changegreen{denotes the} ``transpose'', e.g., the matrix $\MM^{\Transp}$ is the transpose of $\MM$.

\item The acronyms $\OSapproach$, $\OSCapproach$ and $\OSCmin$ refer to Optimal Scaling, Optimal Scaling with Constraints and Generalized Optimal Scaling, respectively.

\end{itemize}

\subsection{The classical projectile model: scaling strategies}	
\label{sec:CaseStudy_ProjPb}

	We choose the classical projectile model as an application for the scaling approaches considered in this work. First, we wish to discuss some possible strategies for scaling this model in the context of the Optimal Scaling to highlight the ideas behind the scaling methods under study. The rigorous scaling of the model is presented in detail in \autoref{sec:intro_ProjPb}.

	A ball is thrown vertically upward from a certain height above the surface of the Earth. The height $\Dimensional{\xPP}(\Dimensional{\tPP})$ reached at time $\Dimensional{\tPP}$ is measured in meters, while the elapsed time $\Dimensional{\tPP}$ in seconds, i.e., $\UNITS{\Dimensional{\xPP}}=\meter$, $\UNITS{\Dimensional{\tPP}}=\second$. As presented in \cite{Holmes2009_ND}, the variable $\Dimensional{\xPP}$ satisfies the Ordinary Differential Equation 
	
\begin{equation}
\frac{d^2\Dimensional{\xPP}}{d\Dimensional{\tPP}^2} =
- \frac{\Dimensional{\gPP} \, \Dimensional{\rPP}^2}
{(\Dimensional{\xPP}+\Dimensional{\rPP})^2},
\quad
\forall \Dimensional{\tPP} \ge \NOUGHT{\Dimensional{\tPP}} \coloneqq 0 \, \second,
\quad \WITH \quad
\Dimensional{\xPP}(\NOUGHT{\Dimensional{\tPP}}) = 
\NOUGHT{\Dimensional{\hPP}},
\quad
\frac{d\Dimensional{\xPP}}{d\Dimensional{\tPP}}(\NOUGHT{\Dimensional{\tPP}}) =
\NOUGHT{\Dimensional{\vPP}},
\label{eqn:dimensional_ODE_projectile}
\end{equation}
	
\noindent being $\Dimensional{\gPP}$ the gravitational acceleration, $\Dimensional{\rPP}$ the radius of the Earth, $\NOUGHT{\Dimensional{\hPP}}$ and $\NOUGHT{\Dimensional{\vPP}}$ the initial height and velocity, respectively. As indicated in \cite{Scaling_Langtangen}, the dimensionless counterpart of \eqref{eqn:dimensional_ODE_projectile} is obtained by the change of variables
	
\begin{equation}
\tPPdl \coloneqq \Dimensional{\tPP} / \Dimensional{\thcc}_1,
\quad
\xPPdl \coloneqq \Dimensional{\xPP} / \Dimensional{\thcc}_2,
\quad 
\UNITS{\Dimensional{\thcc}_1} = \UNITS{\Dimensional{\tPP}},
\quad
\UNITS{\Dimensional{\thcc}_2} = \UNITS{\Dimensional{\xPP}},
\quad
\thcc_1,\thcc_2 \in (0,\infty),
\label{eqn:change_of_var_projectile}
\end{equation}	  

\noindent where $\Dimensional{\thcc}_1$ and $\Dimensional{\thcc}_2$ are characteristic constants of the considered model \eqref{eqn:dimensional_ODE_projectile}. The transformation \eqref{eqn:change_of_var_projectile} leads to
	
\begin{equation}
\frac{d^2\xPPdl}{d\tPPdl^2} = - \frac{\lam_1}{(1 + \lam_2 \, \xPPdl)^2},
\quad
\forall \tPPdl \ge 0,
\quad \WITH \quad
\xPPdl(0) = \lam_3,
\quad
\frac{d\xPPdl}{d\tPPdl}(0) = \lam_4.
\label{eqn:dimensioless_ODE_projectile}
\end{equation}
	
\noindent Here, the dimensionless coefficients $\lam_1,\lam_2,\lam_3,\lam_4 \in (0,\infty)$ are defined as
	
\begin{equation}
\lam_1 \coloneqq \Dimensional{\p}_1 \, \Dimensional{\thcc}_1^2 \, \Dimensional{\thcc}_2^{-1},
\quad
\lam_2 \coloneqq \Dimensional{\p}_2^{-1} \, \Dimensional{\thcc}_2,
\quad  
\lam_3 \coloneqq \Dimensional{\p}_3 \, \Dimensional{\thcc}_2^{-1},
\quad
\lam_4 \coloneqq \Dimensional{\p}_4 \, \Dimensional{\thcc}_1 \, 
\Dimensional{\thcc}_2^{-1},
\label{eqn:lambdas_def_projectile}
\end{equation}

\begin{equation}
\Dimensional{\p}_1 \coloneqq \Dimensional{\gPP},
\quad
\Dimensional{\p}_2 \coloneqq \Dimensional{\rPP},
\quad  
\Dimensional{\p}_3 \coloneqq \NOUGHT{\Dimensional{\hPP}},
\quad
\Dimensional{\p}_4 \coloneqq \NOUGHT{\Dimensional{\vPP}}. 
\label{eqn:PhysParam_def_projectile}
\end{equation}
	
\noindent The scaling \eqref{eqn:change_of_var_projectile} results in $\NUMlam=4$ dimensionless coefficients $\lam_1,\lam_2,\lam_3,\lam_4$ \eqref{eqn:lambdas_def_projectile}, computed as functions of $\NUM{\var}=2$ characteristic constants $\Dimensional{\thcc}_1,\Dimensional{\thcc}_2$ and $\NUM{\p}=4$ predefined physical parameters $\Dimensional{\p}_1,\Dimensional{\p}_2,\Dimensional{\p}_3,\Dimensional{\p}_4$ \eqref{eqn:PhysParam_def_projectile}. One can now optimize some of dimensionless coefficients, i.e., $\lam_1,\lam_2,\lam_3,\lam_4$, as functions of $\Dimensional{\thcc}_1,\Dimensional{\thcc}_2$, which is the path followed by $\OSapproach$ and $\OSCapproach$.

	Alternatively, one can reduce the number of physical parameters to just two, namely 

\begin{equation}
\PI_1 = \Dimensional{\gPP}^{-1/2} \, \Dimensional{\rPP}^{-1/2} \, \NOUGHT{\Dimensional{\vPP}},
\quad
\PI_2 = \Dimensional{\rPP}^{-1} \, \NOUGHT{\Dimensional{\hPP}}.
\end{equation}

\noindent Now it is possible to  optimize, say, the coefficients $\lam_2,\lam_3,\lam_4$ as functions of $\PI_1,\PI_2$ (instead of $\Dimensional{\thcc}_1,\Dimensional{\thcc}_2$, as suggested in \cite{RUSCONI2019106944,RUSCONI2023127756}). If one aims to have $\lam_2,\lam_3,\lam_4$ as close as possible to $1$ (the approach taken by $\OSapproach$), it follows:

\begin{equation}
\lam_2 = \PI_2^{1/2}, 
\quad
\lam_3 = \PI_2^{1/2},  
\quad
\lam_4 = 1.
\end{equation} 

\noindent Then, we will necessarily have that $\lam_1 = \PI_1^{-2} \, \PI_2^{1/2}$. Further, introducing the condition $\lam_1 \ll 1$ as considered by $\OSCapproach$ (which in concrete numerical computations we take to be $\lam_1 \le 10^{-1}$) provides a necessary relationship between $\PI_1$ and $\PI_2$, a threshold below which the $\lam_1$ can be discarded and results in a reduced model.

	It should be noted that the analysis of the $\lam_i$, $i=1,2,3,4$, as functions of $\PI_1,\PI_2$, is based on just dimensional considerations, and is thus equation-free and model independent. Indeed, one needs afterwards a specific study, based at very least on numerical simulations, in order to identify the precise meaning of the approximation of the solutions of \eqref{eqn:dimensioless_ODE_projectile} by the reduced model in which $\lam_1$ is set to be zero. We return to the projectile model in \autoref{sec:intro_ProjPb} in the context of the novel scaling method to be introduced below.

\section{Optimal Scaling: Overview}
\label{sec:OS_overview}

	We consider models that can be formulated in terms of
	
\begin{equation}
\Dimensional{\F} \left(
\Dimensional{\var}_1, \dots, \Dimensional{\var}_{\NUM{\var}},
\Dimensional{\p}_1, \dots, \Dimensional{\p}_{\NUM{\p}}
\right) = \, 0 \, \UNITS{\Dimensional{\F}},
\label{eqn:dimensional_model}
\end{equation}

\noindent where $\Dimensional{\F}$ is a function (with dimensions $\UNITS{\Dimensional{\F}}$) of unknown variables and independent quantities $\Dimensional{\var}_1, \dots, \Dimensional{\var}_{\NUM{\var}}$ and physical parameters $\Dimensional{\p}_1, \dots, \Dimensional{\p}_{\NUM{\p}}$, whose numerical values

\begin{equation}
\F, \var_1, \dots, \var_{\NUM{\var}} \in \mathbb{C},
\quad
\p_1, \dots, \p_{\NUM{\p}} \in (0,\infty).
\end{equation}

\noindent Aiming to rewrite \eqref{eqn:dimensional_model} in unitless form, we define the dimensionless counterparts of $\Dimensional{\var}_1, \dots, \Dimensional{\var}_{\NUM{\var}}$ as:	
		
\begin{equation}
\varDL_1 \coloneqq \Dimensional{\var}_1 / \Dimensional{\thcc}_1,
\dots, 
\varDL_{\NUM{\var}} \coloneqq
\Dimensional{\var}_{\NUM{\var}} / \Dimensional{\thcc}_{\NUM{\var}},
\quad \WITH \quad
\UNITS{\Dimensional{\thcc}_j} = \UNITS{\Dimensional{\var}_j},
\quad \thcc_j \in (0,\infty),
\quad \forall j=1,\dots,\NUM{\var},
\label{eqn:change_of_var}
\end{equation}
		
\noindent being $\Dimensional{\thcc}_1,\dots,\Dimensional{\thcc}_{\NUM{\var}}$ the characteristic constants of the dimensional model \eqref{eqn:dimensional_model}. By applying \eqref{eqn:change_of_var}, it is possible to rewrite \eqref{eqn:dimensional_model} in unitless form
	
\begin{equation}
\f \left(
\varDL_1, \dots, \varDL_{\NUM{\var}},
\lam_1, \dots, \lam_{\NUMlam}
\right) = 0, 
\quad 
\f: \mathbb{C}^{\NUM{\var}} \times (0,\infty)^{\NUMlam} \to \mathbb{C},
\label{eqn:dimensionless_model}
\end{equation}
	
\noindent with $\lam_1, \dots, \lam_{\NUMlam} \in (0,\infty)$ being the dimensionless coefficients in \eqref{eqn:dimensionless_model}. Such coefficients $\lam_1, \dots, \lam_{\NUMlam}$ depend only on the characteristic constants $\Dimensional{\thcc}_1,\dots,\Dimensional{\thcc}_{\NUM{\var}}$ and physical parameters $\Dimensional{\p}_1, \dots, \Dimensional{\p}_{\NUM{\p}}$, i.e.,

\begin{equation}
\lam_i = 
\lam_i(\Dimensional{\thcc}_1,\dots,\Dimensional{\thcc}_{\NUM{\var}},
\Dimensional{\p}_1, \dots, \Dimensional{\p}_{\NUM{\p}})
\in (0,\infty),
\quad \forall i=1, \dots, \NUMlam.
\label{eqn:lam_coeff_0}
\end{equation}

	Given a fixed choice of physical parameters $\Dimensional{\p}_1, \dots, \Dimensional{\p}_{\NUM{\p}}$, the numerical values $\thcc_1,\dots,\thcc_{\NUM{\var}} \in (0,\infty)$ of the characteristic constants $\Dimensional{\thcc}_1,\dots,\Dimensional{\thcc}_{\NUM{\var}}$ should be determined to complete the definition of the dimensionless model  \eqref{eqn:dimensionless_model}. 
	
	In the literature \cite{Holmes2009_ND,Scaling_Langtangen}, it is often suggested to set $\NUM{\var}$ out of the 
$\NUMlam$ coefficients $\lam_1, \dots, \lam_{\NUMlam}$ equal to $1$. Then, one aims to solve the resulting system of equations with respect to the values $\thcc_1,\dots,\thcc_{\NUM{\var}}$ of the characteristic constants. This approach presents some deficiencies, such as, e.g., at most $\NUM{\var}$ out of $\NUMlam$ independent quantities can be enforced to the value of $1$, with no control on the rest of the dimensionless coefficients. Furthermore, the selection of coefficients enforced to the value of $1$ depends on the particular problem with no general guidelines given so far.
	
	In this paper, we address these issues by proposing techniques for selection of numerical values $\thcc_1,\dots,\thcc_{\NUM{\var}} \in (0,\infty)$ of the characteristic constants $\Dimensional{\thcc}_1,\dots,\Dimensional{\thcc}_{\NUM{\var}}$ to scale the dimensional model \eqref{eqn:dimensional_model} effectively, according to the objectives set for the specific scaling method (see \autoref{tab:summary_OS} from \autoref{sec:achieve_OS_MinParam}). We remark that such procedures do not require a physical insight into the considered phenomenon and corresponding equations. 
		
\subsection{Assumptions and Main Results}	
\label{sec:Prop_OS}

	The scaling approaches presented in this study rely on the \namecrefs{prop:lam_coeff_I} listed below. Such assumptions give rise to the \namecrefs{th:OSC_formula} formulated in \autoref{sec:Th_nondim_setting} to constitute the tools for the nondimensional setting of the dimensional model \eqref{eqn:dimensional_model}. In addition, the \namecrefs{th:OSC_formula} stated in \autoref{sec:Th_opt} indicate how to optimize the coefficients of the resulting dimensionless model \eqref{eqn:dimensionless_model}.
			
\begin{assumption}

	Each of the dimensionless coefficients $\lam_1,\dots,\lam_{\NUMlam}$ \eqref{eqn:lam_coeff_0} can be written as a power-law monomial of $\Dimensional{\thcc}_1,\dots,\Dimensional{\thcc}_{\NUM{\var}}$ multiplied by a factor $\Dimensional{\Klam}_i$:
	
\begin{align}
& \lam_i = \Dimensional{\Klam}_i \,
\Dimensional{\thcc}_1^{\CC_{i,1}} \, 
\Dimensional{\thcc}_2^{\CC_{i,2}} \, 
\cdots \, 
\Dimensional{\thcc}_{\NUM{\var}}^{\CC_{i,\NUM{\var}}},
\quad \Dimensional{\Klam}_i = 
\Dimensional{\Klam}_i(\Dimensional{\p}_1, \dots, \Dimensional{\p}_{\NUM{\p}}),
\quad
\p_1, \dots, \p_{\NUM{\p}} \in (0,\infty),
\nonumber \\
& \Klam_i,\thcc_j \in (0,\infty),
\quad \CC_{i,j} \in \mathbb{R},
\quad \forall i=1,\dots,\NUMlam,
\quad \forall j=1,\dots,\NUM{\var},
\label{eqn:lam_coeff_I}
\end{align}	
	
\noindent where each $\Dimensional{\Klam}_i$ is an explicit function of $\Dimensional{\p}_1, \dots, \Dimensional{\p}_{\NUM{\p}}$ only, and all the exponents $\CC_{i,j}$ in \eqref{eqn:lam_coeff_I} are explicitly known real numbers.
	
\label{prop:lam_coeff_I}
\end{assumption}
	
\begin{assumption}

	Each of the factors $\Dimensional{\Klam}_1,\dots,\Dimensional{\Klam}_{\NUMlam}$ in \eqref{eqn:lam_coeff_I} can be written as a power-law monomial of the physical parameters $\Dimensional{\p}_1, \dots, \Dimensional{\p}_{\NUM{\p}}$ only:

\begin{equation}
\Dimensional{\Klam}_i =
\Dimensional{\p}_1^{\DD_{i,1}} \, 
\Dimensional{\p}_2^{\DD_{i,2}} \, 
\cdots \, 
\Dimensional{\p}_{\NUM{\p}}^{\DD_{i,\NUM{\p}}},
\quad \DD_{i,j} \in \mathbb{R},
\quad \forall i=1,\dots,\NUMlam,
\quad \forall j=1,\dots,\NUM{\p},
\label{eqn:Klam_powlaw_pp}
\end{equation}
	
\noindent with all the exponents $\DD_{i,j}$ in \eqref{eqn:Klam_powlaw_pp} to be explicitly known real numbers.

\label{prop:Klam_powlaw_pp}
\end{assumption}	

\begin{assumption}
	
	Each of the characteristic constants $\Dimensional{\thcc}_1,\dots,\Dimensional{\thcc}_{\NUM{\var}}$ can be written as a power-law monomial of the physical parameters $\Dimensional{\p}_1, \dots, \Dimensional{\p}_{\NUM{\p}}$ only:

\begin{equation}
\Dimensional{\thcc}_j = 
\Dimensional{\p}_1^{\TT_{1,j}} \, 
\Dimensional{\p}_2^{\TT_{2,j}} \, 
\cdots \, 
\Dimensional{\p}_{\NUM{\p}}^{\TT_{\NUM{\p},j}},
\quad \TT_{i,j} \in \mathbb{R},
\quad \forall i=1,\dots,\NUM{\p},
\quad \forall j=1,\dots,\NUM{\var},
\label{eqn:char_const_pow_PhysParam}
\end{equation}
	
\noindent where the exponents $\TT_{i,j}$ in \eqref{eqn:char_const_pow_PhysParam} are unknowns of scaling approaches.
	
\label{prop:char_const}
\end{assumption}	

\begin{assumption}

	The physical parameters $\Dimensional{\p}_1, \dots, \Dimensional{\p}_{\NUM{\p}}$ are independent, i.e., $\Dimensional{\p}_1^{\be_1} \, \Dimensional{\p}_2^{\be_2} \, \cdots \, \Dimensional{\p}_{\NUM{\p}}^{\be_{\NUM{\p}}} = 1 \Leftrightarrow \be_1 = \be_2 = \cdots = \be_{\NUM{\p}} = 0$. 
	
\label{prop:phys_param_indep}	
\end{assumption}	

\begin{assumption}

	The physical parameters $\Dimensional{\p}_1, \dots, \Dimensional{\p}_{\NUM{\p}}$ assume strictly positive numerical values $\p_1, \dots, \p_{\NUM{\p}} \in (0,\infty)$.
			
\label{prop:phys_param_positive}	
\end{assumption}	

\begin{assumption}

	The dimension of each $\Dimensional{\p}_1, \dots, \Dimensional{\p}_{\NUM{\p}}$ is not equal to $1$ and given by a power-law monomial of $\NUMunit$ independent units $\unit_1, \dots, \unit_{\NUMunit}$, i.e., 
	
\begin{equation}
\UNITS{\Dimensional{p}_j} =
\unit_1^{\bp_{1,j}} \, \unit_2^{\bp_{2,j}} \, \cdots \, \unit_{\NUMunit}^{\bp_{\NUMunit,j}} \neq 1, 
\quad \bp_{i,j} \in \mathbb{R},
\quad \forall i=1,\dots,\NUMunit,
\quad \forall j=1,\dots,\NUM{\p},
\label{eqn:phys_param_dimen}
\end{equation}

\noindent with $\unit_1^{\be_1} \, \unit_2^{\be_2} \, \cdots \, \unit_{\NUMunit}^{\be_{\NUMunit}} = 1 \Leftrightarrow \be_1 = \be_2 = \cdots = \be_{\NUMunit} = 0$, and all the exponents $\bp_{i,j}$ in \eqref{eqn:phys_param_dimen} being explicitly known real numbers.	

\label{prop:phys_param_dimen}			
\end{assumption}	

	Aiming to minimize the number of involved quantities, one can also require the independence of the $\NUMlam$ coefficients $\lam_1,\dots,\lam_{\NUMlam}$ \eqref{eqn:lam_coeff_0}, or, equivalently, verify that the rank of the matrix made by the exponents in \eqref{eqn:lam_coeff_I} and \eqref{eqn:Klam_powlaw_pp} is equal to $\NUMlam$, i.e.,

\begin{equation}
\Rank \left(
\begin{pmatrix}
\CC_{1,1} & \cdots & \CC_{1,\NUM{\var}} & \DD_{1,1} & \cdots & \DD_{1,\NUM{\p}} \\ 
\vdots & \ddots & \vdots & \vdots & \ddots & \vdots \\
\CC_{\NUMlam,1} & \cdots & \CC_{\NUMlam,\NUM{\var}} & \DD_{\NUMlam,1} & \cdots & \DD_{\NUMlam,\NUM{\p}}
\end{pmatrix} \right)
= \NUMlam.
\end{equation}

	In \autoref{sec:intro_ProjPb}, we show that \cref{prop:lam_coeff_I,prop:Klam_powlaw_pp,prop:phys_param_indep,prop:phys_param_positive,prop:phys_param_dimen} hold when considering the projectile model presented in \autoref{sec:CaseStudy_ProjPb}. For such a physical model, we also make \cref{prop:char_const}, as explicitly indicated by \eqref{eqn:thcc_ProjPb_pow_PhysParam}. Though such an assumption can be motivated by its dimensional consistency, there exist physical models satisfying \cref{prop:lam_coeff_I,prop:Klam_powlaw_pp,prop:phys_param_indep,prop:phys_param_positive,prop:phys_param_dimen}, while violating \cref{prop:char_const} (see, for example, the equation for energy density (2.49) in \cite{Stewart2004}).
	
\subsubsection{Nondimensional Setting}
\label{sec:Th_nondim_setting}

	\autoref{th:DIMLESS_quant} algebraically formalizes a criterion for dimensionlessness of physical parameters $\Dimensional{\p}_1,\dots,\Dimensional{\p}_{\NUM{\p}}$. The proof is shown in \autoref{sec:proofThII}.

\begin{restatable}[]{theorem}{ThII}
\label{th:DIMLESS_quant}
Given \namecref{prop:phys_param_dimen} \textnormal{\ref{prop:phys_param_dimen}} for physical parameters $\Dimensional{\p}_1,\dots,\Dimensional{\p}_{\NUM{\p}}$, any quantity of the form

\begin{equation}
\Dimensional{\p}_1^{\zz_1} \, 
\Dimensional{\p}_2^{\zz_2} \, 
\cdots \, 
\Dimensional{\p}_{\NUM{\p}}^{\zz_{\NUM{\p}}}, 
\quad \WITH \quad \zz_1,\zz_2,\dots,\zz_{\NUM{\p}}\in\mathbb{R},
\label{eqn:ZZdimless}	
\end{equation}

\noindent is dimensionless if and only if the vector $\VEC{\zz} \coloneqq ( \zz_1, \zz_2, \dots, \zz_{\NUM{\p}} )^{\Transp} \in \mathbb{R}^{\NUM{\p}}$ of exponents in \eqref{eqn:ZZdimless} belongs to the null space (kernel) of the matrix 

	\begin{equation}
	\MM  \coloneqq 	
	\begin{pmatrix}
	\bp_{1,1} & \bp_{1,2} & \cdots & \bp_{1,\NUM{\p}} \\
	\bp_{2,1} & \bp_{2,2} & \cdots & \bp_{2,\NUM{\p}} \\
	\vdots    & \vdots    & \ddots & \vdots      \\
	\bp_{\NUMunit,1} & \bp_{\NUMunit,2} & \cdots & \bp_{\NUMunit,\NUM{\p}}
	\end{pmatrix}
	\in \mathbb{R}^{\NUMunit \times \NUM{\p}},
	\label{eqn:matrixMM}
	\end{equation}

\noindent where the entries of $\MM$ \eqref{eqn:matrixMM} are the exponents in \eqref{eqn:phys_param_dimen} specified by \namecref{prop:phys_param_dimen} \textnormal{\ref{prop:phys_param_dimen}}.
\end{restatable}

	\autoref{th:PI_quantities} determines the minimal number of independent dimensionless parameters and indicates a way to compute them in terms of the provided physical parameters $\Dimensional{\p}_1,\dots,\Dimensional{\p}_{\NUM{\p}}$. The proof is shown in \autoref{sec:proofThIII}.	
	
\begin{restatable}[]{theorem}{ThIII}
\label{th:PI_quantities}
Given \namecrefs{prop:phys_param_indep} \textnormal{\ref{prop:phys_param_indep}} to \textnormal{\ref{prop:phys_param_dimen}} for physical parameters $\Dimensional{\p}_1,\dots,\Dimensional{\p}_{\NUM{\p}}$, there exist only $\NUMPI = \NUM{\p} - \Rank(\MM)$ independent dimensionless parameters $\PI_1,\dots,\PI_{\NUMPI} \in (0,\infty)$ of the form \eqref{eqn:ZZdimless}, with $\MM$ as in \eqref{eqn:matrixMM}. Such parameters $\PI_1,\dots,\PI_{\NUMPI}$ are identified by the vectors composing a basis of the null space (kernel) of the matrix $\MM$ \eqref{eqn:matrixMM}:

	\begin{align}
	& \PI_k = \Dimensional{\p}_1^{\BB_{k,1}} \, \Dimensional{\p}_2^{\BB_{k,2}} \, \cdots \, \Dimensional{\p}_{\NUM{\p}}^{\BB_{k,\NUM{\p}}}, 
	\quad \VEC{\BB}_k \coloneqq ( \BB_{k,1}, \BB_{k,2}, \dots, \BB_{k,\NUM{\p}} )^{\Transp} 
	\in \mathbb{R}^{\NUM{\p}},
	\nonumber \\
	& \WITH \quad \VEC{\BB}_1, \VEC{\BB}_2, \dots, \VEC{\BB}_{\NUMPI}  
	\quad \mbox{\changegreen{constituting} a basis of} \quad \Ker(\MM),
	\label{eqn:Indep_Dimless_Param}
	\end{align}

\noindent and $\BB_{k,j}$ being the $j$-th entry of the vector $\VEC{\BB}_k$ for any $k=1,\dots,\NUMPI$ and $j=1,\dots,\NUM{\p}$.
\end{restatable}

	\autoref{th:lamb_pow_PI} states that the dimensionless coefficients $\lam_1,\dots,\lam_{\NUMlam}$ \eqref{eqn:lam_coeff_0} can be written in terms of the independent dimensionless parameters $\PI_1,\dots,\PI_{\NUMPI}$ \eqref{eqn:Indep_Dimless_Param} only. The proof is shown in \autoref{sec:proofThIV}.
	
\begin{restatable}[]{theorem}{ThIV}
\label{th:lamb_pow_PI}
Given \namecrefs{prop:lam_coeff_I} \textnormal{\ref{prop:lam_coeff_I}} to \textnormal{\ref{prop:phys_param_dimen}} and $\PI_1,\dots,\PI_{\NUMPI}$ \eqref{eqn:Indep_Dimless_Param}, each of the dimensionless coefficients $\lam_1,\dots,\lam_{\NUMlam}$ \eqref{eqn:lam_coeff_0} can be written as a power-law monomial of the independent dimensionless parameters $\PI_1,\dots,\PI_{\NUMPI}$ only, i.e.,
	
\begin{equation}
\lam_i = 
\PI_1^{\AAcoeff_{i,1}} \, \PI_2^{\AAcoeff_{i,2}}
\, \cdots \, \PI_{\NUMPI}^{\AAcoeff_{i,\NUMPI}},
\quad \AAcoeff_{i,k} \in \mathbb{R},
\quad \forall i=1,\dots,\NUMlam,
\quad \forall k=1,\dots,\NUMPI,	
\label{eqn:lambdas_fun_PI}
\end{equation}	
\noindent where the exponents $\AAcoeff_{i,k}$ in \eqref{eqn:lambdas_fun_PI} are unknowns of scaling procedures.
\end{restatable}

	\autoref{th:sysYY} offers a way to obtain the unknown exponents $\TT_{i,j}$ in \eqref{eqn:char_const_pow_PhysParam}. The proof is shown in \autoref{sec:proofThV}.
	
\begin{restatable}[]{theorem}{ThV}
\label{th:sysYY}
Given \namecrefs{prop:lam_coeff_I} \textnormal{\ref{prop:lam_coeff_I}} to \textnormal{\ref{prop:phys_param_dimen}}, the vector

\begin{equation}
\VEC{\TT} \coloneqq  
\begin{pmatrix}
\VEC{\TT}_1 \\ \VEC{\TT}_2 \\ \vdots \\ \VEC{\TT}_{\NUM{\var}}
\end{pmatrix}
\in \mathbb{R}^{\NUM{\p} \, \NUM{\var}},
\quad \WITH \quad
\VEC{\TT}_j  \coloneqq 
\begin{pmatrix}
\TT_{1,j} \\ \TT_{2,j} \\ \vdots \\ \TT_{\NUM{\p},j}
\end{pmatrix}
\in \mathbb{R}^{\NUM{\p}},
\label{eqn:vecYY_def}
\end{equation}	

\noindent of unknown exponents in \eqref{eqn:char_const_pow_PhysParam} must satisfy the linear system

\begin{equation}
\matrixYY \, \VEC{\TT} = \vectorRHSyy,
\quad \matrixYY \in \mathbb{R}^{(\NUMlam \, \NUM{\p}) \times (\NUM{\p} \, \NUM{\var})},
\quad \vectorRHSyy \in \mathbb{R}^{\NUMlam \, \NUM{\p}},
\label{eqn:SyS_TT}
\end{equation}

\noindent where the entries of $\matrixYY$ are explicitly known real numbers, while the entries of $\vectorRHSyy$ are explicit affine functions of the exponents $\AAcoeff_{i,k}$ in \eqref{eqn:lambdas_fun_PI}, as specified by \eqref{eqn:matYY_vecYY_def} in \autoref{sec:proofThV}.
\end{restatable}

	\autoref{th:sysAA} provides a way for finding unknown exponents $\AAcoeff_{i,k}$ in \eqref{eqn:lambdas_fun_PI} (\autoref{th:lamb_pow_PI}). The proof is shown in \autoref{sec:proofThVI}.
	 
\begin{restatable}[]{theorem}{ThVI}
\label{th:sysAA}
Given \namecrefs{prop:lam_coeff_I} \textnormal{\ref{prop:lam_coeff_I}} to \textnormal{\ref{prop:phys_param_dimen}} and the consequent expression \eqref{eqn:lambdas_fun_PI} of coefficients $\lam_1,\dots,\lam_{\NUMlam}$ \eqref{eqn:lam_coeff_0}, the vector

\begin{equation}
\VEC{\AAcoeff}  \coloneqq 
\begin{pmatrix}
\VEC{\AAcoeff}_1 \\ \VEC{\AAcoeff}_2 \\ \vdots \\ \VEC{\AAcoeff}_{\NUMlam}
\end{pmatrix}
\in \mathbb{R}^{\NUMlam \, \NUMPI},
\quad \WITH \quad
\VEC{\AAcoeff}_i  \coloneqq  
\begin{pmatrix}
\AAcoeff_{i,1} \\ \AAcoeff_{i,2} \\ \vdots \\ \AAcoeff_{i,\NUMPI}
\end{pmatrix}
\in \mathbb{R}^{\NUMPI},
\label{eqn:AAdef}
\end{equation}
	
\noindent of exponents in \eqref{eqn:lambdas_fun_PI} must satisfy the linear system

\begin{equation}
\matrixSysAA \, \VEC{\AAcoeff} = \lhsSysAA,
\quad
\matrixSysAA \in \mathbb{R}^{\NUMrowsAA \times \NUMcolsAA},
\quad
\lhsSysAA \in \mathbb{R}^{\NUMrowsAA},
\quad
\NUMcolsAA \coloneqq \NUMlam \, \NUMPI,
\quad
\NUMrowsAA \coloneqq \NUMlam \, \NUM{\p} - \Rank(\matrixYY),
\label{eqn:constr_SysAA}
\end{equation}		

\noindent where the entries of $\matrixSysAA$ and $\lhsSysAA$ are explicitly known real numbers identified by \eqref{eqn:SysAA_MatrCoeffdef}-\eqref{eqn:SysAA_RHSdef} in \autoref{sec:proofThVI}. The matrix $\matrixYY \in \mathbb{R}^{(\NUMlam \, \NUM{\p}) \times (\NUM{\p} \, \NUM{\var})}$ is explicitly known, as it is given in \eqref{eqn:SyS_TT} and specified by \eqref{eqn:matYY_vecYY_def} in \autoref{sec:proofThV}.
\end{restatable}

\subsubsection{Optimization}
\label{sec:Th_opt}

	\autoref{th:OSC_formula} enables the computation of numerical values $\thcc_1,\dots,\thcc_{\NUM{\var}}$ of the characteristic constants $\Dimensional{\thcc}_1,\dots,\Dimensional{\thcc}_{\NUM{\var}}$ that minimize the deviation $\costOS_{\UNITscale}$ from unity of the magnitudes of $\lam_i$ \eqref{eqn:lam_coeff_0}, with $i \in \UNITscale \subseteq \{1,\dots,\NUMlam\}$. The proof is shown in \autoref{sec:proofThI}.	 

\begin{restatable}[]{theorem}{ThI}
\label{th:OSC_formula}
Given \namecref{prop:lam_coeff_I} \textnormal{\ref{prop:lam_coeff_I}} for coefficients $\lam_1,\dots,\lam_{\NUMlam}$ \eqref{eqn:lam_coeff_0}, $\UNITscale \subseteq \{1,\dots,\NUMlam\}$ and any fixed choice of $\Dimensional{\p}_1, \dots, \Dimensional{\p}_{\NUM{\p}}$, the values of $(\thcc_1,\dots,\thcc_{\NUM{\var}}) \in (0,\infty)^{\NUM{\var}}$ attaining the global minimum of the function 

\begin{equation}
\costOS_{\UNITscale}(\thcc_1,\dots,\thcc_{\NUM{\var}}) \coloneqq 
\sum_{i \in \UNITscale} \left( 
\log_{10}( 
\lam_i(\Dimensional{\thcc}_1,\dots,\Dimensional{\thcc}_{\NUM{\var}},
\Dimensional{\p}_1, \dots, \Dimensional{\p}_{\NUM{\p}}) ) 
\right)^2,
\quad \costOS_{\UNITscale} : (0,\infty)^{\NUM{\var}} \to [0,\infty),
\label{eqn:costOS_general}
\end{equation}

\noindent correspond to the solutions of the following linear system of $\NUM{\var}$ equations with $\NUM{\var}$ unknowns $\Rthcc_1,\dots,\Rthcc_{\NUM{\var}} \in \mathbb{R}$:

	\begin{equation}
	\sum_{k=1}^{\NUM{\var}}
	\left( \sum_{i \in \UNITscale} \, \CC_{i,k} \, \CC_{i,j} \right)
	\Rthcc_k
	=
	- \sum_{i \in \UNITscale} \, \CC_{i,j} \, \log_{10}(\Klam_i),
	\quad
	\forall j=1,\dots,\NUM{\var},	
	\label{eqn:lin_sys_rho}
	\end{equation}	
	 
\noindent where $\Rthcc_k  \coloneqq  \log_{10}(\thcc_k)$ and the numerical value $\Klam_i \in (0,\infty)$ of each $\Dimensional{\Klam}_i=\Dimensional{\Klam}_i(\Dimensional{\p}_1, \dots, \Dimensional{\p}_{\NUM{\p}})$ is fixed by the choice of $\Dimensional{\p}_1, \dots, \Dimensional{\p}_{\NUM{\p}}$.
\end{restatable}

	\autoref{th:SysToSolveOSAA} enables the computation of numerical values of the exponents $\AAcoeff_{i,k}$ in \eqref{eqn:lambdas_fun_PI} that minimize the deviation $\costOS_{\UNITscale}$ from unity of the magnitudes of $\lam_i$, with $i \in \UNITscale \subseteq \{1,\dots,\NUMlam\}$. The proof is shown in \autoref{sec:proofThVII}.

\begin{restatable}[]{theorem}{ThVII}
\label{th:SysToSolveOSAA}
Given any fixed choice of $\PI_1,\dots,\PI_{\NUMPI} \in (0,\infty)$ and set $\UNITscale \subseteq \{1,\dots,\NUMlam\}$, the values of $\VEC{\AAcoeff} \in \mathbb{R}^{\NUMcolsAA}$ \eqref{eqn:AAdef}, $\NUMcolsAA \coloneqq \NUMlam \, \NUMPI$, attaining the global minimum of the function 

\begin{equation}
\costOS_{\UNITscale}(\VEC{\AAcoeff}) \coloneqq \sum_{i \in \UNITscale} \left( \log_{10}( \lam_i ) \right)^2,
\quad \lam_i = \PI_1^{\AAcoeff_{i,1}} \, \PI_2^{\AAcoeff_{i,2}} \, \cdots \, \PI_{\NUMPI}^{\AAcoeff_{i,\NUMPI}},
\quad \costOS_{\UNITscale} : \mathbb{R}^{\NUMcolsAA} \to [0,\infty), 
\label{eqn:CostOSC_def}
\end{equation}

\noindent subject to a constraint of the form 

\begin{equation}
\helpMatSysAA \, \VEC{\AAcoeff} = \helpRhsSysAA,
\quad \helpMatSysAA \in \mathbb{R}^{\helpNrSysAA \times \NUMcolsAA},
\quad \helpRhsSysAA \in \mathbb{R}^{\helpNrSysAA},
\quad \helpNrSysAA \in \mathbb{N},
\label{eqn:Constr_SysOSAA}
\end{equation}

\noindent satisfy the system of equations

\begin{equation}
\helpMatSysAA \, \VEC{\AAcoeff} = \helpRhsSysAA, \quad
\grad \costOS_{\UNITscale}(\VEC{\AAcoeff}) \, + \, 
\helpMatSysAA^{\Transp} \, \helpVecSysAA = \VEC{0},
\label{eqn:SysToSolveOSAA}
\end{equation}

\noindent assuming that

\begin{equation}
\Rank(\helpMatSysAA) = \helpNrSysAA < \NUMcolsAA,
\label{eqn:HP_OS_AA_analytic_form}
\end{equation}
	
\noindent and there is a vector $\helpVecSysAA \in \mathbb{R}^{\helpNrSysAA}$ such that \eqref{eqn:SysToSolveOSAA} holds.

\end{restatable}	

	\changered{If the Assumptions in \autoref{sec:Prop_OS} are not satisfied, the cost function \eqref{eqn:costOS_general} can still be minimized numerically to determine optimal scales, but the analytical results and systematic reduction of physical parameters are no longer guaranteed, since they rely on the underlying assumptions, particularly the power-law structure required by dimensional analysis.}
	
\subsection{Optimal Scaling}
\label{sec:achieve_OS}

	Optimal Scaling ($\OSapproach$) proposed in \cite{RUSCONI2019106944} enables such a rational choice of $\thcc_1,\dots,\thcc_{\NUM{\var}} \in (0,\infty)$, that ensures the deviation of magnitudes of $\lam_1,\dots,\lam_{\NUMlam}$ \eqref{eqn:lam_coeff_0} from unity to be minimal for any fixed choice of physical parameters $\Dimensional{\p}_1, \dots, \Dimensional{\p}_{\NUM{\p}}$. This is advantageous for numerical computations as the floating-point numbers have the highest density in the interval $(0,1)$ and their density decreases when moving further away from this interval. So, it is always convenient to manipulate numerically the quantities with magnitude $\approx 1$, in order to avoid severe round-off errors and to improve the conditioning of the problem at hand.			
		
	In \cite{RUSCONI2019106944}, such a methodology was first applied to liquid crystal and quantum models. The reported data demonstrate that $\OSapproach$ ensures a strong reduction of the range of orders of magnitude assumed by involved quantities. Despite such a strong transformation, $\OSapproach$ allows for an accurate estimation of characteristic features of considered systems. As an example, we showed that $\OSapproach$ correctly describes the large-body limit of Landau-de Gennes model for liquid crystals \cite{Gartland2018}. Moreover, for Schrödinger Equation \cite{PhysRevA.53.2135}, we proved that the characteristic constants, provided by $\OSapproach$, follow the same trend as the often used atomic units. 
	
	Further, it was demonstrated that $\OSapproach$ was also beneficial for the case of the Population Balance Model \cite{Ramkrishna2000} for Latex Particles Morphology Formation \cite{DDPM_2016,PhDThesis_Rusconi_PMCQS}. When compared with other scaling approaches \cite{Holmes2009_ND,Scaling_Langtangen}, it assured the smallest variation $\Max_{i=1,\dots,\NUMlam} \lam_i / \Min_{i=1,\dots,\NUMlam} \lam_i$ of coefficients \eqref{eqn:lam_coeff_0} and allowed reducing round-off and numerical errors, as discussed in \cite{RUSCONI2019106944}.

\subsubsection{Formulation}
\label{sec:OptScal}
	
	Given any fixed choice of physical parameters $\Dimensional{\p}_1, \dots, \Dimensional{\p}_{\NUM{\p}}$, $\OSapproach$ proposed in \cite{RUSCONI2019106944} finds $\VEC{\thcc}=\thOS$ in such a way that the deviation of magnitudes of $\lam_1,\dots,\lam_{\NUMlam}$ \eqref{eqn:lam_coeff_0} from unity is minimal, i.e., 
	
\begin{align}
& \thOS \coloneqq
\argmin_{\VEC{\thcc} \in (0,\infty)^{\NUM{\var}}}
\costOS(\VEC{\thcc}),
\quad
\costOS(\VEC{\thcc}) \coloneqq 	
\sum_{i=1}^{\NUMlam}
\left( \log_{10}( \lam_i(\Dimensional{\thcc}_1,\dots,\Dimensional{\thcc}_{\NUM{\var}},
\Dimensional{\p}_1, \dots, \Dimensional{\p}_{\NUM{\p}}) ) - \Thcc_i \right)^2,
\nonumber \\
& \WHERE \quad 
\VEC{\thcc} \coloneqq (\thcc_1,\dots,\thcc_{\NUM{\var}}) \in (0,\infty)^{\NUM{\var}}
\quad \mbox{and the constants} \quad 
\Thcc_1=\cdots=\Thcc_{\NUMlam}=0.
\label{eqn:thOS_def}
\end{align}
		
\noindent In \eqref{eqn:thOS_def}, the order of magnitude of $\lam_i$ is computed as $\log_{10}(\lam_i)$, while $\Thcc_i = 0$ to target the desired order $\log_{10}(1) = 0$, $\forall i=1,\dots,\NUMlam$.

	With $\UNITscale = \{1,\dots,\NUMlam\}$, the function $\costOS$ \eqref{eqn:thOS_def} is equivalent to $\costOS_{\UNITscale}$ in \autoref{th:OSC_formula} (\autoref{sec:Th_opt}). Then, given \cref{prop:lam_coeff_I} (\autoref{sec:Prop_OS}) and any fixed choice of $\Dimensional{\p}_1, \dots, \Dimensional{\p}_{\NUM{\p}}$, the points $\thOS$ of global minimum of $\costOS \equiv \costOS_{\UNITscale}$ can be computed by means of \eqref{eqn:lin_sys_rho}, as follows from \autoref{th:OSC_formula} with $\UNITscale = \{1,\dots,\NUMlam\}$. In this way, $\OSapproach$ provides scaled coefficients $\lam_1,\dots,\lam_{\NUMlam}$ with magnitudes as close as possible to the targeted value $1$. In other words, $\OSapproach$ attempts to make
			
\begin{equation}
\lam_i(\Dimensional{\thcc}_1,\dots,\Dimensional{\thcc}_{\NUM{\var}},
\Dimensional{\p}_1, \dots, \Dimensional{\p}_{\NUM{\p}}) \approx 1,
\quad \WHEN \quad
\Dimensional{\thcc}_j = \thOScomp{j} \, \UNITS{\Dimensional{\thcc}_j},
\quad \forall i,j,
\label{eqn:OS_lam_regime}
\end{equation}	
	
\noindent with $\thOScomp{j}$ being the $j$-th component of the vector $\thOS$ \eqref{eqn:thOS_def}.
	
	By its very design, $\OSapproach$ finds the numerical values of the characteristic constants $\Dimensional{\thcc}_1,\dots,\Dimensional{\thcc}_{\NUM{\var}}$ minimizing the Euclidean norm $\costOS$ \eqref{eqn:thOS_def} of the vector $\left( \log_{10}(\lam_1),\log_{10}(\lam_2),\dots,\log_{10}(\lam_{\NUMlam}) \right)$. As discussed in \cite{RUSCONI2019106944}, the Euclidean norm $\costOS$ \eqref{eqn:thOS_def} is not the unique choice for the function to minimize. By considering different objective functions, we show in \cite{RUSCONI2019106944} that the performance of $\OSapproach$ is not affected by the particular choice of the function in use. This observation can be motivated by the equivalence of norms on the finite dimensional space $(-\infty,\infty)^{\NUMlam}$. Nevertheless, the availability of an explicitly computable solution makes $\costOS$ \eqref{eqn:thOS_def} the metric of choice for Optimal Scaling.
		
\subsection{Optimal Scaling with Constraints}
\label{sec:achieve_OSC}

	Because of the similar values achieved by $\lam_1,\dots,\lam_{\NUMlam}$, the Optimal Scaling ($\OSapproach$) summarized above does not allow examining regimes of $\lam_1,\dots,\lam_{\NUMlam}$ \eqref{eqn:lam_coeff_0} in \eqref{eqn:dimensionless_model} which exhibit a large difference in magnitudes. The Optimal Scaling with Constraints ($\OSCapproach$) proposed in \cite{RUSCONI2023127756} addresses scenarios where (at least) one coefficient is much smaller than the others, i.e., most of $\lam_1,\dots,\lam_{\NUMlam}$ in \eqref{eqn:dimensionless_model} are as close as possible to $1$, while (at least) one coefficient is $\ll 1$. The method quantifies those $\thcc_1,\dots,\thcc_{\NUM{\var}} \in (0,\infty)$ that lead to such a regime.  	
		
	As demonstrated in \cite{RUSCONI2023127756} on the example of the Population Balance Model \cite{Ramkrishna2000} for Latex Particles Morphology Formation \cite{DDPM_2016,PhDThesis_Rusconi_PMCQS}, $\OSCapproach$ can also determine a threshold on the values of physical parameters $\Dimensional{\p}_1,\dots,\Dimensional{\p}_{\NUM{\p}}$ that allow dropping the terms relative to the $\lam_i \ll 1$ \changegreen{from} the considered equations without significantly changing their solutions, but dramatically reducing computational effort.	
		
\subsubsection{Formulation}
\label{sec:OptScalConstr}

	Let us denote $\SMALLscaleIDX \in \{1,\dots,\NUMlam\}$ as the index of the coefficient $\lam_\SMALLscaleIDX$ being much smaller than $1$ and $\UNITscale_\SMALLscaleIDX \coloneqq \{1,\dots,\NUMlam\} \setminus \{\SMALLscaleIDX\}$ as the set of indexes of the coefficients 
		
\changered{\begin{equation}
\lam_i = 
\lam_i(\Dimensional{\thcc}_1,\dots,\Dimensional{\thcc}_{\NUM{\var}},
\Dimensional{\p}_1, \dots, \Dimensional{\p}_{\NUM{\p}})
\end{equation}}
	
\noindent that are as close as possible to $1$. In particular, one aims to have	
		
\begin{equation}
\lam_\SMALLscaleIDX \ll 1,
\quad
\lam_i \approx 1,
\quad
\forall i \in \UNITscale_\SMALLscaleIDX 
\coloneqq \{1,\dots,\NUMlam\} \setminus \{\SMALLscaleIDX\},
\quad \WITH \quad
\SMALLscaleIDX \in \{1,\dots,\NUMlam\}.
\label{eqn:scale_sep_lam}
\end{equation}		

\noindent Given any fixed choice of $\Dimensional{\p}_1, \dots, \Dimensional{\p}_{\NUM{\p}}$ and following the procedure proposed in \cite{RUSCONI2023127756}, it is possible to show that the vector $\VEC{\thcc}=\thOSC$ targeting the regime \eqref{eqn:scale_sep_lam} of $\lam_1,\dots,\lam_{\NUMlam}$ can be found as 

\begin{align}
& \thOSC \coloneqq
\argmin_{\VEC{\thcc} \in (0,\infty)^{\NUM{\var}}}
\costOS_{\UNITscale_\SMALLscaleIDX}(\VEC{\thcc}),
\quad
\costOS_{\UNITscale_\SMALLscaleIDX}(\VEC{\thcc}) \coloneqq 
\sum_{i \in \UNITscale_\SMALLscaleIDX } 
\left( \log_{10}( \lam_i(\Dimensional{\thcc}_1,\dots,\Dimensional{\thcc}_{\NUM{\var}},
\Dimensional{\p}_1, \dots, \Dimensional{\p}_{\NUM{\p}}) ) - \Thcc_i \right)^2,
\nonumber \\
& \WHERE \quad 
\VEC{\thcc} \coloneqq (\thcc_1,\dots,\thcc_{\NUM{\var}}) \in (0,\infty)^{\NUM{\var}},
\quad \Thcc_1=\cdots=\Thcc_{\NUMlam}=0,
\label{eqn:OS_scale_sep}	
\end{align}

\noindent and $\costOS_{\UNITscale_\SMALLscaleIDX}$ accounts only for the coefficients $\lam_i \approx 1$, as specified by $\UNITscale_\SMALLscaleIDX$ \eqref{eqn:scale_sep_lam}. Moreover, we remark that \eqref{eqn:OS_scale_sep} quantifies the order of magnitude of $\lam_i$ as $\log_{10}(\lam_i)$, while $\Thcc_i = 0$ to target the desired order $\log_{10}(1) = 0$, $\forall i=1,\dots,\NUMlam$. By this means, the cost function $\costOS_{\UNITscale_\SMALLscaleIDX}$ \eqref{eqn:OS_scale_sep} measures the distance between the vector holding the orders of magnitude of coefficients $\lam_i$, with $i \in \UNITscale_\SMALLscaleIDX$, and the corresponding vector of desired orders $\Thcc_i=0$, with $i \in \UNITscale_\SMALLscaleIDX$. Made \cref{prop:lam_coeff_I} (\autoref{sec:Prop_OS}), provided $\UNITscale=\UNITscale_\SMALLscaleIDX$ \eqref{eqn:scale_sep_lam} and given any fixed $\Dimensional{\p}_1, \dots, \Dimensional{\p}_{\NUM{\p}}$, the points $\thOSC$ of global minimum of $\costOS_{\UNITscale_\SMALLscaleIDX}$ can be computed by means of \eqref{eqn:lin_sys_rho}, as follows from \autoref{th:OSC_formula} (\autoref{sec:Th_opt}).

	Then, the following features can be attained by the $\OSCapproach$ methodology \cite{RUSCONI2023127756} and the computed $\thOSC$ \eqref{eqn:OS_scale_sep}. For any given $\Dimensional{\p}_1, \dots, \Dimensional{\p}_{\NUM{\p}}$, the coefficients $\lam_i$, with $i \in \UNITscale_\SMALLscaleIDX$ \eqref{eqn:scale_sep_lam}, \changegreen{are} as close as possible to the targeted value $1$, i.e., $\OSCapproach$ attempts to make
	
\begin{equation}
\lam_i(\Dimensional{\thcc}_1,\dots,\Dimensional{\thcc}_{\NUM{\var}},
\Dimensional{\p}_1, \dots, \Dimensional{\p}_{\NUM{\p}}) \approx 1,
\quad \WHEN \quad
\Dimensional{\thcc}_j = \thOSCcomp{j} \, \UNITS{\Dimensional{\thcc}_j},
\quad \forall i \in \UNITscale_\SMALLscaleIDX,
\quad \forall j=1,\dots,\NUM{\var},
\label{eqn:threshold_OSC_unitary_lam}
\end{equation}	
	
\noindent with $\thOSCcomp{j}$ being the $j$-th component of the vector $\thOSC$ \eqref{eqn:OS_scale_sep}. Hence, the scaled model \eqref{eqn:dimensionless_model} accounting only for the terms relative to $\lam_i$, with $i \in \UNITscale_\SMALLscaleIDX$ \eqref{eqn:scale_sep_lam}, inherits the features of the $\OSapproach$ approach \cite{RUSCONI2019106944} discussed in \autoref{sec:achieve_OS}. 

	On the other hand, the inequality 
	
\begin{equation}
\lam_\SMALLscaleIDX
(\Dimensional{\thcc}_1,\dots,\Dimensional{\thcc}_{\NUM{\var}}, 
\Dimensional{\p}_1, \dots, \Dimensional{\p}_{\NUM{\p}}) \ll 1,
\quad \WITH \quad
\Dimensional{\thcc}_j = \thOSCcomp{j} \, \UNITS{\Dimensional{\thcc}_j},
\quad \forall j=1,\dots,\NUM{\var},
\label{eqn:threshold_OSC_ll}
\end{equation}

\noindent can be used together with \eqref{eqn:threshold_OSC_unitary_lam} to identify a threshold for the values of physical parameters $\Dimensional{\p}_1, \dots, \Dimensional{\p}_{\NUM{\p}}$ leading to a large difference in magnitudes between $\lam_\SMALLscaleIDX$ and all the remaining $\lam_i$, with $i \in \UNITscale_\SMALLscaleIDX$ \eqref{eqn:scale_sep_lam}.

 	Finally, we remark that $\OSapproach$ presented in \autoref{sec:achieve_OS} can be viewed as a special case of $\OSCapproach$ discussed here. Indeed, the optimization problem \eqref{eqn:thOS_def} is equivalent to \eqref{eqn:OS_scale_sep} with $\SMALLscaleIDX=0$ and $\UNITscale_0$ defined as in \eqref{eqn:scale_sep_lam}, i.e., $\UNITscale_0 \coloneqq \{1,\dots,\NUMlam\} \setminus \{0\} = \{1,\dots,\NUMlam\}$. For this reason, from now on we will focus on $\OSCapproach$ only, with the non-negative integer $\SMALLscaleIDX \le \NUMlam$ assuming the use of $\OSapproach$ when $\SMALLscaleIDX=0$.
			
\section{Generalized Optimal Scaling}
\label{sec:achieve_OS_MinParam}

	Next, we intend to enrich the $\OSCapproach$ method with an extra feature, namely the ability to reduce to its minimum the number of parameters employed in a model. With this addition, $\OSCapproach$ can optionally control over-parametrization and should be beneficial for real-life models with a large number of physical parameters. This can be achieved by applying ideas and concepts relative to dimensional analysis and the Buckingham $\pi$ Theorem \cite{PhysRev.4.345} in particular. From now on, we will refer to $\OSCapproach$ with such a functionality as Generalized Optimal Scaling, or $\OSCmin$. 	
	
\subsection{Formulation}
\label{sec:OptScal_MinimalParam}

	By requiring \cref{prop:lam_coeff_I,prop:Klam_powlaw_pp,prop:char_const,prop:phys_param_indep,prop:phys_param_positive,prop:phys_param_dimen} (\autoref{sec:Prop_OS}) and applying \cref{th:DIMLESS_quant,th:PI_quantities,th:lamb_pow_PI,th:sysYY,th:sysAA,th:SysToSolveOSAA} (\autoref{sec:Th_nondim_setting} and \autoref{sec:Th_opt}), we formulate here the $\OSCmin$ methodology that \changegreen{ensures} all the features of $\OSCapproach$, \changegreen{while enabling} the minimization of the number of parameters employed. 

	Given \cref{prop:phys_param_indep,prop:phys_param_positive,prop:phys_param_dimen} for physical parameters $\Dimensional{\p}_1,\dots,\Dimensional{\p}_{\NUM{\p}}$, \autoref{th:DIMLESS_quant} and \autoref{th:PI_quantities} assure that there exist only $\NUMPI$ independent dimensionless parameters $\PI_1,\dots,\PI_{\NUMPI} \in (0,\infty)$ being a power-law monomial of $\Dimensional{\p}_1, \dots, \Dimensional{\p}_{\NUM{\p}}$. \changegreen{This is to say} that any dimensionless power-law monomial of $\Dimensional{\p}_1, \dots, \Dimensional{\p}_{\NUM{\p}}$ can be written in terms of $\PI_1,\dots,\PI_{\NUMPI}$ only and, thus, $\NUMPI$ is the minimal number of dimensionless power-law monomials of $\Dimensional{\p}_1, \dots, \Dimensional{\p}_{\NUM{\p}}$ one can consider.
	
	By invoking \namecrefs{prop:lam_coeff_I} \ref{prop:lam_coeff_I}, \ref{prop:Klam_powlaw_pp} and \ref{prop:char_const}, one reveals that each coefficient $\lam_1,\dots,\lam_{\NUMlam}$ \eqref{eqn:lam_coeff_0} of model \eqref{eqn:dimensionless_model} is a dimensionless power-law monomial of $\Dimensional{\p}_1, \dots, \Dimensional{\p}_{\NUM{\p}}$ and, thus, it can be written in terms of $\PI_1,\dots,\PI_{\NUMPI}$ only, as stated by \eqref{eqn:lambdas_fun_PI} in \autoref{th:lamb_pow_PI}. This implies that \eqref{eqn:dimensional_model} and \eqref{eqn:dimensionless_model} are equivalent to a model with governing coefficients given by $\PI_1,\dots,\PI_{\NUMPI}$ only. For this reason, it is possible to reformulate the scaling approach in terms of the independent dimensionless parameters $\PI_1,\dots,\PI_{\NUMPI}$, whose number $\NUMPI$ is minimal. 
	
	\changered{This is precisely the approach adopted by $\OSCmin$ for the optimization problem \eqref{eqn:OS_scale_sep}. In such a way, \eqref{eqn:OS_scale_sep} can be further expressed in terms of exponents $\AAcoeff_{i,k}$ \eqref{eqn:lambdas_fun_PI} which constitute the vector $\VEC{\AAcoeff} \in \mathbb{R}^{\NUMlam \, \NUMPI}$ introduced in \eqref{eqn:AAdef}. As dictated by \autoref{th:sysYY} and \autoref{th:sysAA} when \cref{prop:lam_coeff_I,prop:Klam_powlaw_pp,prop:char_const,prop:phys_param_indep,prop:phys_param_positive,prop:phys_param_dimen} hold, $\VEC{\AAcoeff}$ satisfies the linear system $\matrixSysAA \, \VEC{\AAcoeff} = \lhsSysAA$ \eqref{eqn:constr_SysAA}. Then, the optimization of $\VEC{\AAcoeff}$ must take into account such a constraint, as specified in what follows.}
		
	Given any fixed choice of $\PI_1,\dots,\PI_{\NUMPI} \in (0,\infty)$, it is possible to target values of $\VEC{\AAcoeff} \in \mathbb{R}^{\NUMlam \, \NUMPI}$ \eqref{eqn:AAdef} satisfying $\matrixSysAA \, \VEC{\AAcoeff} = \lhsSysAA$ \eqref{eqn:constr_SysAA} and minimizing the deviation $\costOS_{\UNITscale_\SMALLscaleIDX} : \mathbb{R}^{\NUMlam \, \NUMPI} \to [0,\infty)$ from unity of the magnitudes of $\lam_i$, with $i \in \UNITscale_\SMALLscaleIDX \coloneqq \{1,\dots,\NUMlam\} \setminus \{ \SMALLscaleIDX \}$ for a fixed $\SMALLscaleIDX \in \{0,1,\dots,\NUMlam\}$:
		
\begin{equation}
\AAcoeffvecOS \coloneqq
\argmin_{\substack{ \VEC{\AAcoeff} \in \mathbb{R}^{\NUMlam \, \NUMPI} \\ 
\matrixSysAA \, \VEC{\AAcoeff} = \lhsSysAA }}
\costOS_{\UNITscale_\SMALLscaleIDX}(\VEC{\AAcoeff}),
\quad
\costOS_{\UNITscale_\SMALLscaleIDX}(\VEC{\AAcoeff}) \coloneqq 
\sum_{i \in \UNITscale_\SMALLscaleIDX} \left( \log_{10}( \lam_i ) \right)^2,
\quad
\lam_i = \PI_1^{\AAcoeff_{i,1}} \, \PI_2^{\AAcoeff_{i,2}} \, \cdots \, \PI_{\NUMPI}^{\AAcoeff_{i,\NUMPI}},
\label{eqn:AAcoeffOPT_def}
\end{equation}

\noindent where $\SMALLscaleIDX = 0$ encodes the use of Optimal Scaling ($\OSapproach$), while fixing $\SMALLscaleIDX \in \{1,\dots,\NUMlam\}$ corresponds to the application of Optimal Scaling with Constraints ($\OSCapproach$).

	Given $\UNITscale = \UNITscale_\SMALLscaleIDX$, let $\helpMatSysAA = \matrixSysAA$ and $\helpRhsSysAA = \lhsSysAA$ satisfy \eqref{eqn:HP_OS_AA_analytic_form}. Then, \autoref{th:SysToSolveOSAA} allows computing $\AAcoeffvecOS$ \eqref{eqn:AAcoeffOPT_def} by means of \eqref{eqn:SysToSolveOSAA}. If $\helpMatSysAA = \matrixSysAA$ and $\helpRhsSysAA = \lhsSysAA$ do not fulfill \eqref{eqn:HP_OS_AA_analytic_form}, but such an assumption is accomplished by the row echelon form of $\matrixSysAA \, \VEC{\AAcoeff} = \lhsSysAA$, it is possible to assign $\helpMatSysAA$ and $\helpRhsSysAA$ as the corresponding matrices of such a reduced system. 
	
	
	By setting $\VEC{\AAcoeff}=\AAcoeffvecOS$ \eqref{eqn:AAcoeffOPT_def}, the following features can be attained by the resulting coefficients $\lam_1,\dots,\lam_{\NUMlam}$ \eqref{eqn:lam_coeff_0}. For any fixed choice of $\PI_1,\dots,\PI_{\NUMPI} \in (0,\infty)$, the coefficients $\lam_i$, with $i \in \UNITscale_\SMALLscaleIDX$, \changegreen{are} as close as possible to the targeted value $1$, i.e., $\OSCmin$ attempts to make
	
\begin{equation}
\lam_i = \PI_1^{\AAcoeff_{i,1}} \, \PI_2^{\AAcoeff_{i,2}} \, \cdots \, \PI_{\NUMPI}^{\AAcoeff_{i,\NUMPI}} \approx 1,
\quad \WHEN \quad
\VEC{\AAcoeff} \coloneqq (\AAcoeff_{1,1},\dots,\AAcoeff_{\NUMlam,\NUMPI})^\Transp = \AAcoeffvecOS,
\label{eqn:regime_unity_GOS}
\end{equation}	
	
\noindent for all $i \in \UNITscale_\SMALLscaleIDX \coloneqq \{1,\dots,\NUMlam\} \setminus \{ \SMALLscaleIDX \}$ and a fixed $\SMALLscaleIDX \in \{0,1,\dots,\NUMlam\}$. 

	When $\SMALLscaleIDX \in \{1,\dots,\NUMlam\}$, the inequality 
	
\begin{equation}
\lam_\SMALLscaleIDX = 
\PI_1^{\AAcoeff_{\SMALLscaleIDX,1}} \, 
\PI_2^{\AAcoeff_{\SMALLscaleIDX,2}} \, 
\cdots \, 
\PI_{\NUMPI}^{\AAcoeff_{\SMALLscaleIDX,\NUMPI}} \ll 1,
\quad \WITH \quad
\VEC{\AAcoeff} \coloneqq (\AAcoeff_{1,1},\dots,\AAcoeff_{\NUMlam,\NUMPI})^\Transp = \AAcoeffvecOS,
\label{eqn:threshold_GOS}
\end{equation}

\noindent can be used together with \eqref{eqn:regime_unity_GOS} to identify a threshold for the values of parameters $\PI_1,\dots,\PI_{\NUMPI}$ leading to a large difference in magnitudes between $\lam_\SMALLscaleIDX$ and all the remaining $\lam_i$, with $i \in \UNITscale_\SMALLscaleIDX$.

	The objectives and features of presented scaling procedures are highlighted in \autoref{tab:summary_OS}.	

\begin{table}[!h]
\begin{small}
{\renewcommand{\arraystretch}{2}
\begin{tabular}{p{2.5cm} p{7cm} p{7cm}}
    
\toprule[0.4mm]
& \textbf{Optimal Scaling with Constraints} 
& \textbf{Generalized Optimal Scaling} \\
\bottomrule[0.4mm]
	
\rowcolor[HTML]{EFEFEF}		
\textbf{Objectives} 
&
Finds coefficients $\lam_1,\dots,\lam_{\NUMlam}$ in \eqref{eqn:dimensionless_model} that meet conditions: $\lam_\SMALLscaleIDX \ll 1$, $\lam_i \approx 1$, $\forall i \neq \SMALLscaleIDX \in \{1,\dots,\NUMlam\}$
&
Finds coefficients $\lam_1,\dots,\lam_{\NUMlam}$ in \eqref{eqn:dimensionless_model} that meet conditions: $\lam_\SMALLscaleIDX \ll 1$, $\lam_i \approx 1$, $\forall i \neq \SMALLscaleIDX \in \{1,\dots,\NUMlam\}$
\newline \newline
Controls over-parametrization
\\ 	

\textbf{Employed Assumptions}
&
\cref{prop:lam_coeff_I}
&
\cref{prop:lam_coeff_I,prop:Klam_powlaw_pp,prop:char_const,prop:phys_param_indep,prop:phys_param_positive,prop:phys_param_dimen}
\\
	
\rowcolor[HTML]{EFEFEF}
\textbf{Features}
&
Provides the conditions \eqref{eqn:threshold_OSC_unitary_lam}-\eqref{eqn:threshold_OSC_ll} that allow discarding the terms in \eqref{eqn:dimensionless_model} related to $\lam_\SMALLscaleIDX \ll 1$, without significantly modifying the solutions to \eqref{eqn:dimensionless_model}
\newline \newline
The model \eqref{eqn:dimensionless_model} with $\lam_\SMALLscaleIDX=0$ shows minimal variability $\Max_{i\neq \SMALLscaleIDX} \lam_i / \Min_{i\neq \SMALLscaleIDX} \lam_i$ and reduced round-off errors 
&
Inherits all the features of Optimal Scaling with Constraints, with the conditions  \eqref{eqn:threshold_OSC_unitary_lam}-\eqref{eqn:threshold_OSC_ll} replaced by \eqref{eqn:regime_unity_GOS}-\eqref{eqn:threshold_GOS}
\newline \newline
Employs the minimal number of governing parameters
\\

\textbf{Procedure}
& 
Given any fixed choice of physical parameters $\Dimensional{\p}_1, \dots, \Dimensional{\p}_{\NUM{\p}}$, finds $\VEC{\thcc} \coloneqq (\thcc_1,\dots,\thcc_{\NUM{\var}}) \in (0,\infty)^{\NUM{\var}}$ that minimizes the deviation $\costOS_{\UNITscale_\SMALLscaleIDX}$ \eqref{eqn:OS_scale_sep} from unity of the magnitudes of $\lam_i$, $\forall i \in \UNITscale_\SMALLscaleIDX \coloneqq \{1,\dots,\NUMlam\} \setminus \{\SMALLscaleIDX\}$
&
Given any fixed choice of $\PI_1,\dots,\PI_{\NUMPI} \in (0,\infty)$, finds $\VEC{\AAcoeff}$ \eqref{eqn:AAdef} that satisfies \eqref{eqn:constr_SysAA} and minimizes the deviation $\costOS_{\UNITscale_\SMALLscaleIDX}$ \eqref{eqn:AAcoeffOPT_def} from unity of the magnitudes of $\lam_i$, $\forall i \in \UNITscale_\SMALLscaleIDX \coloneqq \{1,\dots,\NUMlam\} \setminus \{\SMALLscaleIDX\}$
\\ 

\rowcolor[HTML]{EFEFEF}
\textbf{Realization} 
&
The points $\VEC{\thcc}$ of global minimum of $\costOS_{\UNITscale_\SMALLscaleIDX}$ \eqref{eqn:OS_scale_sep} can be computed by means of \eqref{eqn:lin_sys_rho} provided $\UNITscale = \UNITscale_\SMALLscaleIDX$, as follows from \autoref{th:OSC_formula}	
&
With $\UNITscale = \UNITscale_\SMALLscaleIDX$, $\helpMatSysAA = \matrixSysAA$, $\helpRhsSysAA = \lhsSysAA$ \eqref{eqn:constr_SysAA} and assuming \eqref{eqn:HP_OS_AA_analytic_form}, \autoref{th:SysToSolveOSAA} allows computing $\VEC{\AAcoeff}=\AAcoeffvecOS$ \eqref{eqn:AAcoeffOPT_def} by means of \eqref{eqn:SysToSolveOSAA}
\newline \newline
Plugs $\VEC{\AAcoeff}=\AAcoeffvecOS$ \eqref{eqn:AAcoeffOPT_def} into \eqref{eqn:matYY_vecYY_def}, then solves for $\VEC{\TT}$ \eqref{eqn:vecYY_def} the system \eqref{eqn:SyS_TT} and uses the found $\VEC{\TT}$ to compute $\Dimensional{\thcc}_1, \dots, \Dimensional{\thcc}_{\NUM{\var}}$ as in \eqref{eqn:char_const_pow_PhysParam}
\\	

\toprule[0.4mm]	
    
\end{tabular}}
\caption{Summary and features of Optimal Scaling with Constraints \cite{RUSCONI2023127756} and Generalized Optimal Scaling ($\OSCmin$). The Optimal Scaling \cite{RUSCONI2019106944} is equivalent to $\OSCmin$ with $\SMALLscaleIDX=0$.}
\label{tab:summary_OS}
\end{small}
\end{table}		

\section{A Case Study: the classical projectile model}
\label{sec:intro_ProjPb}

	Aiming to test the scaling techniques discussed above, we apply them to the projectile model introduced in \autoref{sec:CaseStudy_ProjPb}. 
	
	First, we use the \namecrefs{th:OSC_formula} from \autoref{sec:Th_nondim_setting} to formulate the projectile model with the minimal number of parameters. \cref{prop:phys_param_indep,prop:phys_param_positive,prop:phys_param_dimen} are satisfied by the $\NUM{\p}=4$ physical parameters $\Dimensional{\p}_1,\Dimensional{\p}_2,\Dimensional{\p}_3,\Dimensional{\p}_4$ \eqref{eqn:PhysParam_def_projectile}, with
	
\begin{equation}
\UNITS{\Dimensional{\p}_1} = \UNITS{\Dimensional{\gPP}} = \meter \, \second^{-2},
\quad
\UNITS{\Dimensional{\p}_2} = \UNITS{\Dimensional{\rPP}} = \meter,
\quad
\UNITS{\Dimensional{\p}_3} = \UNITS{\NOUGHT{\Dimensional{\hPP}}} = \meter,
\quad
\UNITS{\Dimensional{\p}_4} = \UNITS{\NOUGHT{\Dimensional{\vPP}}} = \meter \, \second^{-1},
\end{equation}
	
\noindent being $\meter$ meters and $\second$ seconds. Then, \autoref{th:PI_quantities} guarantees that there exist only $\NUMPI = \NUM{\p} - \Rank(\MM) = 2$ independent dimensionless parameters, where $\MM$ is the matrix

\begin{equation}
\MM =
\begin{pmatrix}
1 & 1 & 1 & 1 \\
-2 & 0 & 0 & -1
\end{pmatrix},
\end{equation}

\noindent and the $\NUMPI = 2$ independent dimensionless parameters can be computed as in \eqref{eqn:Indep_Dimless_Param}, i.e., 

\begin{equation}
\PI_1 = \Dimensional{\gPP}^{-1/2} \, \Dimensional{\rPP}^{-1/2} \, \NOUGHT{\Dimensional{\vPP}},
\quad
\PI_2 = \Dimensional{\rPP}^{-1} \, \NOUGHT{\Dimensional{\hPP}},
\quad
\PI_1,\PI_2 \in (0,\infty).
\label{eqn:PI_ProjPb}
\end{equation}

\noindent \cref{prop:lam_coeff_I,prop:Klam_powlaw_pp} are fulfilled by $\NUMlam=4$ coefficients $\lam_1,\lam_2,\lam_3,\lam_4$ \eqref{eqn:lambdas_def_projectile}, with

\begin{equation}
\Dimensional{\Klam}_1 = \Dimensional{\p}_1 = \Dimensional{\gPP},
\quad
\Dimensional{\Klam}_2 = \Dimensional{\p}_2^{-1} = \Dimensional{\rPP}^{-1},
\quad
\Dimensional{\Klam}_3 = \Dimensional{\p}_3 = \NOUGHT{\Dimensional{\hPP}},
\quad
\Dimensional{\Klam}_4 = \Dimensional{\p}_4 = \NOUGHT{\Dimensional{\vPP}},
\label{eqn:Kappas_ProjPb}
\end{equation}

\begin{equation}
\CC_{1,1} = 2, \quad \CC_{1,2} = -1, \quad 
\CC_{2,1} = 0, \quad \CC_{2,2} = 1, \quad
\CC_{3,1} = 0, \quad \CC_{3,2} = -1, \quad 
\CC_{4,1} = 1, \quad \CC_{4,2} = -1.
\label{eqn:CC_ProjPb}
\end{equation}

\noindent Then, by also bearing \cref{prop:char_const} for $\Dimensional{\thcc}_1,\Dimensional{\thcc}_2$ \eqref{eqn:change_of_var_projectile}, i.e.,

\begin{equation}
\Dimensional{\thcc}_1 
= 
\Dimensional{\p}_1^{\TT_{1,1}} \, 
\Dimensional{\p}_2^{\TT_{2,1}} \, 
\Dimensional{\p}_3^{\TT_{3,1}} \, 
\Dimensional{\p}_4^{\TT_{4,1}}
=
\Dimensional{\gPP}^{\TT_{1,1}} \, 
\Dimensional{\rPP}^{\TT_{2,1}} \, 
\NOUGHT{\Dimensional{\hPP}}^{\TT_{3,1}} \, 
\NOUGHT{\Dimensional{\vPP}}^{\TT_{4,1}},
\quad
\Dimensional{\thcc}_2 
= 
\Dimensional{\p}_1^{\TT_{1,2}} \, 
\Dimensional{\p}_2^{\TT_{2,2}} \, 
\Dimensional{\p}_3^{\TT_{3,2}} \, 
\Dimensional{\p}_4^{\TT_{4,2}}
=
\Dimensional{\gPP}^{\TT_{1,2}} \, 
\Dimensional{\rPP}^{\TT_{2,2}} \, 
\NOUGHT{\Dimensional{\hPP}}^{\TT_{3,2}} \, 
\NOUGHT{\Dimensional{\vPP}}^{\TT_{4,2}},
\label{eqn:thcc_ProjPb_pow_PhysParam}
\end{equation}

\noindent \autoref{th:lamb_pow_PI} guarantees that the coefficients $\lam_1,\lam_2,\lam_3,\lam_4$ \eqref{eqn:lambdas_def_projectile} can be written as power-law monomials of $\PI_1,\PI_2$ \eqref{eqn:PI_ProjPb} only, i.e.,

\begin{equation}
\lam_1 = \PI_1^{\AAcoeff_{1,1}} \, \PI_2^{\AAcoeff_{1,2}}, 
\quad
\lam_2 = \PI_1^{\AAcoeff_{2,1}} \, \PI_2^{\AAcoeff_{2,2}}, 
\quad
\lam_3 = \PI_1^{\AAcoeff_{3,1}} \, \PI_2^{\AAcoeff_{3,2}}, 
\quad
\lam_4 = \PI_1^{\AAcoeff_{4,1}} \, \PI_2^{\AAcoeff_{4,2}}. 
\label{eqn:lam_ProjPb_pow_PI}
\end{equation}

\noindent Finally, \autoref{th:sysAA} indicates under \cref{prop:lam_coeff_I,prop:Klam_powlaw_pp,prop:char_const,prop:phys_param_indep,prop:phys_param_positive,prop:phys_param_dimen} that the vector

\begin{equation}
\VEC{\AAcoeff} =
\begin{pmatrix}
\AAcoeff_{1,1} & \AAcoeff_{1,2} &
\AAcoeff_{2,1} & \AAcoeff_{2,2} &
\AAcoeff_{3,1} & \AAcoeff_{3,2} &
\AAcoeff_{4,1} & \AAcoeff_{4,2}
\end{pmatrix}^{\Transp}
\in \mathbb{R}^{8}
\label{eqn:vecAA_ProjPb}
\end{equation}

\noindent of exponents in \eqref{eqn:lam_ProjPb_pow_PI} must satisfy the linear system $\matrixSysAA \, \VEC{\AAcoeff} = \lhsSysAA$ \eqref{eqn:constr_SysAA}, which, for the projectile model, reads in row echelon form as

\begin{equation}
\begin{pmatrix}
-1/2 & 0 & 1/2 & 0 & 0 & 0 & 1 & 0 \\
0 & -1/2 & 0 & 1/2 & 0 & 0 & 0 & 1 \\
0 & 0 & 1 & 0 & 1 & 0 & 0 & 0 \\
0 & 0 & 0 & 1 & 0 & 1 & 0 & 0
\end{pmatrix}
\VEC{\AAcoeff}
=
\begin{pmatrix}
1 \\ 0 \\ 0 \\ 1
\end{pmatrix}.
\label{eqn:AAsys_ProjPb}
\end{equation}

\subsection{Optimally Scaled Coefficients}
\label{sec:ProjPb_coeff_OS_OSC}

	We compute the values of coefficients $\lam_1,\lam_2,\lam_3,\lam_4$ \eqref{eqn:lambdas_def_projectile} using the Generalized Optimal Scaling approach presented in \autoref{sec:achieve_OS_MinParam}. In particular, we provide and discuss the values of coefficients $\lam_1,\lam_2,\lam_3,\lam_4$ \eqref{eqn:lam_ProjPb_pow_PI} when the vector $\VEC{\AAcoeff}$ \eqref{eqn:vecAA_ProjPb} of exponents in \eqref{eqn:lam_ProjPb_pow_PI} is optimized according to \eqref{eqn:AAcoeffOPT_def}. In this case, $\UNITscale = \UNITscale_\SMALLscaleIDX \coloneqq \{1,2,3,4\} \setminus \{ \SMALLscaleIDX \}$, where $\SMALLscaleIDX = 0$ means the use of $\OSapproach$, while fixing $\SMALLscaleIDX \in \{1,2,3,4\}$ corresponds to the application of $\OSCapproach$.
	
	\autoref{th:SysToSolveOSAA} allows solving the optimization problem \eqref{eqn:AAcoeffOPT_def} for any fixed choice of $\PI_1,\PI_2$ \eqref{eqn:PI_ProjPb}, since one can set \eqref{eqn:AAsys_ProjPb} as the constraint $\helpMatSysAA \, \VEC{\AAcoeff} = \helpRhsSysAA$ \eqref{eqn:Constr_SysOSAA} to satisfy the assumption \eqref{eqn:HP_OS_AA_analytic_form}. In particular, the system \eqref{eqn:SysToSolveOSAA} can be solved to find optimal values of $\VEC{\AAcoeff}$ \eqref{eqn:vecAA_ProjPb} for plugging into \eqref{eqn:lam_ProjPb_pow_PI}. Then, the system \eqref{eqn:SysToSolveOSAA} reads as

\begin{equation}
\begin{pmatrix}
\helpMatSysAA & \MATRmn{0}{4}{4} \\
\Gmatr & \helpMatSysAA^{\Transp}
\end{pmatrix}
\begin{pmatrix} 
\VEC{\AAcoeff} \\ \helpVecSysAA
\end{pmatrix}
=
\begin{pmatrix}
\helpRhsSysAA \\ \MATRmn{0}{8}{1} 
\end{pmatrix},
\label{eqn:SysAAopt1_ProjPb}
\end{equation}

\noindent with $\VEC{\AAcoeff} \in \mathbb{R}^8$ given by \eqref{eqn:vecAA_ProjPb}, $\helpVecSysAA = (\hVSysAAcomp_1,\hVSysAAcomp_2,\hVSysAAcomp_3,\hVSysAAcomp_4)^{\Transp} \in \mathbb{R}^4$, 

\begin{equation}
\helpMatSysAA =
\begin{pmatrix}
-1/2 & 0 & 1/2 & 0 & 0 & 0 & 1 & 0 \\
0 & -1/2 & 0 & 1/2 & 0 & 0 & 0 & 1 \\
0 & 0 & 1 & 0 & 1 & 0 & 0 & 0 \\
0 & 0 & 0 & 1 & 0 & 1 & 0 & 0
\end{pmatrix},
\quad
\helpRhsSysAA =
\begin{pmatrix}
1 \\ 0 \\ 0 \\ 1
\end{pmatrix},
\label{eqn:SysAAopt2_ProjPb}
\end{equation}

\begin{equation}
\Gmatr =
\begin{pmatrix}
\indicator{\SMALLscaleIDX}{1} \, \Pmatr & \MATRmn{0}{2}{2} & \MATRmn{0}{2}{2} & \MATRmn{0}{2}{2} \\
\MATRmn{0}{2}{2} & \indicator{\SMALLscaleIDX}{2} \, \Pmatr & \MATRmn{0}{2}{2} & \MATRmn{0}{2}{2} \\
\MATRmn{0}{2}{2} & \MATRmn{0}{2}{2} & \indicator{\SMALLscaleIDX}{3} \, \Pmatr & \MATRmn{0}{2}{2} \\
\MATRmn{0}{2}{2} & \MATRmn{0}{2}{2} & \MATRmn{0}{2}{2} & \indicator{\SMALLscaleIDX}{4} \, \Pmatr
\end{pmatrix},
\quad
\Pmatr =
\begin{pmatrix}
2 (\log_{10}(\PI_1))^2 & 2 \log_{10}(\PI_1) \log_{10}(\PI_2) \\
2 \log_{10}(\PI_1) \log_{10}(\PI_2) & 2 (\log_{10}(\PI_2))^2   
\end{pmatrix},
\label{eqn:SysAAopt3_ProjPb}
\end{equation}

\noindent where $\indicator{\SMALLscaleIDX}{i}=1$ if $i \neq \SMALLscaleIDX$, $\indicator{\SMALLscaleIDX}{i}=0$ if $i = \SMALLscaleIDX$, and $\MATRmn{0}{m}{n}$ is a matrix with $m$ rows and $n$ columns containing all zeros. By plugging the solutions to \eqref{eqn:SysAAopt1_ProjPb}-\eqref{eqn:SysAAopt3_ProjPb} into \eqref{eqn:lam_ProjPb_pow_PI}, one obtains

\begin{align}
& \lam_1 = \PI_1^{-4/11} \, \PI_2^{1/11}, 
&& \lam_2 = \PI_1^{2/11} \, \PI_2^{5/11},  
&& \lam_3 = \PI_1^{-2/11} \, \PI_2^{6/11}, 
&& \lam_4 = \PI_1^{8/11} \, \PI_2^{-2/11}, 
&& \IF \quad \SMALLscaleIDX = 0,
\label{eqn:OS_ProgPb_lam_PI} \\	
& \lam_1 = \PI_1^{-2} \, \PI_2^{1/2}, 
&& \lam_2 = \PI_2^{1/2}, 
&& \lam_3 = \PI_2^{1/2},  
&& \lam_4 = 1, 
&& \IF \quad \SMALLscaleIDX = 1,
\label{eqn:OSC_ProgPb_s1_lam_PI} \\	
& \lam_1 = \PI_1^{-1/3} \, \PI_2^{1/6}, 
&& \lam_2 = \PI_1^{1/3} \, \PI_2^{5/6},  
&& \lam_3 = \PI_1^{-1/3} \, \PI_2^{1/6}, 
&& \lam_4 = \PI_1^{2/3} \, \PI_2^{-1/3},
&& \IF \quad \SMALLscaleIDX = 2,
\label{eqn:OSC_ProgPb_s2_lam_PI} \\
& \lam_1 = \PI_1^{-1/3},  
&& \lam_2 = \PI_1^{1/3},  
&& \lam_3 = \PI_1^{-1/3} \, \PI_2,  
&& \lam_4 = \PI_1^{2/3}, 
&& \IF \quad \SMALLscaleIDX = 3,
\label{eqn:OSC_ProgPb_s3_lam_PI} \\
& \lam_1 = 1,  
&& \lam_2 = \PI_2^{1/2},  
&& \lam_3 = \PI_2^{1/2},  
&& \lam_4 = \PI_1 \, \PI_2^{-1/4},  
&& \IF \quad \SMALLscaleIDX = 4.
\label{eqn:OSC_ProgPb_s4_lam_PI} 
\end{align}

\noindent The coefficients $\lam_1,\lam_2,\lam_3,\lam_4$ \eqref{eqn:OS_ProgPb_lam_PI}-\eqref{eqn:OSC_ProgPb_s4_lam_PI} are depicted in \autoref{fig:compare_lambdas_OS_ProjPb} as functions of the parameters $\PI_1,\PI_2$ \eqref{eqn:PI_ProjPb}. For all tested $\SMALLscaleIDX=0,1,2,3,4$, $\OSCmin$ identifies the values of $\PI_1,\PI_2$ which satisfy the regime \eqref{eqn:regime_unity_GOS}-\eqref{eqn:threshold_GOS}. Such values lie in the areas indicated in \autoref{fig:compare_lambdas_OS_ProjPb} by magenta dashed lines.
			
\newgeometry{top=0.5cm,bottom=1.5cm,left=0.75cm,right=0.75cm}

\begin{figure}[!hp]
\centering
\begin{subfigure}[c]{1\linewidth}
\centering
\includegraphics[scale=0.82]{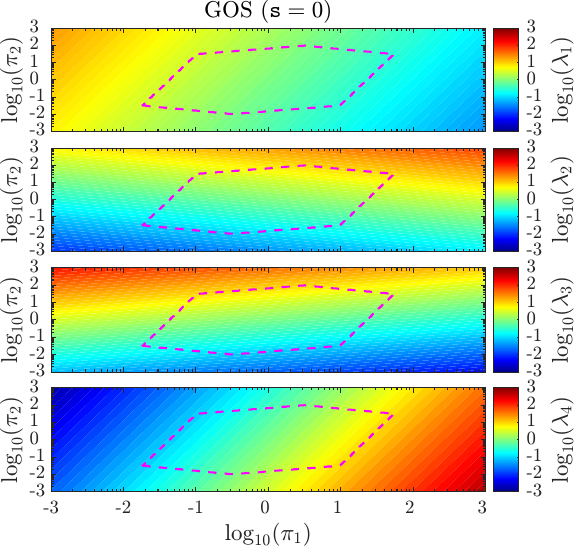}
\caption{Optimally scaled coefficients $\lam_1,\lam_2,\lam_3,\lam_4$ \eqref{eqn:OS_ProgPb_lam_PI} with $\SMALLscaleIDX=0$.}
\label{fig:lambdas_OS_ProjPb}
\end{subfigure}
\newline \\
\vspace{0.5cm}
\begin{subfigure}[c]{.495\linewidth}
\centering
\includegraphics[scale=0.82]{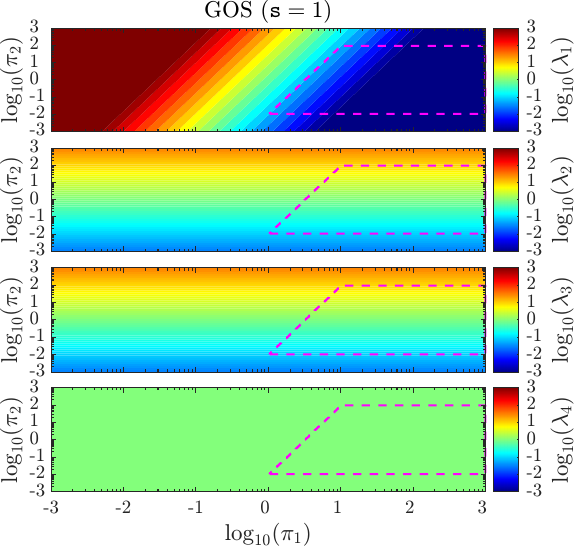}
\caption{Optimally scaled coefficients $\lam_1,\lam_2,\lam_3,\lam_4$ \eqref{eqn:OSC_ProgPb_s1_lam_PI} with $\SMALLscaleIDX=1$.}
\label{fig:lambdas_OSC1_ProjPb}
\end{subfigure}	
\begin{subfigure}[c]{.495\linewidth}
\centering
\includegraphics[scale=0.82]{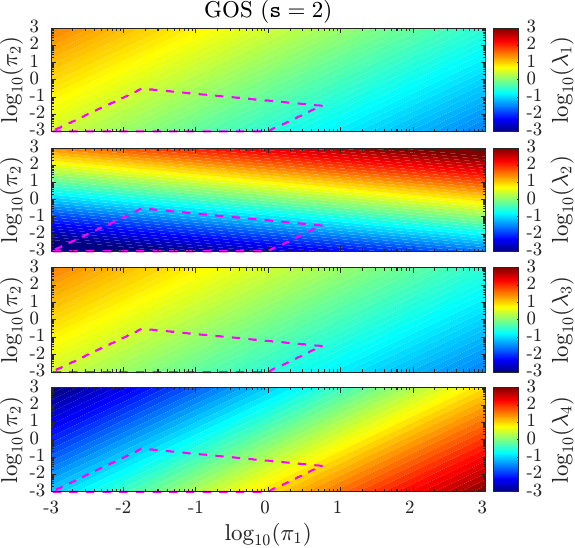}
\caption{Optimally scaled coefficients $\lam_1,\lam_2,\lam_3,\lam_4$ \eqref{eqn:OSC_ProgPb_s2_lam_PI} with $\SMALLscaleIDX=2$.}
\label{fig:lambdas_OSC2_ProjPb}
\end{subfigure}	
\newline \\
\vspace{0.5cm}
\begin{subfigure}[c]{.495\linewidth}
\centering
\includegraphics[scale=0.82]{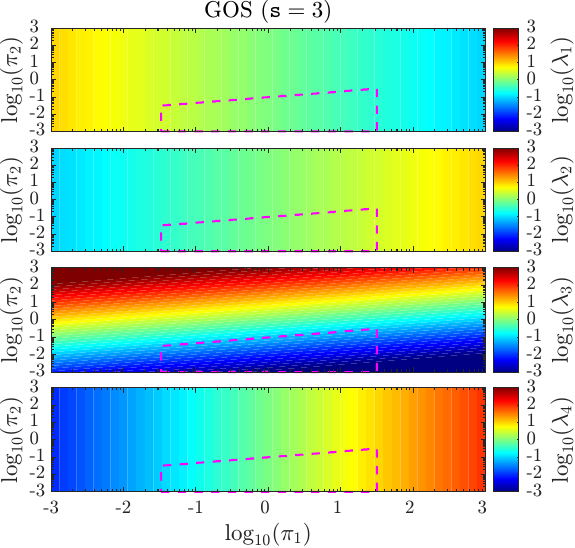}
\caption{Optimally scaled coefficients $\lam_1,\lam_2,\lam_3,\lam_4$ \eqref{eqn:OSC_ProgPb_s3_lam_PI} with $\SMALLscaleIDX=3$.}
\label{fig:lambdas_OSC3_ProjPb}
\end{subfigure}	
\begin{subfigure}[c]{.495\linewidth}
\centering
\includegraphics[scale=0.82]{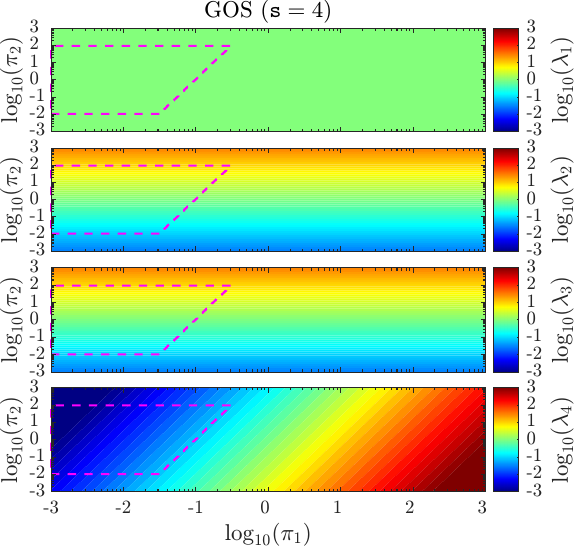}
\caption{Optimally scaled coefficients $\lam_1,\lam_2,\lam_3,\lam_4$ \eqref{eqn:OSC_ProgPb_s4_lam_PI} with $\SMALLscaleIDX=4$.}
\label{fig:lambdas_OSC4_ProjPb}
\end{subfigure}	
\caption{Coefficients \eqref{eqn:OS_ProgPb_lam_PI}-\eqref{eqn:OSC_ProgPb_s4_lam_PI} provided by $\OSCmin$ as functions of $\PI_1,\PI_2$ \eqref{eqn:PI_ProjPb}. The regime \eqref{eqn:regime_unity_GOS}-\eqref{eqn:threshold_GOS} is achieved in the areas delimited by magenta dashed lines.}
\label{fig:compare_lambdas_OS_ProjPb}
\end{figure}	

\restoregeometry

\subsection{Generalized Optimal Scaling $(\SMALLscaleIDX=0)$: Results}
\label{sec:OS_ProgPb}

	Next, we illustrate the advantages provided by the coefficients $\lam_1,\lam_2,\lam_3,\lam_4$ \eqref{eqn:OS_ProgPb_lam_PI} resulting from application of Generalized Optimal Scaling with $\SMALLscaleIDX=0$ to the projectile model.
		
	\autoref{tab:summary_OS} states that such coefficients allow minimizing the variability $\Max_i \lam_i / \Min_i \lam_i$. In order to verify such a claim, we compare the values achieved by coefficients \eqref{eqn:OS_ProgPb_lam_PI} with \changegreen{those obtained} by the following approach suggested in the literature \cite{Holmes2009_ND,Scaling_Langtangen}. Provided the values of physical parameters $\Dimensional{\p}_1,\Dimensional{\p}_2,\Dimensional{\p}_3,\Dimensional{\p}_4$ \eqref{eqn:PhysParam_def_projectile}, it is possible to impose $\NUM{\var}=2$ out of the $\NUMlam=4$ coefficients $\lam_1,\lam_2,\lam_3,\lam_4$ \eqref{eqn:lambdas_def_projectile} to be equal to $1$. Then, one aims to solve the resulting system of equations with respect to $\Dimensional{\thcc}_1$ and $\Dimensional{\thcc}_2$. Finally, the corresponding values of $\lam_1,\lam_2,\lam_3,\lam_4$ \eqref{eqn:lambdas_def_projectile} are obtained by plugging the computed $\Dimensional{\thcc}_1$ and $\Dimensional{\thcc}_2$ into \eqref{eqn:lambdas_def_projectile}.
	
	Aiming to perform such a comparison, we set $\Dimensional{\gPP} = 9.81 \, \meter \, \second^{-2}$, $\Dimensional{\rPP} = 6.3781 \times 10^{6} \, \meter$, $\NOUGHT{\Dimensional{\hPP}} = 1 \, \meter$ and $\NOUGHT{\Dimensional{\vPP}} = 1 \, \meter \, \second^{-1}$, being $\Dimensional{\gPP}$ the gravitational acceleration on the planet Earth and $\Dimensional{\rPP}$ the radius of the Earth. Though the chosen values of initial height $\NOUGHT{\Dimensional{\hPP}}$ and initial speed $\NOUGHT{\Dimensional{\vPP}}$ have no particular physical meaning, they are reasonable for a ball thrown vertically upward from the surface of the Earth. The comparison of $\lam_i$ values for different scaling methods and the corresponding variability $\Max_i \lam_i / \Min_i \lam_i$ are shown in \autoref{fig:lambdas_val_PP_(a)} and \autoref{fig:lambdas_val_PP_(b)}, respectively. In the performed experiment, $\OSCmin$ allows reducing the variability of coefficients $\lam_1,\lam_2,\lam_3,\lam_4$ by at least two more orders of magnitude compared to any of other traditional scaling approaches, as illustrated by \autoref{fig:lambdas_val_PP_(b)}.

\begin{figure}[!h]
\centering
\begin{subfigure}[c]{.495\linewidth}
\centering
\includegraphics[scale=0.65]{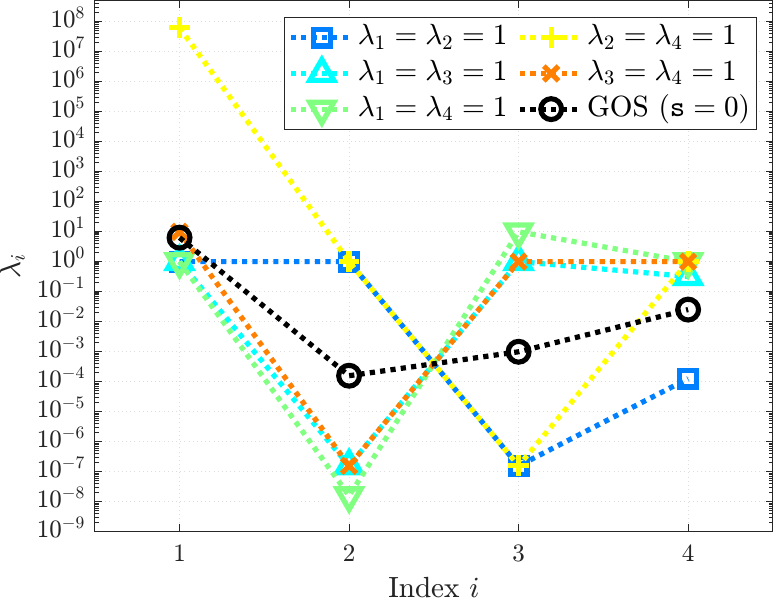}
\caption{Coefficients $\lam_1,\lam_2,\lam_3,\lam_4$ \eqref{eqn:lambdas_def_projectile} are shown by colored symbols, while coefficients $\lam_1,\lam_2,\lam_3,\lam_4$ \eqref{eqn:OS_ProgPb_lam_PI} are depicted with black circles.}
\label{fig:lambdas_val_PP_(a)}
\end{subfigure}	
\captionsetup[subfigure]{oneside,margin={0.5cm,0cm}}
\begin{subfigure}[c]{.495\linewidth}
\centering
\includegraphics[scale=0.65]{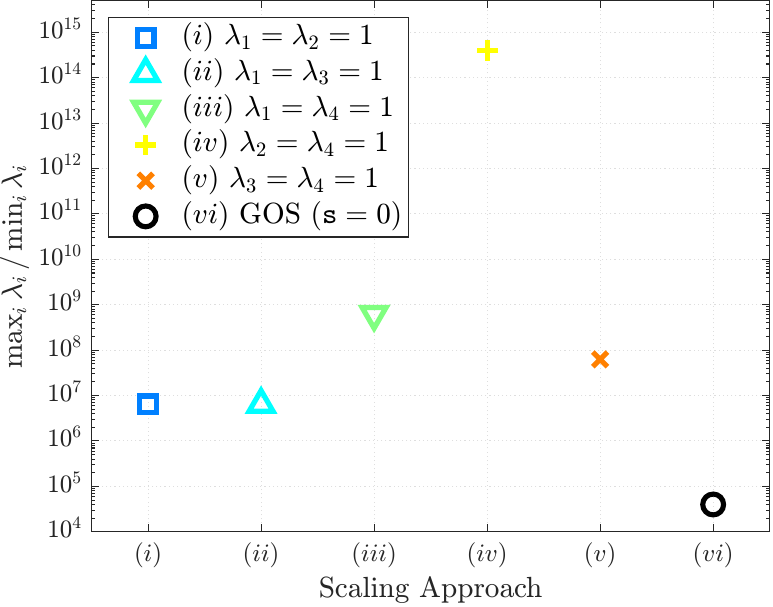}
\caption{Variability $\Max_i \lam_i / \Min_i \lam_i$, with $\lam_1,\lam_2,\lam_3,\lam_4$ \eqref{eqn:lambdas_def_projectile} and \eqref{eqn:OS_ProgPb_lam_PI} shown by colored symbols and black circles, respectively.}
\label{fig:lambdas_val_PP_(b)}
\end{subfigure}	
\caption{Values attained by coefficients $\lam_1,\lam_2,\lam_3,\lam_4$ \eqref{eqn:lambdas_def_projectile} and \eqref{eqn:OS_ProgPb_lam_PI} (\autoref{fig:lambdas_val_PP_(a)}) and corresponding variability $\Max_i \lam_i / \Min_i \lam_i$ (\autoref{fig:lambdas_val_PP_(b)}), when $\Dimensional{\gPP} = 9.81 \, \meter \, \second^{-2}$, $\Dimensional{\rPP} = 6.3781 \times 10^{6} \, \meter$, $\NOUGHT{\Dimensional{\hPP}} = 1 \, \meter$, $\NOUGHT{\Dimensional{\vPP}} = 1 \, \meter \, \second^{-1}$ and $\PI_1,\PI_2$ \eqref{eqn:PI_ProjPb}. In both \autoref{fig:lambdas_val_PP_(a)} and \autoref{fig:lambdas_val_PP_(b)}, the black circles provide the numerical values obtained by Generalized Optimal Scaling ($\OSCmin$) with $\SMALLscaleIDX=0$, i.e., the coefficients $\lam_1,\lam_2,\lam_3,\lam_4$ \eqref{eqn:OS_ProgPb_lam_PI}. \autoref{fig:lambdas_val_PP_(a)} and \autoref{fig:lambdas_val_PP_(b)} also show the values achieved by $\lam_1,\lam_2,\lam_3,\lam_4$ \eqref{eqn:lambdas_def_projectile} if $\NUM{\var}=2$ out of the $\NUMlam=4$ are imposed to be equal to $1$, as indicated by the legend in \autoref{fig:lambdas_val_PP_(a)} and \autoref{fig:lambdas_val_PP_(b)}. The \figuresname do not show the data corresponding to $\lam_2=\lam_3=1$, since it leads to $\Dimensional{\p}_2^{-1} \, \Dimensional{\thcc}_2 = \Dimensional{\p}_3 \, \Dimensional{\thcc}_2^{-1} = 1$ and there are no $\Dimensional{\thcc}_1$ and $\Dimensional{\thcc}_2$ solving such a system of equations.}
\label{fig:lambdas_val_ProjPb}
\end{figure}	

	The minimal variability of the coefficients $\lam_1,\lam_2,\lam_3,\lam_4$ is beneficial for reducing round-off errors, as claimed in \autoref{tab:summary_OS}. In order to illustrate the goodness of the $\OSCmin$ approach for numerical computations, we perform the following experiment.

	The Matlab solver $\odeFF$ \cite{ShampineReicheltODEMatlab} is employed to compute the numerical solution $\xPPdl=\xPPNdig$ to the Ordinary Differential Equation \eqref{eqn:dimensioless_ODE_projectile} in the time interval $[0,\TmaxPP/\thcc_1]$. The subscript $\Ndigit$ indicates that the numerical values of coefficients $\lam_1,\lam_2,\lam_3,\lam_4$ and \changegreen{every} computed quantity at each time step are rounded to $\Ndigit$ decimal places. The relative error is estimated as 

\begin{equation}
\errRoundPP \coloneqq
\frac{1}{\TmaxPP}
\int_{0}^{\TmaxPP}
\! \frac{ \left| \thcc_2 \, \xPPNdig( \tPPdl/\thcc_1 ) - \xPPref(\tPPdl) \right|  }
{ \left| \xPPref(\tPPdl) \right| }
\, d\tPPdl,
\label{eqn:error_Round_PP} 
\end{equation}

\noindent where $\xPPref$ is the numerical solution to the Ordinary Differential Equation \eqref{eqn:dimensioless_ODE_projectile} with $\lam_1 = \gPP$, $\lam_2 = \rPP^{-1}$, $\lam_3=\NOUGHT{\hPP}$, $\lam_4=\NOUGHT{\vPP}$, i.e., $\thcc_1=\thcc_2=1$, that has been taken as a reference in the time interval $[0,\TmaxPP]$. Such a solution $\xPPref$ is obtained by using the Matlab solver $\odeFF$ \cite{ShampineReicheltODEMatlab} with the default 16 significant digits of precision for numerical computations. A relative error tolerance of $10^{-5}$ and an absolute error tolerance of $10^{-8}$ have been required for all the computations performed by $\odeFF$ \cite{ShampineReicheltODEMatlab}, being $2$ orders of magnitude smaller than the default tolerances. The Simpson's quadrature rule \cite{DG.2024.simps} has been used to compute the integral in \eqref{eqn:error_Round_PP}.

	\autoref{fig:ref_sol_Round_PP} shows the reference solution $\xPPref$, while \autoref{fig:en_Round_PP} compares the error $\errRoundPP$ \eqref{eqn:error_Round_PP} when $\xPPNdig$ is computed with several choices of coefficients $\lam_1,\lam_2,\lam_3,\lam_4$, as indicated in the legend of \autoref{fig:en_Round_PP}. The colored symbols denote the error corresponding to $\xPPNdig$ calculated by using the numerical values of $\lam_1,\lam_2,\lam_3,\lam_4$ obtained by imposing $\NUM{\var}=2$ out of the $\NUMlam=4$ coefficients $\lam_1,\lam_2,\lam_3,\lam_4$ \eqref{eqn:lambdas_def_projectile} to be equal to $1$, as detailed above. It is possible to notice that the numerical values of $\lam_1,\lam_2,\lam_3,\lam_4$ \eqref{eqn:OS_ProgPb_lam_PI} provided by the $\OSCmin$ approach allow reducing the error $\errRoundPP$ (black circles) up to two more orders of magnitude, compared to the other traditional scaling approaches.

\begin{figure}[!h]
\centering
\begin{subfigure}[c]{.495\linewidth}
\centering
\includegraphics[scale=0.65]{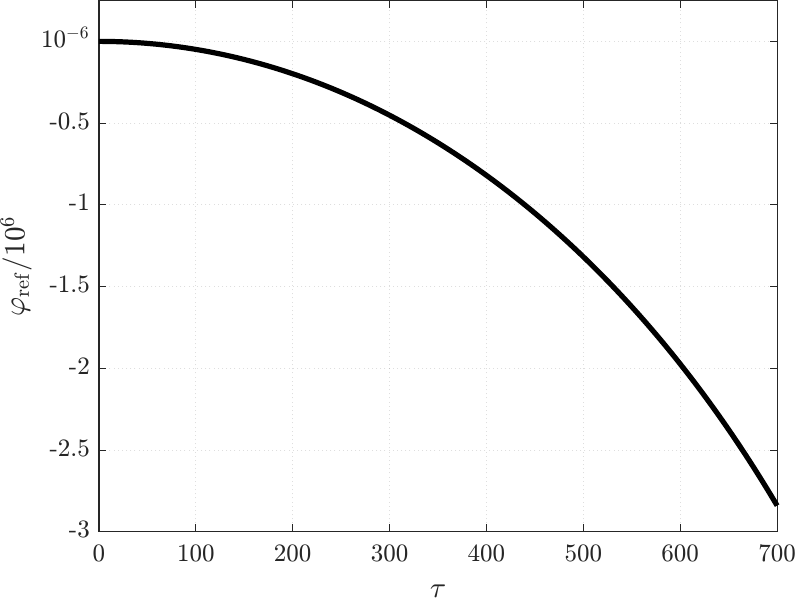}
\caption{Reference solution $\xPPdl=\xPPref$ to \eqref{eqn:dimensioless_ODE_projectile}.}
\label{fig:ref_sol_Round_PP}
\end{subfigure}	
\begin{subfigure}[c]{.495\linewidth}
\centering
\includegraphics[scale=0.65]{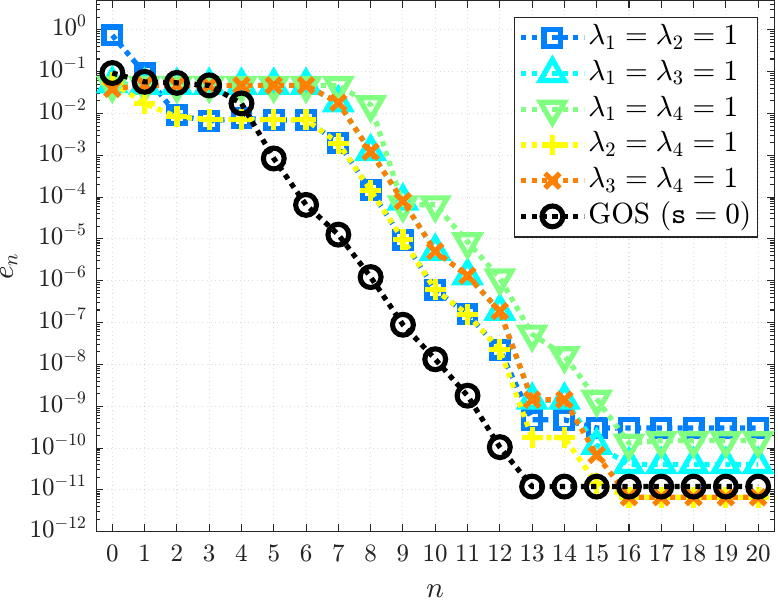}
\caption{Error $\errRoundPP$ \eqref{eqn:error_Round_PP}.}
\label{fig:en_Round_PP}
\end{subfigure}	
\caption{\autoref{fig:ref_sol_Round_PP} shows the numerical solution $\xPPdl=\xPPref$ to the Ordinary Differential Equation \eqref{eqn:dimensioless_ODE_projectile} in the time interval $[0,\TmaxPP]$, with $\lam_1 = \gPP$, $\lam_2 = \rPP^{-1}$, $\lam_3=\NOUGHT{\hPP}$, $\lam_4=\NOUGHT{\vPP}$ and $\TmaxPP=700$, being $\Dimensional{\gPP} = 9.81 \, \meter \, \second^{-2}$, $\Dimensional{\rPP} = 6.3781 \times 10^{6} \, \meter$, $\NOUGHT{\Dimensional{\hPP}} = 1 \, \meter$, $\NOUGHT{\Dimensional{\vPP}} = 1 \, \meter \, \second^{-1}$. Given the shown $\xPPref$ and $\Dimensional{\gPP}$, $\Dimensional{\rPP}$, $\NOUGHT{\Dimensional{\hPP}}$, $\NOUGHT{\Dimensional{\vPP}}$ as above, \autoref{fig:en_Round_PP} depicts the error $\errRoundPP$ \eqref{eqn:error_Round_PP} corresponding to several choices of the coefficients $\lam_1,\lam_2,\lam_3,\lam_4$, as indicated in the legend. The black circles show the error made by using $\lam_1,\lam_2,\lam_3,\lam_4$ \eqref{eqn:OS_ProgPb_lam_PI} with $\PI_1,\PI_2$ \eqref{eqn:PI_ProjPb}, as provided by Generalized Optimal Scaling ($\OSCmin$) with $\SMALLscaleIDX=0$. On the other hand, the colored symbols indicate the error corresponding to numerical values of $\lam_1,\lam_2,\lam_3,\lam_4$ obtained by imposing $\NUM{\var}=2$ out of the $\NUMlam=4$ coefficients $\lam_1,\lam_2,\lam_3,\lam_4$ \eqref{eqn:lambdas_def_projectile} to be equal to $1$. \autoref{fig:en_Round_PP} does not show the data corresponding to $\lam_2=\lam_3=1$, since it leads to $\Dimensional{\p}_2^{-1} \, \Dimensional{\thcc}_2 = \Dimensional{\p}_3 \, \Dimensional{\thcc}_2^{-1} = 1$ and there are no $\Dimensional{\thcc}_1$ and $\Dimensional{\thcc}_2$ solving such a system of equations.}
\label{fig:error_Round_PP}
\end{figure}	

\changered{The numerical experiments presented above use a single set of physical parameters, with $\NOUGHT{\Dimensional{\hPP}}$ and $\NOUGHT{\Dimensional{\vPP}}$ of comparable magnitude, namely $\Dimensional{\gPP} = 9.81 \, \meter \, \second^{-2}$, $\Dimensional{\rPP} = 6.3781 \times 10^{6} \, \meter$, $\NOUGHT{\Dimensional{\hPP}} = 1 \, \meter$ and $\NOUGHT{\Dimensional{\vPP}} = 1 \, \meter \, \second^{-1}$. To further investigate the robustness of the proposed approach, we compute the coefficient variability, $\Max_i \lam_i / \Min_i \lam_i$, and the round-off error, $\errRoundPP$ \eqref{eqn:error_Round_PP}, keeping all parameters fixed as in \autoref{fig:lambdas_val_ProjPb} and \ref{fig:error_Round_PP}, respectively, while varying the value of $\NOUGHT{\Dimensional{\hPP}}$ over several orders of magnitude (\autoref{fig:CoeffVar_h0} and \ref{fig:RoundErr_h0}). We perform an analogous analysis by varying $\NOUGHT{\Dimensional{\vPP}}$ over several orders of magnitude, while keeping all other parameters fixed (\autoref{fig:CoeffVar_v0} and \ref{fig:RoundErr_v0}). The results presented in \autoref{fig:ParamSensitivity_h0_v0} show that there exist ranges of parameter values for which $\OSCmin$ reduces both the coefficient variability and the round-off error by up to two additional orders of magnitude compared with traditional scaling techniques.}

\begin{figure}[!h]
\centering
\begin{subfigure}[c]{.495\linewidth}
\centering
\includegraphics[scale=0.65]{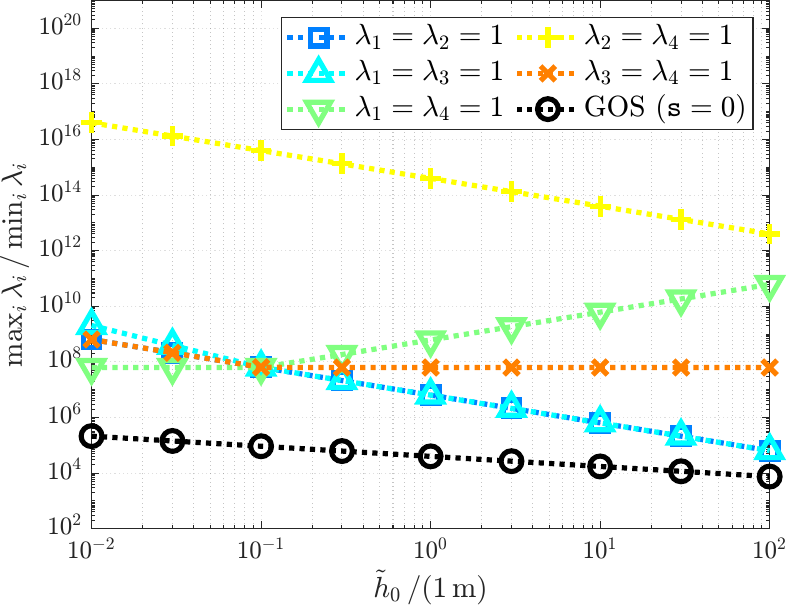}
\caption{Variability $\Max_i \lam_i / \Min_i \lam_i$ versus $\NOUGHT{\Dimensional{\hPP}}$.}
\label{fig:CoeffVar_h0}
\end{subfigure}	
\begin{subfigure}[c]{.495\linewidth}
\centering
\includegraphics[scale=0.65]{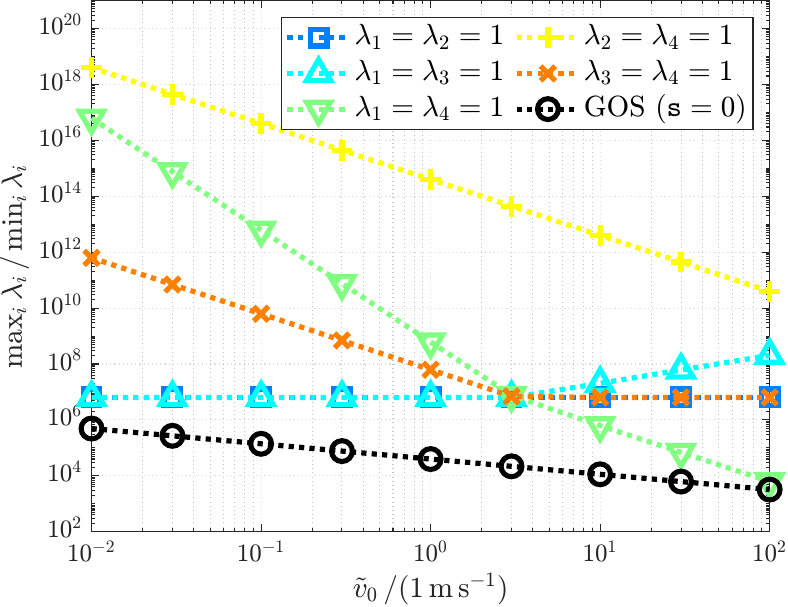}
\caption{Variability $\Max_i \lam_i / \Min_i \lam_i$ versus $\NOUGHT{\Dimensional{\vPP}}$.}
\label{fig:CoeffVar_v0}
\end{subfigure}	\\
\vspace{0.1in}
\centering
\begin{subfigure}[c]{.495\linewidth}
\centering
\includegraphics[scale=0.65]{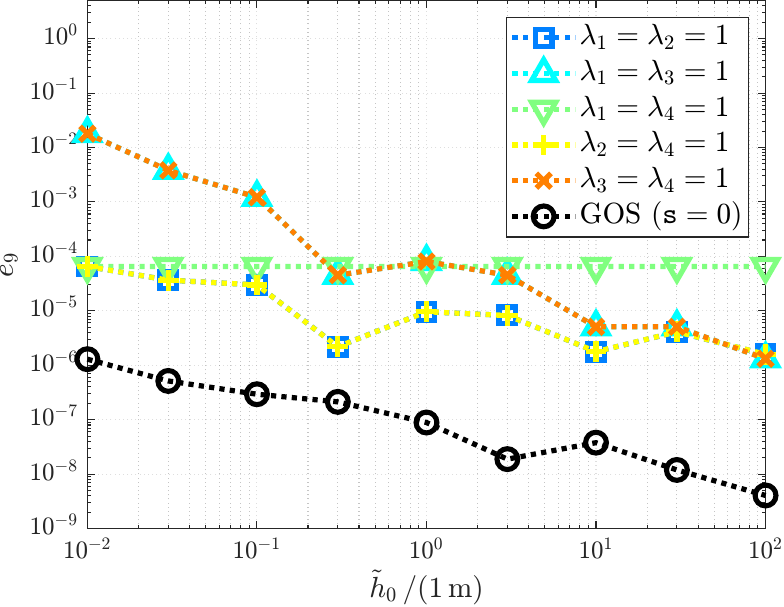}
\caption{Error $\errRoundPP$ \eqref{eqn:error_Round_PP}, with $\Ndigit=9$, versus $\NOUGHT{\Dimensional{\hPP}}$.}
\label{fig:RoundErr_h0}
\end{subfigure}	
\begin{subfigure}[c]{.495\linewidth}
\centering
\includegraphics[scale=0.65]{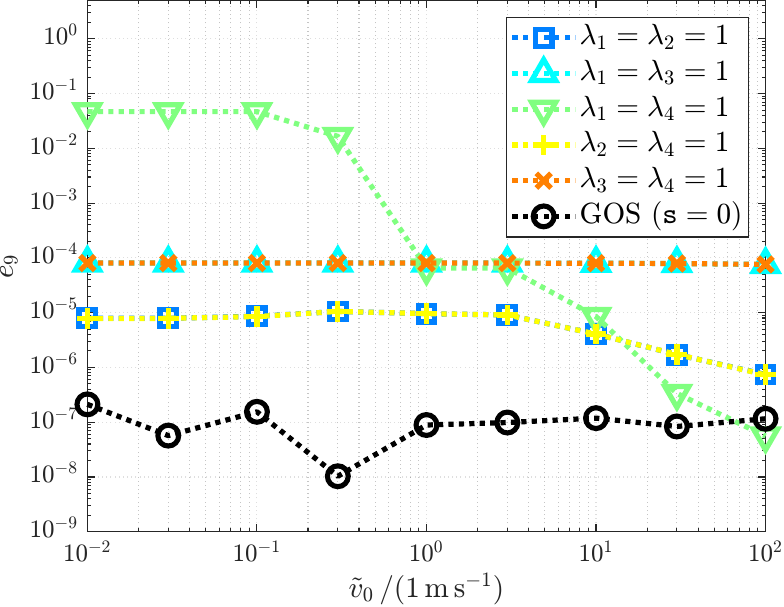}
\caption{Error $\errRoundPP$ \eqref{eqn:error_Round_PP}, with $\Ndigit=9$, versus $\NOUGHT{\Dimensional{\vPP}}$.}
\label{fig:RoundErr_v0}
\end{subfigure}	
\caption{Variability $\Max_i \lam_i / \Min_i \lam_i$ and round-off error $\errRoundPP$ \eqref{eqn:error_Round_PP} with parameters set as in \autoref{fig:lambdas_val_ProjPb} and \autoref{fig:error_Round_PP}, respectively. In \autoref{fig:CoeffVar_h0} and \ref{fig:RoundErr_h0}, the numerical value assigned to $\NOUGHT{\Dimensional{\hPP}}$ is varied by several orders of magnitude, while, in \autoref{fig:CoeffVar_v0} and \ref{fig:RoundErr_v0}, $\NOUGHT{\Dimensional{\vPP}}$ spans different orders of magnitude. \autoref{fig:RoundErr_h0} and \ref{fig:RoundErr_v0} show the error $\errRoundPP$ \eqref{eqn:error_Round_PP}, with $\Ndigit=9$, as representative of $\errRoundPP$ \eqref{eqn:error_Round_PP}, with $\Ndigit \in \{0,1,2,\dots,20\}$, provided in \autoref{fig:en_Round_PP}.}
\label{fig:ParamSensitivity_h0_v0}
\end{figure}	

\subsection{Generalized Optimal Scaling $(\SMALLscaleIDX=1,2,3,4)$: Results}
\label{sec:OSC_ProgPb}

	Finally, we show that, in the case of the projectile model, the Generalized Optimal Scaling method is able to provide a threshold on values of parameters that allow discarding terms \changegreen{from} equation \eqref{eqn:dimensioless_ODE_projectile} without significantly modifying its solution.
				
	By applying the $\OSCmin$ approach with a fixed $\SMALLscaleIDX \in \{1,2,3,4\}$ to the projectile model, one aims to identify a quantitative criterion for allowing the approximation of $\xPPdl$ in \eqref{eqn:dimensioless_ODE_projectile} by the solutions $\xPPdl=\xPPdl_\SMALLscaleIDX$ of \eqref{eqn:dimensioless_ODE_projectile} with $\lam_\SMALLscaleIDX=0$, i.e.,
		
\begin{equation}
\xPPdl_\SMALLscaleIDX(\tPPdl) =
\begin{cases}
\lam_4 \, \tPPdl \, + \, \lam_3, 
\quad \forall \tPPdl \ge 0,
& \IF \quad \SMALLscaleIDX=1, 
\\
(-\lam_1/2) \, \tPPdl^2 \, + \, \lam_4 \, \tPPdl \, + \, \lam_3,
\quad \forall \tPPdl \ge 0,
& \IF \quad \SMALLscaleIDX=2,
\\
\mbox{solution to} \quad
\frac{d^2\xPPdl_\SMALLscaleIDX}{d\tPPdl^2} = - \frac{\lam_1}{(1 + \lam_2 \, \xPPdl_\SMALLscaleIDX)^2},
\quad \forall \tPPdl \ge 0,
\quad \xPPdl_\SMALLscaleIDX(0) = 0, 
\quad \frac{d\xPPdl_\SMALLscaleIDX}{d\tPPdl}(0) = \lam_4,
& \IF \quad \SMALLscaleIDX=3,
\\
\mbox{solution to} \quad
\frac{d^2\xPPdl_\SMALLscaleIDX}{d\tPPdl^2} = - \frac{\lam_1}{(1 + \lam_2 \, \xPPdl_\SMALLscaleIDX)^2},
\quad \forall \tPPdl \ge 0,
\quad \xPPdl_\SMALLscaleIDX(0) = \lam_3, 
\quad \frac{d\xPPdl_\SMALLscaleIDX}{d\tPPdl}(0) = 0,
& \IF \quad \SMALLscaleIDX=4.
\end{cases}
\label{eqn:OSC_ProgPb_approx_sol} 
\end{equation}

\noindent In order to quantify the goodness of approximation, we first set the coefficients $\lam_1,\lam_2,\lam_3,\lam_4$ as in \eqref{eqn:OSC_ProgPb_s1_lam_PI}, \eqref{eqn:OSC_ProgPb_s2_lam_PI}, \eqref{eqn:OSC_ProgPb_s3_lam_PI}, \eqref{eqn:OSC_ProgPb_s4_lam_PI} for $\SMALLscaleIDX=1,2,3,4$, respectively, and then, estimate the approximation error as

\begin{align}
\errReducePP & \coloneqq
\frac{ \int_{0}^{\TmaxPP} \! \left| \xPPdl_\SMALLscaleIDX(\tPPdl) - \xPPref(\tPPdl) \right|^2 \, d\tPPdl }
{ \int_{0}^{\TmaxPP} \! \left| \xPPref(\tPPdl) \right|^2 \, d\tPPdl },
\quad \WITH \quad
\TmaxPP \coloneqq \Min \left\{ \TsingRedPP, \TmaxRedPP \right\},
\quad 
\TmaxRedPP \in (0,\infty),
\nonumber \\
\TsingRedPP & \coloneqq
\begin{cases}
\Min
\left\{ \tPPdl \ge 0 \, : \, \left| 1 \, + \, \lam_2 \, \xPPref(\tPPdl) \right| \le 10^{-2} \right\},
& \IF \quad \SMALLscaleIDX=1,2, \\
\Min
\left\{ \tPPdl \ge 0 \, : 
\, \left| 1 \, + \, \lam_2 \, \xPPdl_\SMALLscaleIDX(\tPPdl) \right| \le 10^{-2}
\, \lor 
\, \left| 1 \, + \, \lam_2 \, \xPPref(\tPPdl) \right| \le 10^{-2} \right\},
& \IF \quad \SMALLscaleIDX=3,4, \\
\end{cases}
\label{eqn:OSC_ProgPb_approx_error}
\end{align}

\noindent where $\xPPdl_\SMALLscaleIDX(\tPPdl)$ is the approximating solution \eqref{eqn:OSC_ProgPb_approx_sol}, while $\xPPref(\tPPdl)$ is the numerical solution to \eqref{eqn:dimensioless_ODE_projectile} provided by the Matlab solver $\odeFF$ \cite{ShampineReicheltODEMatlab}, that has been taken as a reference since it is the solution one aims to approximate. If $\SMALLscaleIDX=3,4$, the approximating solution \eqref{eqn:OSC_ProgPb_approx_sol} is also numerically computed by means of the Matlab solver $\odeFF$ \cite{ShampineReicheltODEMatlab}. A relative error tolerance of $10^{-6}$ and an absolute error tolerance of $10^{-9}$ have been required for all the computations performed by $\odeFF$ \cite{ShampineReicheltODEMatlab}, i.e., $3$ orders of magnitude smaller than the default tolerances. The Simpson's quadrature rule \cite{DG.2024.simps} has been used to compute the integrals in \eqref{eqn:OSC_ProgPb_approx_error}. The time $\TmaxPP$ in \eqref{eqn:OSC_ProgPb_approx_error} is chosen with the aim of stopping the numerical integration just before singularities are attained in \eqref{eqn:dimensioless_ODE_projectile} or \eqref{eqn:OSC_ProgPb_approx_sol}. Such irregularities appear when the computed solutions of \eqref{eqn:dimensioless_ODE_projectile} or \eqref{eqn:OSC_ProgPb_approx_sol} are negative, i.e., equal to $-1/\lam_2<0$. In other words, the time interval $[0,\TmaxPP]$ in \eqref{eqn:OSC_ProgPb_approx_error} allows investigating the dynamics of the projectile further than it reaches the ground level. The resulting error $\errReducePP$ \eqref{eqn:OSC_ProgPb_approx_error} is shown in \autoref{fig:OSC_approx_err_lambdas_small} as a function of $\PI_1$ and $\PI_2$ \eqref{eqn:PI_ProjPb}. 

	As explained in \autoref{sec:OptScal_MinimalParam}, $\OSCmin$ identifies \eqref{eqn:regime_unity_GOS}-\eqref{eqn:threshold_GOS} as conditions for a large difference in magnitudes between $\lam_\SMALLscaleIDX$ and all the remaining $\lam_i$, with $i \neq \SMALLscaleIDX \in \{1,\dots,\NUMlam\}$. Such requirements can be written for some $\DeltaLam \in (-\infty,-1]$ and $\GammaLam \in (0,1]$ as
	
\begin{equation}
\lam_\SMALLscaleIDX \le 10^{\DeltaLam}
\quad \wedge \quad
10^{-\GammaLam} < \lam_i < 10^{\GammaLam},
\quad
\forall i \neq \SMALLscaleIDX.
\label{eqn:scale_sep_delta_gamma}
\end{equation}
	
\noindent By using \eqref{eqn:OSC_ProgPb_s1_lam_PI}, \eqref{eqn:OSC_ProgPb_s2_lam_PI}, \eqref{eqn:OSC_ProgPb_s3_lam_PI}, \eqref{eqn:OSC_ProgPb_s4_lam_PI} for $\SMALLscaleIDX=1,2,3,4$, respectively, the inequalities \eqref{eqn:scale_sep_delta_gamma} give

\begin{align}
& \lam_\SMALLscaleIDX = \PI_1^{-2} \, \PI_2^{1/2} \le 10^{\DeltaLam} 
&& \wedge && -2\GammaLam < \log_{10}(\PI_2) < 2\GammaLam, 
&& \IF \quad \SMALLscaleIDX=1,
\nonumber \\
& \lam_\SMALLscaleIDX = \PI_1^{1/3} \, \PI_2^{5/6} \le 10^{\DeltaLam} 
&& \wedge && -3\GammaLam < \log_{10}(\PI_2) - 2\log_{10}(\PI_1) < 3\GammaLam,
&& \IF \quad \SMALLscaleIDX=2,
\nonumber \\
& \lam_\SMALLscaleIDX = \PI_1^{-1/3} \, \PI_2 \le 10^{\DeltaLam} 
&& \wedge && -3\GammaLam < 2\log_{10}(\PI_1) < 3\GammaLam,
&& \IF \quad \SMALLscaleIDX=3,
\nonumber \\
& \lam_\SMALLscaleIDX = \PI_1 \, \PI_2^{-1/4} \le 10^{\DeltaLam} 
&& \wedge && -2\GammaLam < \log_{10}(\PI_2) < 2\GammaLam, 
&& \IF \quad \SMALLscaleIDX=4,
\label{eqn:OSC_ProgPb_lams_small}
\end{align}

\noindent with $\DeltaLam \in (-\infty,-1]$ and $\GammaLam \in (0,1]$. \autoref{fig:OSC_approx_err_lambdas_small} shows that requiring \eqref{eqn:OSC_ProgPb_lams_small} with $\DeltaLam=-2$ and $\GammaLam=1/2$ guarantees a relative error of at most $10\%$, i.e., $\errReducePP \le 10^{-1}$, for all simulated intervals of time $\TmaxPP$ and all $\SMALLscaleIDX=1,2,3,4$. In \autoref{fig:OSC_approx_err_lambdas_small}, the threshold $\errReducePP = 10^{-1}$ is highlighted by white solid lines, while the conditions \eqref{eqn:OSC_ProgPb_lams_small} with $\DeltaLam=-2$ and $\GammaLam=1/2$ are satisfied by the values of $\PI_1,\PI_2$ lying in the areas delimited by magenta dashed lines. Finally, as clearly demonstrated by \autoref{fig:OSC_approx_err_lambdas_small}, the approximation by an asymptotic model is improving, i.e., the value of $\errReducePP$ is decreasing for all tested $\SMALLscaleIDX$ and $\TmaxPP$, by choosing a smaller $\DeltaLam \in (-\infty,-2]$ at the fixed $\GammaLam=1/2$ in \eqref{eqn:OSC_ProgPb_lams_small} to obtain the ranges of appropriate physical parameters for such a model.   

\section{Conclusions \& Discussion}
\label{sec:concl_discuss}

	We reviewed methodologies for scaling of dimensional models and presented a novel approach which we called Generalized Optimal Scaling, or $\OSCmin$. $\OSCmin$ inherits all the features of its predecessors, i.e., Optimal Scaling ($\OSapproach$) and Optimal Scaling with Constraints ($\OSCapproach$), including the ability to quantify conditions for asymptotic models and, in addition, provides control of over-parametrization - the property beneficial for complex multiscale models. Both $\OSapproach$ and $\OSCapproach$ can be recovered from $\OSCmin$ with the special choices of the $\OSCmin$' parameter $\SMALLscaleIDX$. Such a non-negative integer $\SMALLscaleIDX$ is the index of the coefficient $\lam_\SMALLscaleIDX$ vanishing to $0$ when an asymptotic model is considered.	 The $\OSCmin$ method was rigorously formulated and carefully tested on the classical projectile model.

	By setting $\SMALLscaleIDX=0$ in $\OSCmin$ (see \autoref{sec:achieve_OS_MinParam}), one makes use of the Optimal Scaling methodology \cite{RUSCONI2019106944} presented in \autoref{sec:achieve_OS}. For the projectile model, we demonstrated that $\OSapproach$ allows reducing the variability of the involved coefficients by \changered{up to} two more orders of magnitude compared to other scaling approaches traditionally applied to this model. Such a reduction is beneficial for diminishing round-off and numerical errors. Indeed, we showed that the coefficients provided by the Optimal Scaling approach allow reducing round-off errors by \changered{up to} two more orders of magnitude compared to the traditional scaling techniques.
			
	When $\OSCmin$ is applied with $\SMALLscaleIDX \ge 1$ (see \autoref{sec:achieve_OS_MinParam}), one performs the Optimal Scaling with Constraints \cite{RUSCONI2023127756} explained in \autoref{sec:achieve_OSC}. Such an approach provides a threshold on values of parameters that allow discarding terms \changegreen{from} the considered equations, without significantly modifying the solutions. Thus, one can drop some terms and consider a simpler model, if the provided criterion is met. For the projectile model, $\OSCapproach$ is able to identify values of parameters \changegreen{yielding} a relative approximation error of at most $10\%$ for all simulated intervals of time and all $\SMALLscaleIDX \ge 1$.	

	We remark that $\OSCmin$ is applicable to models of arbitrary levels of complexity. \changered{This is illustrated in the ongoing study \cite{ceuca2026effcomplred}, with additional examples forthcoming.} Moreover, a natural future step is an incorporation of $\OSCmin$ in an automated tool to streamline model reduction and scaling, optimizing parameter selection before further modeling steps. This would facilitate the use of $\OSCmin$ in complex, high-dimensional models, even for users with less mathematical expertise. 

\newgeometry{top=0.5cm,bottom=1.5cm,left=0.75cm,right=0.75cm}

\begin{figure}[!hp]
\centering
\includegraphics[scale=1.2]{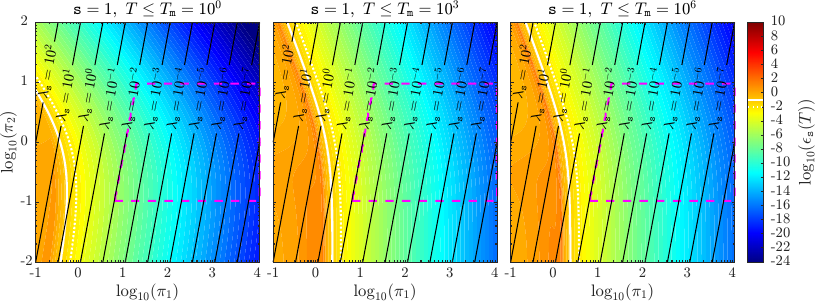} \\
\vspace{0.1in}
\includegraphics[scale=1.2]{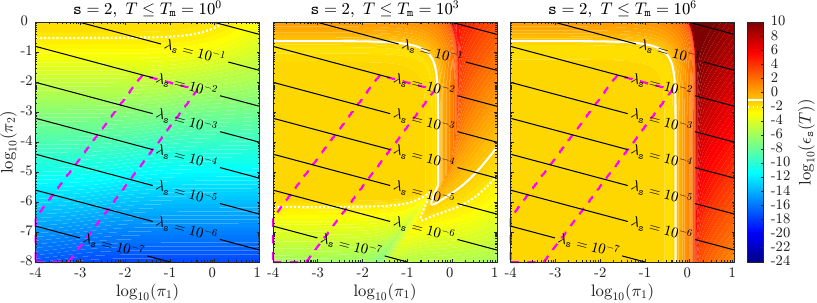} \\
\vspace{0.1in}
\includegraphics[scale=1.2]{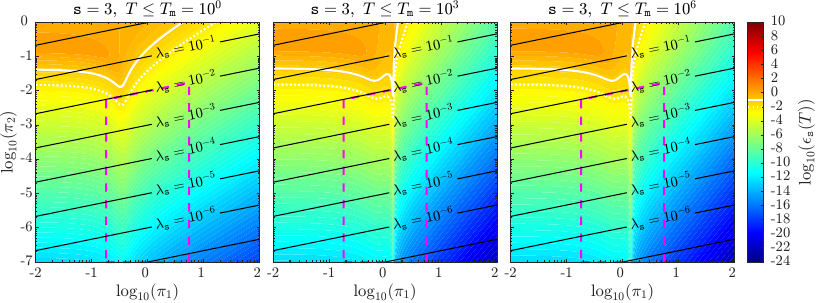} \\
\vspace{0.1in}
\includegraphics[scale=1.2]{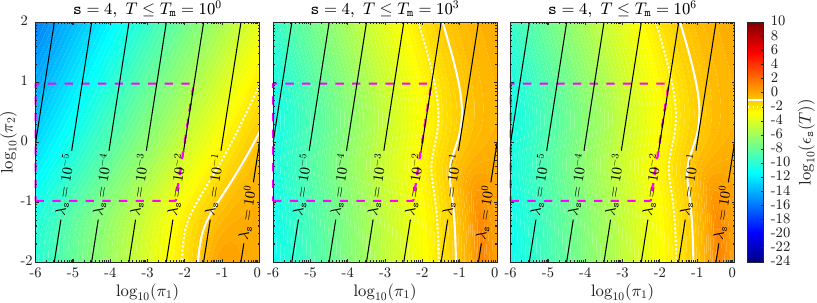} 
\caption{Error $\errReducePP$ \eqref{eqn:OSC_ProgPb_approx_error} and coefficient $\lam_\SMALLscaleIDX$ \eqref{eqn:OSC_ProgPb_lams_small} as functions of $\PI_1,\PI_2$ \eqref{eqn:PI_ProjPb}. Along the white solid and dotted lines, the error $\errReducePP$ is equal to $10^{-1}$ and $10^{-2}$, respectively. The inequalities \eqref{eqn:OSC_ProgPb_lams_small} are met with $\DeltaLam=-2$ and $\GammaLam=1/2$ in the areas delimited by magenta dashed lines.}
\label{fig:OSC_approx_err_lambdas_small}
\end{figure}	

\restoregeometry	

\section*{Acknowledgements}

	This research was supported by MICIU/AEI/10.13039/501100011033 and by ERDF A way for Europe under Grant PID2022-136585NB-C22; funded by MICIU/AEI/10.13039/501100011033 and cofunded by the European Union under Grant   PID2023-146764NB-I00; financed by the Basque Government through ELKARTEK Programme under Grant KK-2024/00062 and through the BERC 2022-2025 program.  
	
	We acknowledge the financial support by the Ministry of Science and Innovation through BCAM Severo Ochoa accreditation CEX2021-001142-S/MICIU/AEI/10.13039/501100011033 and ``PLAN COMPLEMENTARIO MATERIALES AVANZADOS 2022-2025'', PROYECTO N$^{\mbox{\underline{\footnotesize{o}}}}$ 1101288.
	
	The authors acknowledge the financial support received from the grant BCAM-IKUR, funded by the Basque Government, by the IKUR Strategy and by the European Union NextGenerationEU/PRTR.
	
	This work has been possible thanks to the support of the computing infrastructure of Barcelona Supercomputing Center (RES, QHS-2025-1-0027) and DIPC Computer Center.
	
	This work was  partially supported by a grant of the Ministry of Research, Innovation and Digitization, CNCS-UEFISCDI, project number PN-IV-P2-2.1-T-TE-2023-0219, within PNCDI IV.

\appendix

\counterwithin{equation}{section}
\renewcommand{\theequation}{\thesection.\arabic{equation}}

\renewcommand{\thesection}{A}		
\section{Proof of \autoref{th:DIMLESS_quant}}
\label{sec:proofThII}
	
\ThII*

\begin{proof}

	Any quantity of the form \eqref{eqn:ZZdimless} is unitless if and only if its dimensions are 
	
\begin{equation}
\UNITS{\Dimensional{\p}_1}^{\zz_1} \, 
\UNITS{\Dimensional{\p}_2}^{\zz_2} \, 
\cdots \, 
\UNITS{\Dimensional{\p}_{\NUM{\p}}}^{\zz_{\NUM{\p}}} = 1,
\end{equation}

\noindent or, equivalently, by using \eqref{eqn:phys_param_dimen} (\cref{prop:phys_param_dimen}),

\begin{equation}
\UNITS{\Dimensional{\p}_1}^{\zz_1} \, \UNITS{\Dimensional{\p}_2}^{\zz_2} \, \cdots \, \UNITS{\Dimensional{\p}_{\NUM{\p}}}^{\zz_{\NUM{\p}}} 
=
\unit_1^{\sum_{j=1}^{\NUM{\p}} \bp_{1,j} \, \zz_j} \, 
\unit_2^{\sum_{j=1}^{\NUM{\p}} \bp_{2,j} \, \zz_j} \, 
\cdots \, 
\unit_{\NUMunit}^{\sum_{j=1}^{\NUM{\p}} \bp_{\NUMunit,j} \, \zz_j}
= 
1.
\label{eqn:Dimensions_UNITlessQ}
\end{equation}	

\noindent From independence of $\unit_1, \dots, \unit_{\NUMunit}$ assured by \cref{prop:phys_param_dimen}, \eqref{eqn:Dimensions_UNITlessQ} is equivalent to

\begin{equation}
\sum_{j=1}^{\NUM{\p}} \bp_{i,j} \, \zz_j = 0,
\quad \forall i=1,\dots,\NUMunit,
\quad \IdEst \quad
\MM \, \VEC{\zz} = \VEC{0},
\end{equation}

\noindent for $\MM$ defined in \eqref{eqn:matrixMM} and $\VEC{\zz} \coloneqq ( \zz_1, \zz_2, \dots, \zz_{\NUM{\p}} )^{\Transp} \in \mathbb{R}^{\NUM{\p}}$. In other words, any quantity of the form \eqref{eqn:ZZdimless} is dimensionless if and only if the vector $\VEC{\zz}$ of exponents in \eqref{eqn:ZZdimless} belongs to the null space (kernel) of matrix $\MM$ \eqref{eqn:matrixMM}, i.e., $\VEC{\zz} \in \Ker(\MM)$.

\end{proof}

\renewcommand{\thesection}{B}		
\section{Proof of \autoref{th:PI_quantities}}
\label{sec:proofThIII}

\ThIII*

\begin{proof}

	First, we show that \cref{prop:phys_param_indep,prop:phys_param_positive,prop:phys_param_dimen} imply that the parameters $\PI_1,\dots,\PI_{\NUMPI}$ \eqref{eqn:Indep_Dimless_Param} are dimensionless, strictly positive and independent. Then, we conclude the proof of \autoref{th:PI_quantities} by demonstrating that any dimensionless quantity of the form \eqref{eqn:ZZdimless} can be written in terms of such $\PI_1,\dots,\PI_{\NUMPI}$ only, as this means that there are no further independent quantities one can consider.
		
	Given the matrix $\MM$ \eqref{eqn:matrixMM} of exponents in \eqref{eqn:phys_param_dimen} (\cref{prop:phys_param_dimen}), any basis of $\Ker(\MM)$ is composed by $ \NUMPI \coloneqq \Dim(\Ker(\MM)) = \NUM{\p} - \Rank(\MM) $ independent vectors $\VEC{\BB}_1, \dots, \VEC{\BB}_{\NUMPI} \in \Ker(\MM) \subseteq \mathbb{R}^{\NUM{\p}}$. The corresponding $\PI_1,\dots,\PI_{\NUMPI}$ \eqref{eqn:Indep_Dimless_Param} are dimensionless since $\VEC{\BB}_1, \dots, \VEC{\BB}_{\NUMPI} \in \Ker(\MM)$, \changered{as guaranteed by \autoref{th:DIMLESS_quant}.}
	
	Then, by employing the notation \eqref{eqn:notation_dimensional_quant} in \eqref{eqn:Indep_Dimless_Param}, one has
	
\begin{equation}
\PI_k = \Dimensional{\p}_1^{\BB_{k,1}} \, \cdots \, \Dimensional{\p}_{\NUM{\p}}^{\BB_{k,\NUM{\p}}} =
\p_1^{\BB_{k,1}} \, \UNITS{\Dimensional{\p}_1}^{\BB_{k,1}}
\, \cdots \,
\p_{\NUM{\p}}^{\BB_{k,{\NUM{\p}}}} \, \UNITS{\Dimensional{\p}_{\NUM{\p}}}^{\BB_{k,{\NUM{\p}}}},
\quad
\forall k=1,\dots,\NUMPI.
\end{equation}

\noindent As $\PI_1,\dots,\PI_{\NUMPI}$ are dimensionless, i.e., 

\begin{equation}
\UNITS{\PI_k}
=
\UNITS{\Dimensional{\p}_1}^{\BB_{k,1}}
\, \cdots \,
\UNITS{\Dimensional{\p}_{\NUM{\p}}}^{\BB_{k,{\NUM{\p}}}}
= 1,
\quad \forall k=1,\dots,\NUMPI,
\end{equation}

\noindent and $\p_1, \dots, \p_{\NUM{\p}} \in (0,\infty)$ (\cref{prop:phys_param_positive}), it follows the strict positiveness of $\PI_1,\dots,\PI_{\NUMPI}$ \eqref{eqn:Indep_Dimless_Param}, i.e.,

\begin{equation}
\PI_k = \p_1^{\BB_{k,1}} \, \cdots \, \p_{\NUM{\p}}^{\BB_{k,\NUM{\p}}}
\in (0,\infty), \quad \forall k=1,\dots,\NUMPI.
\end{equation}

	Moreover, such $\PI_1,\dots,\PI_{\NUMPI}$ \eqref{eqn:Indep_Dimless_Param} are independent since
	
\begin{equation}
\PI_1^{\bbexp_1} \, \PI_2^{\bbexp_2} \, \cdots \, \PI_{\NUMPI}^{\bbexp_{\NUMPI}} = 1 
\Leftrightarrow 
\bbexp_1 = \bbexp_2 = \cdots = \bbexp_{\NUMPI} = 0.
\label{eqn:PI_INDEP}
\end{equation}	
	
\noindent On the one hand, $\bbexp_1 = \bbexp_2 = \cdots = \bbexp_{\NUMPI} = 0$ trivially implies that $\PI_1^{\bbexp_1} \, \PI_2^{\bbexp_2} \, \cdots \, \PI_{\NUMPI}^{\bbexp_{\NUMPI}} = 1 $. On the other hand, given $\PI_1^{\bbexp_1} \, \PI_2^{\bbexp_2} \, \cdots \, \PI_{\NUMPI}^{\bbexp_{\NUMPI}} = 1$, it follows from \eqref{eqn:Indep_Dimless_Param} that

\begin{equation}
\PI_1^{\bbexp_1} \, \PI_2^{\bbexp_2} \, \cdots \, \PI_{\NUMPI}^{\bbexp_{\NUMPI}}
=
\Dimensional{\p}_1^{\sum_{k=1}^{\NUMPI} \bbexp_k \, \BB_{k,1} } \,
\Dimensional{\p}_2^{\sum_{k=1}^{\NUMPI} \bbexp_k \, \BB_{k,2} } \,
\cdots \,
\Dimensional{\p}_{\NUM{\p}}^{\sum_{k=1}^{\NUMPI} \bbexp_k \, \BB_{k,\NUM{\p}} } \,
=
1, 
\end{equation} 
	
\noindent which implies, from independence of $\Dimensional{\p}_1, \dots, \Dimensional{\p}_{\NUM{\p}}$ assured by \cref{prop:phys_param_indep},

\begin{equation}
\sum_{k=1}^{\NUMPI} \bbexp_k \, \BB_{k,j} = 0,
\quad
\forall j=1,\dots,\NUM{\p},
\quad \IdEst \quad
\bbexp_1 \, \VEC{\BB}_1
\, + \,
\bbexp_2 \, \VEC{\BB}_2
\, + \, \cdots \,
\bbexp_{\NUMPI} \, \VEC{\BB}_{\NUMPI}
=
\VEC{0}.
\label{eqn:PI_BB_INDEP_}
\end{equation} 

\noindent As $\VEC{\BB}_1, \dots, \VEC{\BB}_{\NUMPI}$ constitute a basis of $\Ker(\MM)$, they are linearly independent vectors. Thus, \eqref{eqn:PI_BB_INDEP_} gives $\bbexp_1 = \bbexp_2 = \cdots = \bbexp_{\NUMPI} = 0$ and \eqref{eqn:PI_INDEP} holds, showing that $\PI_1,\dots,\PI_{\NUMPI}$ \eqref{eqn:Indep_Dimless_Param} are independent.
	
	Finally, any dimensionless quantity of the form $\Dimensional{\p}_1^{\zz_1} \, \Dimensional{\p}_2^{\zz_2} \, \cdots \, \Dimensional{\p}_{\NUM{\p}}^{\zz_{\NUM{\p}}} $, with $\zz_1, \zz_2, \dots, \zz_{\NUM{\p}} \in \mathbb{R}$, is such that $\VEC{\zz} \coloneqq ( \zz_1, \zz_2, \dots, \zz_{\NUM{\p}} )^{\Transp} \in \Ker(\MM)$, \changered{as guaranteed by \autoref{th:DIMLESS_quant}.} Then, the vector $\VEC{\zz}$ can be written as a linear combination of the vectors $\VEC{\BB}_1, \dots, \VEC{\BB}_{\NUMPI}$ composing a basis of $\Ker(\MM)$. By using \eqref{eqn:Indep_Dimless_Param}, it follows that

\begin{equation}
\Dimensional{\p}_1^{\zz_1} \,
\Dimensional{\p}_2^{\zz_2} \,
\cdots \,
\Dimensional{\p}_{\NUM{\p}}^{\zz_{\NUM{\p}}} 
=
\Dimensional{\p}_1^{\sum_{k=1}^{\NUMPI} \bbexp_k \, \BB_{k,1} } \,
\Dimensional{\p}_2^{\sum_{k=1}^{\NUMPI} \bbexp_k \, \BB_{k,2} } \,
\cdots \,
\Dimensional{\p}_{\NUM{\p}}^{\sum_{k=1}^{\NUMPI} \bbexp_k \, \BB_{k,\NUM{\p}} } \,
=
\PI_1^{\bbexp_1} \, \PI_2^{\bbexp_2} \, \cdots \, \PI_{\NUMPI}^{\bbexp_{\NUMPI}}, 
\end{equation} 

\noindent for some $\bbexp_1,\dots,\bbexp_{\NUMPI} \in \mathbb{R}$. In conclusion, any dimensionless quantity of the form \eqref{eqn:ZZdimless} can be written in terms of $\PI_1,\dots,\PI_{\NUMPI}$ \eqref{eqn:Indep_Dimless_Param} only.
\end{proof}

\renewcommand{\thesection}{C}		
\section{Proof of \autoref{th:lamb_pow_PI}}
\label{sec:proofThIV}

\ThIV*

\begin{proof}

	As it does for $\PI_1,\dots,\PI_{\NUMPI}$ \eqref{eqn:Indep_Dimless_Param} (see the proof of \autoref{th:PI_quantities} in \autoref{sec:proofThIII}), \cref{prop:phys_param_positive} ensures that the numerical values $\Klam_i$ and $\thcc_j$ of all $\Dimensional{\Klam}_i$ in \eqref{eqn:Klam_powlaw_pp} and $\Dimensional{\thcc}_j$ in \eqref{eqn:char_const_pow_PhysParam} are strictly positive. Then, it is possible to plug \eqref{eqn:Klam_powlaw_pp} and \eqref{eqn:char_const_pow_PhysParam} into \eqref{eqn:lam_coeff_I} (see \cref{prop:lam_coeff_I,prop:Klam_powlaw_pp,prop:char_const}). This uncovers that each of the dimensionless coefficients $\lam_1,\dots,\lam_{\NUMlam}$ \eqref{eqn:lam_coeff_0} can be written as a power-law monomial of the physical parameters $\Dimensional{\p}_1,\dots,\Dimensional{\p}_{\NUM{\p}}$ only, i.e., in the form \eqref{eqn:ZZdimless}:
	
\begin{equation}
\exists \, \VEC{\LL}_i \coloneqq
( \LL_{i,1}, \LL_{i,2}, \dots, \LL_{i,\NUM{\p}} )^{\Transp}
\in \mathbb{R}^{\NUM{\p}}
\quad \ST \quad
\lam_i =
\Dimensional{\p}_1^{\LL_{i,1}} \, 
\Dimensional{\p}_2^{\LL_{i,2}} \, 
\cdots \, 
\Dimensional{\p}_{\NUM{\p}}^{\LL_{i,\NUM{\p}}},
\label{eqn:lambdas_pow_param_p}
\end{equation}

\noindent being $\LL_{i,j}$ the $j$-th component of the vector $\VEC{\LL}_i$ for any $i=1,\dots,\NUMlam$ and $j=1,\dots,\NUM{\p}$.

	Each of the coefficients $\lam_1,\dots,\lam_{\NUMlam}$ in \eqref{eqn:lambdas_pow_param_p} is dimensionless and, thus, \changered{\autoref{th:DIMLESS_quant} guarantees that} each of the vectors $\VEC{\LL}_1,\dots,\VEC{\LL}_{\NUMlam}$ of exponents in \eqref{eqn:lambdas_pow_param_p} belongs to the kernel of the matrix $\MM$ \eqref{eqn:matrixMM}, i.e., $\VEC{\LL}_1,\dots,\VEC{\LL}_{\NUMlam} \in \Ker(\MM)$. \changered{Recalling \autoref{th:PI_quantities},} one considers the basis' vectors $\VEC{\BB}_1, \dots, \VEC{\BB}_{\NUMPI}$ \eqref{eqn:Indep_Dimless_Param} for $\Ker(\MM)$. Then, any vector $\VEC{\LL}_i$ \eqref{eqn:lambdas_pow_param_p} belonging to $\Ker(\MM)$ can be written (in a unique way) as a linear combination of such $\VEC{\BB}_1, \dots, \VEC{\BB}_{\NUMPI}$ \eqref{eqn:Indep_Dimless_Param}, i.e.,
		
\begin{equation}
\forall \, \VEC{\LL}_i \, \eqref{eqn:lambdas_pow_param_p}
\quad \exists! \, 
\VEC{\AAcoeff}_i \coloneqq
(\AAcoeff_{i,1},\AAcoeff_{i,2},\dots,\AAcoeff_{i,\NUMPI})^{\Transp} \in \mathbb{R}^{\NUMPI}
\quad \ST \quad
\VEC{\LL}_i = \sum_{k=1}^{\NUMPI} \AAcoeff_{i,k} \, \VEC{\BB}_k,
\label{eqn:LL_lin_combo_ZZ}
\end{equation}	 

\noindent where $\AAcoeff_{i,k}$ is the $k$-th component of the vector $\VEC{\AAcoeff}_i$, with $i=1,\dots,\NUMlam$ and $k=1,\dots,\NUMPI$. 

	By plugging \eqref{eqn:LL_lin_combo_ZZ} into \eqref{eqn:lambdas_pow_param_p} and, then, \changered{using \eqref{eqn:Indep_Dimless_Param} (\autoref{th:PI_quantities}),} it follows that each of the coefficients $\lam_1,\dots,\lam_{\NUMlam}$ can be written as a power-law monomial of the independent dimensionless parameters $\PI_1,\dots,\PI_{\NUMPI}$ \eqref{eqn:Indep_Dimless_Param} only, i.e.,
	
\begin{equation}
\lam_i 
= 
\Dimensional{\p}_1^{\LL_{i,1}} \, 
\cdots \, 
\Dimensional{\p}_{\NUM{\p}}^{\LL_{i,\NUM{\p}}}
=
\Dimensional{\p}_1^{\sum_{k=1}^{\NUMPI} \AAcoeff_{i,k} \, \BB_{k,1}} \, 
\cdots \,
\Dimensional{\p}_{\NUM{\p}}^{\sum_{k=1}^{\NUMPI} \AAcoeff_{i,k} \, \BB_{k,\NUM{\p}}} 
=
\PI_1^{\AAcoeff_{i,1}} \, 
\cdots \, 
\PI_{\NUMPI}^{\AAcoeff_{i,\NUMPI}},
\label{eqn:lam_PI_AA}
\end{equation}		
		
\noindent being $\LL_{i,j}$ and $\BB_{k,j}$ the $j$-th components of the vectors $\VEC{\LL}_i$ and $\VEC{\BB}_k$, respectively, for any $i=1,\dots,\NUMlam$, $j=1,\dots,\NUM{\p}$ and $k=1,\dots,\NUMPI$. We remark that the values of the exponents $\AAcoeff_{i,k}$ in \eqref{eqn:lam_PI_AA} are not explicitly known and, in this sense, we state that they are unknowns of scaling procedures.		
\end{proof}

\renewcommand{\thesection}{D}		
\section{Proof of \autoref{th:sysYY}}
\label{sec:proofThV}

\ThV*

\begin{proof}
	
	We consider the vectors $\VEC{\LL}_1,\dots,\VEC{\LL}_{\NUMlam} \in \mathbb{R}^{\NUM{\p}}$ of exponents in \eqref{eqn:lambdas_pow_param_p}. As explained in the proof of \autoref{th:lamb_pow_PI} (\autoref{sec:proofThIV}), we have that \changered{\namecrefs{th:DIMLESS_quant} \ref{th:DIMLESS_quant} and \ref{th:PI_quantities} ensure} \eqref{eqn:LL_lin_combo_ZZ} for each vector $\VEC{\LL}_1,\dots,\VEC{\LL}_{\NUMlam}$. This can be understood as follows. Let us denote $\BBmatrix$ as the matrix whose columns are the vectors $\VEC{\BB}_1,\dots,\VEC{\BB}_{\NUMPI}$ \eqref{eqn:Indep_Dimless_Param}, i.e.,
	
\begin{equation}
\BBmatrix \coloneqq
\begin{pmatrix}
\VEC{\BB}_1 & \VEC{\BB}_2 & \cdots & \VEC{\BB}_{\NUMPI} 
\end{pmatrix}
\in \mathbb{R}^{\NUM{\p} \times \NUMPI},
\quad \WITH \quad
\NUMPI<\NUM{\p}.
\label{eqn:BB_def}
\end{equation}

\noindent Each of the vectors $\VEC{\LL}_1,\dots,\VEC{\LL}_{\NUMlam}$ in \eqref{eqn:LL_lin_combo_ZZ} belongs to the column space of such a matrix $\BBmatrix$ \eqref{eqn:BB_def}, which is equivalent to the orthogonal complement of the null space (kernel) of $\BBmatrix^{\Transp}$, i.e.,

\begin{equation}
\VEC{\LL}_i \in 
\Colsp(\BBmatrix) \equiv \left( \Ker(\BBmatrix^{\Transp}) \right)^{\Ortho},
\quad \forall i=1,\dots,\NUMlam.
\label{eqn:LLvec_OrthoCompl}
\end{equation}

\noindent One has $\NUMPI=\NUM{\p}-\Rank(\MM)<\NUM{\p}$ in \eqref{eqn:BB_def} owing to $\Rank(\MM)\ge1$. The latter inequality follows from \cref{prop:phys_param_dimen} guaranteeing that $\MM$ \eqref{eqn:matrixMM} cannot be a matrix whose entries are all equal to zero. Then, since $\Rank(\BBmatrix^{\Transp})=\Rank(\BBmatrix)=\NUMPI$ and $\Dim \left( \Ker(\BBmatrix^{\Transp}) \right) = \NUM{\p}-\NUMPI > 0$, we can consider $\NUM{\p}-\NUMPI$ vectors $\VEC{\VV}_1,\dots,\VEC{\VV}_{\NUM{\p}-\NUMPI} \in \mathbb{R}^{\NUM{\p}}$ to form a basis of $\Ker(\BBmatrix^{\Transp})$. The assertion \eqref{eqn:LLvec_OrthoCompl} is equivalently stated as each of $\VEC{\LL}_1,\dots,\VEC{\LL}_{\NUMlam}$ being perpendicular to each of the basis' vectors $\VEC{\VV}_1,\dots,\VEC{\VV}_{\NUM{\p}-\NUMPI}$, i.e.,

\begin{equation}
\VVmatrix^{\Transp} \, \VEC{\LL}_i = \VEC{0},
\quad
\VVmatrix \coloneqq
\begin{pmatrix}
\VEC{\VV}_1 & \VEC{\VV}_2 & \cdots & \VEC{\VV}_{\NUM{\p}-\NUMPI}
\end{pmatrix}
\in \mathbb{R}^{\NUM{\p} \times (\NUM{\p}-\NUMPI)},
\quad 
\forall i=1,\dots,\NUMlam.
\label{eqn:LLvec_eqVV}
\end{equation}

\noindent \cref{prop:phys_param_positive} guarantees that the numerical values $\Klam_i$ and $\thcc_j$ of all $\Dimensional{\Klam}_i$ in \eqref{eqn:Klam_powlaw_pp} and $\Dimensional{\thcc}_j$ in \eqref{eqn:char_const_pow_PhysParam} are strictly positive. Then, one can plug \eqref{eqn:Klam_powlaw_pp} and \eqref{eqn:char_const_pow_PhysParam} into \eqref{eqn:lam_coeff_I} to get 

\begin{equation}
\lam_i =
\Dimensional{\p}_1^{\sum_{j=1}^{\NUM{\var}} \CC_{i,j} \, \TT_{1,j} \, + \, \DD_{i,1}} \, 
\Dimensional{\p}_2^{\sum_{j=1}^{\NUM{\var}} \CC_{i,j} \, \TT_{2,j} \, + \, \DD_{i,2}} \, 
\cdots \, 
\Dimensional{\p}_{\NUM{\p}}^{\sum_{j=1}^{\NUM{\var}} \CC_{i,j} \, \TT_{\NUM{\p},j} \, + \, \DD_{i,\NUM{\p}}},
\quad \forall i = 1,\dots,\NUMlam.
\label{eqn:lam_exponents_}
\end{equation}

\noindent By using \eqref{eqn:lam_exponents_} and the independence of physical parameters $\Dimensional{\p}_1,\dots,\Dimensional{\p}_{\NUM{\p}}$ (\cref{prop:phys_param_indep}), the vectors $\VEC{\LL}_1,\dots,\VEC{\LL}_{\NUMlam}$ of exponents in \eqref{eqn:lambdas_pow_param_p} can be computed in terms of $\VEC{\TT}_1,\dots,\VEC{\TT}_{\NUM{\var}}$ \eqref{eqn:vecYY_def}:

\begin{equation}
\VEC{\LL}_i
=
\CC_{i,1} \, \VEC{\TT}_1
\, + \,
\CC_{i,2} \, \VEC{\TT}_2
\, + \, \cdots \, + \,
\CC_{i,\NUM{\var}}	\, \VEC{\TT}_{\NUM{\var}}
\, + \,
\VEC{\DD}_i,
\quad
\VEC{\DD}_i \coloneqq ( \DD_{i,1}, \DD_{i,2}, \dots, \DD_{i,{\NUM{\p}}} )^{\Transp},
\quad 
\forall i=1,\dots,\NUMlam.
\label{eqn:LinSys1_TT}
\end{equation}

\noindent Inserting $\VEC{\LL}_i$ \eqref{eqn:LinSys1_TT} into \eqref{eqn:LLvec_eqVV} gives

\begin{equation}
\CC_{i,1} \, \VVmatrix^{\Transp} \, \VEC{\TT}_1
\, + \,
\CC_{i,2} \, \VVmatrix^{\Transp} \, \VEC{\TT}_2
\, + \, \cdots \, + \,
\CC_{i,\NUM{\var}}	\, \VVmatrix^{\Transp} \, \VEC{\TT}_{\NUM{\var}}
= - \VVmatrix^{\Transp} \, \VEC{\DD}_i,
\quad \forall i=1,\dots,\NUMlam,
\end{equation}

\noindent that reads in terms of the vector $\VEC{\TT}$ \eqref{eqn:vecYY_def} of unknowns as

\begin{equation}
\matrixTT \, \VEC{\TT} = \vectorRHStt,
\quad
\matrixTT \in \mathbb{R}^{\NUM{\TT} \times (\NUM{\p} \, \NUM{\var})},
\quad
\vectorRHStt \in \mathbb{R}^{\NUM{\TT}},
\quad
\NUM{\TT} \coloneqq \NUMlam \, ( \NUM{\p} - \NUMPI ),
\label{eqn:LinSys_0}
\end{equation}

\begin{equation}
\matrixTT =
\begin{pmatrix}
\CC_{1,1} \, \VVmatrix^{\Transp} & \CC_{1,2} \, \VVmatrix^{\Transp} &
\cdots & \CC_{1,\NUM{\var}} \, \VVmatrix^{\Transp} \\
\CC_{2,1} \, \VVmatrix^{\Transp} & \CC_{2,2} \, \VVmatrix^{\Transp} &
\cdots & \CC_{2,\NUM{\var}} \, \VVmatrix^{\Transp} \\
\vdots & \vdots & \ddots & \vdots \\
\CC_{\NUMlam,1} \, \VVmatrix^{\Transp} & \CC_{\NUMlam,2} \, \VVmatrix^{\Transp} &
\cdots & \CC_{\NUMlam,\NUM{\var}} \, \VVmatrix^{\Transp} 
\end{pmatrix},
\quad 
\vectorRHStt =
- \begin{pmatrix}
\VVmatrix^{\Transp} \, \VEC{\DD}_1 \\
\VVmatrix^{\Transp} \, \VEC{\DD}_2 \\
\vdots \\
\VVmatrix^{\Transp} \, \VEC{\DD}_{\NUMlam} 
\end{pmatrix},
\label{eqn:TTtau_def}
\end{equation}

\noindent with $\VVmatrix$ and $\VEC{\DD}_1,\dots,\VEC{\DD}_{\NUMlam}$ defined as in \eqref{eqn:LLvec_eqVV} and \eqref{eqn:LinSys1_TT}, respectively. By using any solution $\VEC{\TT} \coloneqq (\VEC{\TT}_1,\VEC{\TT}_2,\dots,\VEC{\TT}_{\NUM{\var}})^{\Transp}$ to \eqref{eqn:LinSys_0}-\eqref{eqn:TTtau_def} in \eqref{eqn:LinSys1_TT}, the resulting vectors $\VEC{\LL}_1,\dots,\VEC{\LL}_{\NUMlam}$ \eqref{eqn:LinSys1_TT} fulfill \eqref{eqn:LLvec_eqVV} that implies \eqref{eqn:LLvec_OrthoCompl}. In other words, such $\VEC{\LL}_1,\dots,\VEC{\LL}_{\NUMlam}$ belong to the column space of the matrix $\BBmatrix$ \eqref{eqn:BB_def} generated by the basis' vectors $\VEC{\BB}_1,\dots,\VEC{\BB}_{\NUMPI}$ \eqref{eqn:Indep_Dimless_Param}. Thus, each of the computed $\VEC{\LL}_1,\dots,\VEC{\LL}_{\NUMlam}$ admits the unique vector $\VEC{\AAcoeff}_i$ \eqref{eqn:LL_lin_combo_ZZ} of coordinates with respect to the basis $\VEC{\BB}_1,\dots,\VEC{\BB}_{\NUMPI}$ \eqref{eqn:Indep_Dimless_Param}. That is to say, each of the $\NUMlam$ linear systems written in \eqref{eqn:LL_lin_combo_ZZ}, i.e., $\BBmatrix \, \VEC{\AAcoeff}_i = \VEC{\LL}_i$, with $i=1,\dots,\NUMlam$, can be solved providing the unique solution $\VEC{\AAcoeff}_i \in \mathbb{R}^{\NUMPI}$:

\begin{equation}
\VEC{\AAcoeff}_i = \BBmatrix^{\PseudoInv} \, \VEC{\LL}_i,
\quad
\BBmatrix^{\PseudoInv} \coloneqq
\left( \BBmatrix^{\Transp} \, \BBmatrix \, \right)^{\Inv}
\BBmatrix^{\Transp} \in \mathbb{R}^{\NUMPI \times \NUM{\p}},
\quad 
\forall i=1,\dots,\NUMlam,
\label{eqn:AAi_solved}
\end{equation}

\noindent where $\BBmatrix^{\PseudoInv}$ is the pseudo-inverse of the rectangular matrix $\BBmatrix$ \eqref{eqn:BB_def}. The inverse matrix $\left( \BBmatrix^{\Transp} \, \BBmatrix \, \right)^{\Inv} \in \mathbb{R}^{\NUMPI \times \NUMPI}$ in \eqref{eqn:AAi_solved} is well defined because $\Rank(\BBmatrix)=\NUMPI<\NUM{\p}$ implies that $\Rank(\BBmatrix^{\Transp} \, \BBmatrix ) = \NUMPI$. By plugging $\VEC{\LL}_i$ \eqref{eqn:LinSys1_TT} into \eqref{eqn:AAi_solved}, it follows

\begin{equation}
\CC_{i,1} \, \BBmatrix^{\PseudoInv} \, \VEC{\TT}_1
\, + \,
\CC_{i,2} \, \BBmatrix^{\PseudoInv} \, \VEC{\TT}_2
\, + \, \cdots \, + \,
\CC_{i,\NUM{\var}}	\, \BBmatrix^{\PseudoInv} \, \VEC{\TT}_{\NUM{\var}}
= 
\VEC{\AAcoeff}_i 
\, - \, 
\BBmatrix^{\PseudoInv} \, \VEC{\DD}_i,
\quad \forall i=1,\dots,\NUMlam,
\end{equation}

\noindent that reads in terms of the vector $\VEC{\TT}$ \eqref{eqn:vecYY_def} of unknowns as

\begin{equation}
\aamatrix \, \VEC{\TT} = \VEC{\AAcoeff} \, - \, \aavec,
\quad
\aamatrix \in \mathbb{R}^{(\NUMlam \, \NUMPI)\times(\NUM{\p} \, \NUM{\var})},
\quad
\VEC{\AAcoeff},\aavec \in \mathbb{R}^{\NUMlam \, \NUMPI}, 
\label{eqn:LinSys_I}
\end{equation}

\begin{equation}
\aamatrix =
\begin{pmatrix}
\CC_{1,1} \, \BBmatrix^{\PseudoInv} & \CC_{1,2} \, \BBmatrix^{\PseudoInv} &
\cdots & \CC_{1,\NUM{\var}} \, \BBmatrix^{\PseudoInv} \\
\CC_{2,1} \, \BBmatrix^{\PseudoInv} & \CC_{2,2} \, \BBmatrix^{\PseudoInv} &
\cdots & \CC_{2,\NUM{\var}} \, \BBmatrix^{\PseudoInv} \\
\vdots & \vdots & \ddots & \vdots \\
\CC_{\NUMlam,1} \, \BBmatrix^{\PseudoInv} & \CC_{\NUMlam,2} \, \BBmatrix^{\PseudoInv} &
\cdots & \CC_{\NUMlam,\NUM{\var}} \, \BBmatrix^{\PseudoInv}
\end{pmatrix},
\quad 
\aavec =
\begin{pmatrix}
\BBmatrix^{\PseudoInv} \, \VEC{\DD}_1 \\
\BBmatrix^{\PseudoInv} \, \VEC{\DD}_2 \\
\vdots \\
\BBmatrix^{\PseudoInv} \, \VEC{\DD}_{\NUMlam} 
\end{pmatrix},
\label{eqn:aamat_aavec_def}
\end{equation}

\noindent where $\VEC{\AAcoeff}$ is defined as in \eqref{eqn:AAdef}, being the vector $\in \mathbb{R}^{\NUMlam \, \NUMPI}$ whose entries are the exponents $\AAcoeff_{i,k}$ in \eqref{eqn:lambdas_fun_PI}. The matrix $\BBmatrix^{\PseudoInv}$ and the vectors $\VEC{\DD}_1,\dots,\VEC{\DD}_{\NUMlam}$ are defined as in \eqref{eqn:AAi_solved} and \eqref{eqn:LinSys1_TT}, respectively. Both linear systems \eqref{eqn:LinSys_0}-\eqref{eqn:TTtau_def} and \eqref{eqn:LinSys_I}-\eqref{eqn:aamat_aavec_def} must be satisfied by the vector $\VEC{\TT}$ \eqref{eqn:vecYY_def} of unknowns, \changegreen{yielding} the matrix equation $\matrixYY \, \VEC{\TT} = \vectorRHSyy$ \eqref{eqn:SyS_TT} for $\VEC{\TT}$, with $\matrixYY$ and $\vectorRHSyy$ given by 

\begin{equation}
\matrixYY =
\begin{pmatrix} 
\aamatrix \\
\matrixTT
\end{pmatrix} \in \mathbb{R}^{(\NUMlam \, \NUM{\p}) \times (\NUM{\p} \, \NUM{\var})},
\quad
\vectorRHSyy =
\begin{pmatrix} 	 
\VEC{\AAcoeff} - \aavec \\
\vectorRHStt
\end{pmatrix} \in \mathbb{R}^{\NUMlam \, \NUM{\p}}.
\label{eqn:matYY_vecYY_def}
\end{equation}

\noindent The entries of $\matrixYY$ \eqref{eqn:matYY_vecYY_def} are explicitly known real numbers, as given by those of $\aamatrix$ \eqref{eqn:aamat_aavec_def} and $\matrixTT$ \eqref{eqn:TTtau_def}. Since the components of $\aavec$ \eqref{eqn:aamat_aavec_def} and $\vectorRHStt$ \eqref{eqn:TTtau_def} are explicitly known, the entries of $\vectorRHSyy$ \eqref{eqn:matYY_vecYY_def} are explicit affine functions of the exponents $\AAcoeff_{i,k}$ in \eqref{eqn:lambdas_fun_PI}, collected by the vector $\VEC{\AAcoeff}$ \eqref{eqn:AAdef}.

\end{proof}

\renewcommand{\thesection}{E}		
\section{Proof of \autoref{th:sysAA}}
\label{sec:proofThVI}

\ThVI*

\begin{proof}

	Bearing \cref{prop:char_const}, there exists at least a vector $\VEC{\TT} \coloneqq \left( \TT_{1,1}, \dots,\TT_{\NUM{\p},\NUM{\var}} \right)^{\Transp} \in \mathbb{R}^{\NUM{\p} \, \NUM{\var}}$ of exponents that ensures \eqref{eqn:char_const_pow_PhysParam}. Under \cref{prop:lam_coeff_I,prop:Klam_powlaw_pp,prop:char_const,prop:phys_param_indep,prop:phys_param_positive,prop:phys_param_dimen}, \autoref{th:sysYY} guarantees that such a vector $\VEC{\TT}$ \eqref{eqn:vecYY_def} satisfies the linear system $\matrixYY \, \VEC{\TT} = \vectorRHSyy$ \eqref{eqn:SyS_TT}. Then, there must be at least a solution to such a matrix equation \eqref{eqn:SyS_TT}. Equivalently, the vector $\vectorRHSyy \in \mathbb{R}^{\NUMlam \, \NUM{\p}}$ must belong to the column space of the matrix $\matrixYY \in \mathbb{R}^{(\NUMlam \, \NUM{\p}) \times (\NUM{\p} \, \NUM{\var})}$, which is the orthogonal complement of the null space (kernel) of $\matrixYY^{\Transp}$, i.e.,

\begin{equation}
\vectorRHSyy \in \Colsp(\matrixYY) 
\equiv \left( \Ker(\matrixYY^{\Transp}) \right)^{\Ortho}.
\label{eqn:YYvec_OrthoCompl}
\end{equation}

\noindent Let us consider $\NUMrowsAA \coloneqq \NUMlam \, \NUM{\p} - \Rank(\matrixYY) = \Dim \left( \Ker(\matrixYY^{\Transp}) \right)$ vectors $\VEC{\WW}_1,\dots,\VEC{\WW}_{\NUMrowsAA} \in \mathbb{R}^{\NUMlam \, \NUM{\p}}$ to form a basis of the kernel of $\matrixYY^{\Transp}$ and the corresponding matrix

\begin{equation}
\WWmatrix \coloneqq
\begin{pmatrix}
\VEC{\WW}_1 & \VEC{\WW}_2 & \cdots & \VEC{\WW}_{\NUMrowsAA}
\end{pmatrix}
\in \mathbb{R}^{(\NUMlam \, \NUM{\p}) \times \NUMrowsAA}.
\label{eqn:WWmat_def}
\end{equation}

\noindent For the matrix $\WWmatrix$ \eqref{eqn:WWmat_def}, the statement \eqref{eqn:YYvec_OrthoCompl} is equivalent to

\begin{equation}
\WWmatrix^{\Transp} \, \vectorRHSyy = \VEC{0},
\end{equation}

\noindent or, given $\vectorRHSyy$ as in \eqref{eqn:matYY_vecYY_def} and $\WW_{i,j}$ as the entry of row $i$ and column $j$ of $\WWmatrix^{\Transp} \in \mathbb{R}^{\NUMrowsAA \times (\NUMlam \, \NUM{\p})}$, to 

\begin{equation}
\begin{pmatrix}
\WW_{1,1} & \WW_{1,2} & \cdots & \WW_{1,\NUMlam \NUM{\p}} \\
\WW_{2,1} & \WW_{2,2} & \cdots & \WW_{2,\NUMlam \NUM{\p}} \\
\vdots & \vdots & \ddots & \vdots \\
\WW_{\NUMrowsAA,1} & \WW_{\NUMrowsAA,2} & \cdots & \WW_{\NUMrowsAA,\NUMlam \NUM{\p}} 
\end{pmatrix}
\begin{pmatrix} 	 
\VEC{\AAcoeff} - \aavec \\
\vectorRHStt
\end{pmatrix}
= \VEC{0},
\label{eqn:sys_WT_AA_omega_gamma}
\end{equation}

\noindent where $\VEC{\AAcoeff} \in \mathbb{R}^{\NUMlam \, \NUMPI}$ \eqref{eqn:AAdef} is the vector of exponents in \eqref{eqn:lambdas_fun_PI}, while $\aavec \in \mathbb{R}^{\NUMlam \, \NUMPI}$ and $\vectorRHStt \in \mathbb{R}^{\NUMlam \, (\NUM{\p}-\NUMPI)}$ are explicitly known, as specified by \eqref{eqn:aamat_aavec_def} and \eqref{eqn:TTtau_def}, respectively. By using $\NUMPI = \NUM{\p} - \Rank(\MM)$ \changered{from \autoref{th:PI_quantities},} we remark that $\NUMcolsAA \coloneqq \NUMlam \, \NUMPI = \NUMlam \, ( \NUM{\p} - \Rank(\MM) ) < \NUMlam \, \NUM{\p}$, being $\MM$ \eqref{eqn:matrixMM} \changered{defined by \autoref{th:DIMLESS_quant}.} The latter inequality follows from \cref{prop:phys_param_dimen} guaranteeing that $\MM$ \eqref{eqn:matrixMM} cannot be a matrix whose entries are all equal to zero, i.e., $\Rank(\MM)\ge1$. Then, it is possible to write \eqref{eqn:sys_WT_AA_omega_gamma} in terms of $\VEC{\AAcoeff} \in \mathbb{R}^{\NUMcolsAA}$ as the matrix equation $\matrixSysAA \, \VEC{\AAcoeff} = \lhsSysAA$ \eqref{eqn:constr_SysAA}, where the matrix $\matrixSysAA$ and the vector $\lhsSysAA$ are 

\begin{equation}
\matrixSysAA =
\begin{pmatrix}
\WW_{1,1} & \WW_{1,2} & \cdots & \WW_{1,\NUMlam \NUMPI} \\
\WW_{2,1} & \WW_{2,2} & \cdots & \WW_{2,\NUMlam \NUMPI} \\
\vdots & \vdots & \ddots & \vdots \\
\WW_{\NUMrowsAA,1} & \WW_{\NUMrowsAA,2} & \cdots & \WW_{\NUMrowsAA,\NUMlam \NUMPI} 
\end{pmatrix}
\in \mathbb{R}^{\NUMrowsAA \times (\NUMlam \, \NUMPI)},
\label{eqn:SysAA_MatrCoeffdef}
\end{equation}

\begin{equation}
\lhsSysAA =
\begin{pmatrix}
\WW_{1,1} & \WW_{1,2} & \cdots & \WW_{1,\NUMlam \NUMPI} \\
\WW_{2,1} & \WW_{2,2} & \cdots & \WW_{2,\NUMlam \NUMPI} \\
\vdots & \vdots & \ddots & \vdots \\
\WW_{\NUMrowsAA,1} & \WW_{\NUMrowsAA,2} & \cdots & \WW_{\NUMrowsAA,\NUMlam \NUMPI} 
\end{pmatrix}
\aavec -
\begin{pmatrix}
\WW_{1,\NUMlam \NUMPI + 1} & \WW_{1,\NUMlam \NUMPI + 2} & \cdots & \WW_{1,\NUMlam \NUM{\p}} \\
\WW_{2,\NUMlam \NUMPI + 1} & \WW_{2,\NUMlam \NUMPI + 2} & \cdots & \WW_{2,\NUMlam \NUM{\p}} \\
\vdots & \vdots & \ddots & \vdots \\
\WW_{\NUMrowsAA,\NUMlam \NUMPI + 1} & \WW_{\NUMrowsAA,\NUMlam \NUMPI + 2} & \cdots & \WW_{\NUMrowsAA,\NUMlam \NUM{\p}} 
\end{pmatrix}
\vectorRHStt
\in \mathbb{R}^{\NUMrowsAA},
\label{eqn:SysAA_RHSdef}
\end{equation}

\noindent with $\WW_{i,j}$ being the explicitly known entry of row $i$ and column $j$ of the transpose of the matrix $\WWmatrix$ \eqref{eqn:WWmat_def}, while $\aavec \in \mathbb{R}^{\NUMlam \, \NUMPI}$ and $\vectorRHStt \in \mathbb{R}^{\NUMlam \, (\NUM{\p}-\NUMPI)}$ being explicitly given by \eqref{eqn:aamat_aavec_def} and \eqref{eqn:TTtau_def}, respectively. 

\end{proof}

\renewcommand{\thesection}{F}		
\section{Proof of \autoref{th:OSC_formula}}
\label{sec:proofThI}

\ThI*

\begin{proof}

	By employing the notation \eqref{eqn:notation_dimensional_quant} in \eqref{eqn:lam_coeff_I} (\cref{prop:lam_coeff_I}), one has
	
\begin{equation}
\lam_i = \Dimensional{\Klam}_i \,
\Dimensional{\thcc}_1^{\CC_{i,1}} \, 
\Dimensional{\thcc}_2^{\CC_{i,2}} \, 
\cdots \, 
\Dimensional{\thcc}_{\NUM{\var}}^{\CC_{i,\NUM{\var}}}
=
\Klam_i \UNITS{\Dimensional{\Klam}_i} \, 
\thcc_1^{\CC_{i,1}} \UNITS{\Dimensional{\thcc}_1}^{\CC_{i,1}} \,
\thcc_2^{\CC_{i,2}} \UNITS{\Dimensional{\thcc}_2}^{\CC_{i,2}} \,
\cdots \,
\thcc_{\NUM{\var}}^{\CC_{i,{\NUM{\var}}}} \UNITS{\Dimensional{\thcc}_{\NUM{\var}}}^{\CC_{i,{\NUM{\var}}}},
\quad \forall i=1,\dots,\NUMlam.
\label{eqn:lam_units_TH1}
\end{equation}
	
\noindent As each of $\lam_1,\dots,\lam_{\NUMlam}$ is dimensionless, i.e.,

\begin{equation}
\UNITS{\lam_i} =
\UNITS{\Dimensional{\Klam}_i} \, 
\UNITS{\Dimensional{\thcc}_1}^{\CC_{i,1}} \,
\UNITS{\Dimensional{\thcc}_2}^{\CC_{i,2}} \,
\cdots \,
\UNITS{\Dimensional{\thcc}_{\NUM{\var}}}^{\CC_{i,{\NUM{\var}}}}
= 1, 
\quad \forall i=1,\dots,\NUMlam,
\end{equation}

\noindent it follows from \eqref{eqn:lam_coeff_I} and \eqref{eqn:lam_units_TH1} that

\begin{equation}
\lam_i = \Klam_i \, 
\thcc_1^{\CC_{i,1}} \,
\thcc_2^{\CC_{i,2}} \,
\cdots \,
\thcc_{\NUM{\var}}^{\CC_{i,{\NUM{\var}}}}
\in (0,\infty),
\quad \forall i=1,\dots,\NUMlam.
\label{eqn:lam_numval_thcc}
\end{equation}
	
\noindent Then, \eqref{eqn:lam_numval_thcc} allows writing the function $\costOS_{\UNITscale}$ as

\begin{equation}
\costOS_{\UNITscale}(\VEC{\Rthcc}) =
\sum_{i \in \UNITscale}
\left( \log_{10}(\Klam_i) + \sum_{k=1}^{\NUM{\var}} \, \CC_{i,k}  \, \Rthcc_k  \right)^2,
\quad
\forall \VEC{\Rthcc} \coloneqq ( \Rthcc_1,\dots,\Rthcc_{\NUM{\var}} ) \in \mathbb{R}^{\NUM{\var}},
\label{eqn:costOSC_fun_rho}
\end{equation}
	
\noindent where $\Rthcc_k \coloneqq \log_{10}(\thcc_k)$ for any $k=1,\dots,\NUM{\var}$. The twice-differentiable function $\costOS_{\UNITscale}(\VEC{\Rthcc})$ \eqref{eqn:costOSC_fun_rho} is convex because its Hessian matrix is positive semi-definite $\forall \VEC{\Rthcc} \in \mathbb{R}^{\NUM{\var}}$. Then, as guaranteed by Theorem 2.5 in \cite{Nocedal_book_2006}, the global minimizers of $\costOS_{\UNITscale}$ can be found by imposing $\Partial_{\Rthcc_j} \costOS_{\UNITscale} = 0$, $\forall j=1,\dots,\NUM{\var}$, \changegreen{yielding} the linear system \eqref{eqn:lin_sys_rho} of $\NUM{\var}$ equations with $\NUM{\var}$ unknowns $\Rthcc_1,\dots,\Rthcc_{\NUM{\var}} \in \mathbb{R}$.

\end{proof}

\renewcommand{\thesection}{G}		
\section{Proof of \autoref{th:SysToSolveOSAA}}
\label{sec:proofThVII}

\ThVII*

\begin{proof}

	The twice continuously differentiable function $\costOS_{\UNITscale}=\costOS_{\UNITscale}(\VEC{\AAcoeff})$ in \eqref{eqn:CostOSC_def} is convex as its Hessian matrix is diagonal with non-negative entries $\forall \VEC{\AAcoeff} \in \mathbb{R}^{\NUMcolsAA}$. Suppose one also aims to satisfy a constraint of the form \eqref{eqn:Constr_SysOSAA}. It follows that finding any minimum of $\costOS_{\UNITscale}=\costOS_{\UNITscale}(\VEC{\AAcoeff})$ reads as a convex optimization problem with equality constraints, as considered in Section 10.1 of \cite{boyd_vandenberghe_2004}. 

	Then, as guaranteed in Section 10.1 of \cite{boyd_vandenberghe_2004} by assuming \eqref{eqn:HP_OS_AA_analytic_form}, the vector $\VEC{\AAcoeff} \in \mathbb{R}^{\NUMcolsAA}$ is optimal for the problem under consideration if and only if there is a vector $\helpVecSysAA \in \mathbb{R}^{\helpNrSysAA}$ such that \eqref{eqn:SysToSolveOSAA} holds. Any found local minimizer is automatically a global minimizer, because the considered problem is convex, as discussed in Section 4.2.2 of \cite{boyd_vandenberghe_2004}.

\end{proof}

\addcontentsline{toc}{section}{References}

\bibliography{refs.bib}

\end{document}